\documentclass[12pt]{article}
\usepackage{setspace}
\usepackage[letterpaper,margin=1.25in]{geometry}
\usepackage[backref,colorlinks,citecolor=blue,bookmarks=true]{hyperref}
\usepackage{enumitem}
\usepackage{amsthm,amsmath, amssymb}
\usepackage{bbm}
\usepackage{pgfplots}
\pgfplotsset{compat=1.4}
\usetikzlibrary{patterns}
\usepackage{mdframed}
\usepackage[compact]{titlesec}
\usepackage{microtype}

\definecolor{lightgray}{rgb}{0.95,0.95,0.95}
\mdfdefinestyle{boxed}{
 skipabove=6pt,
 skipbelow=6pt,
innerleftmargin=5,innerrightmargin=5,
backgroundcolor=lightgray,
  topline=false,
  rightline=false,
  leftline=false,
  bottomline=false}

\newtheorem{definition}{Definition}
\newtheorem{example}{Example}
\newtheorem{lemma}{Lemma}
\newtheorem{theorem}{Theorem}
\newtheorem{corollary}{Corollary}
\newtheorem{claim}{Claim}

\newtheorem{proposition}{Proposition}
\theoremstyle{remark}
\newtheorem{remark}{Remark}
\newenvironment{prevproof}[2]{\noindent {\sc {Proof of {#1}~\ref{#2}:}}}{\hfill $\Box$\vskip \belowdisplayskip}

\newcommand{\Lone}[1]{\ensuremath{\|#1\|_1}}
\newcommand{\Linf}[1]{\ensuremath{\|#1\|_{\infty}}}

\newcommand{\cM}{{\cal M}}

\newcommand{\cP}{{\cal P}}

\newcommand{\cT}{{\cal T}}

\begin{document}
\onehalfspacing

\title{Strong Duality for a Multiple-Good Monopolist}

\author {
Constantinos Daskalakis\thanks{Supported by a Sloan Foundation Fellowship, a Microsoft Research Faculty Fellowship, and NSF Awards CCF-0953960 (CAREER) and  CCF-1101491.}\\
EECS, MIT \\
\and
Alan Deckelbaum\thanks{Supported by Fannie and John Hertz Foundation Daniel Stroock Fellowship and NSF Award  CCF-1101491.}\\
Math, MIT\\
\and
Christos Tzamos\thanks{Supported by NSF Award CCF-1101491 and a Simons Award for Graduate Students in TCS.}\\
EECS, MIT\\
}
\maketitle

\begin{abstract}

We characterize optimal mechanisms for the multiple-good monopoly problem and provide a framework to find them. We show that a mechanism is optimal if and only if a measure $\mu$ derived from the buyer's type distribution satisfies certain stochastic dominance conditions. This measure expresses the marginal change in the seller's revenue under marginal changes in the rent paid to subsets of buyer types. As a corollary, we characterize the optimality of grand-bundling mechanisms, strengthening several results in the literature, where only sufficient optimality conditions have been derived. As an application, we show that the optimal mechanism for $n$ independent uniform items each supported on $[c,c+1]$ is a grand-bundling mechanism, as long as $c$ is sufficiently large, extending Pavlov's result for $2$ items~\cite{Pavlov11}. At the same time, our characterization also implies that, for all $c$ and for all sufficiently large $n$, the optimal mechanism for $n$ independent uniform items supported on $[c,c+1]$ {\em is not} a grand bundling mechanism.

\end{abstract}


\newpage 

\section{Introduction} \label{sec:introduction}
We study the problem of revenue maximization for a multiple-good monopolist. Given $n$ heterogenous goods and a probability distribution $f$ over $\mathbb{R}^n_{\ge 0}$, we wish to design a mechanism that optimizes the monopolist's expected revenue against an additive (linear) buyer whose values for the goods are distributed according to $f$. 

The single-good version of this problem---namely, $n=1$---is well-understood,~going back~to~\cite{riley1981optimal,Myerson81,maskin1984monopoly,riley1983optimal},~where~it is shown that a take-it-or-leave-it offer of the good at some price is optimal, and the optimal price can be easily calculated from $f$. 

For general $n$, it has been known that the optimal mechanism may exhibit much richer structure. Even when the item values are independent, the mechanism may benefit from selling bundles of items or even lotteries over bundles of items~\cite{McAfeeMW89,BakosB99,Thanassoulis04,ManelliV06}. Moreover, no general framework to approach this problem has been proposed in the literature, making it dauntingly difficult both to identify optimal solutions and to certify the optimality of those solutions. As a consequence, seemingly simple special cases (even $n=2$) remain poorly understood, despite much research for a few decades. See, e.g.,~\cite{rochet2003economics} for a comprehensive survey of work spanning our problem, as well as~\cite{ManelliV07} and~\cite{figalli2011multidimensional} for additional references.

We propose a novel framework for revenue maximization based on duality theory. We identify a minimization problem that is dual to revenue maximization and prove that  the optimal values of these problems are always equal. Our framework allows us to identify optimal mechanisms in general settings, and certify their optimality by providing a complementary solution to the dual problem, namely finding a solution to the dual whose objective value equals the mechanism's revenue. Our framework is applicable to arbitrary settings of $n$ and $f$, with mild assumptions such as differentiability. In particular, we strengthen prior work~\cite{ManelliV06,DaskalakisDT13,GiannakopoulosK14}, which identified optimal mechanisms in special cases. We exhibit the practicality of our framework by solving several examples. Importantly, we can leverage our duality theorem to characterize optimal multi-item mechanisms. From a technical standpoint we provide new analytical methodology for multi-dimensional mechanism design by providing extensions to Monge-Kantorovich duality for optimal transportation.
We proceed to discuss our contributions in detail, providing a roadmap to the paper, and conclude this section with a discussion of related work.

\paragraph{Strong Duality.} Our first main result (presented as Theorem~\ref{strongduality}) formulates a dual problem to the optimal mechanism design problem, and establishes strong duality between the two problems. That is, we show that the optimal values of the two optimization problems are identical. Our approach for developing this dual problem is outlined below.

 We start by formulating optimal mechanism design as a maximization problem over convex, non-decreasing and $1$-Lipschitz continuous functions $u$, representing the utility of the buyer as a function of her type, as in~\cite{Rochet1987}. The objective function of this maximization problem can be written as the expectation of $u$ with respect to a signed measure $\mu$ over the type space of the buyer.
Measure $\mu$ is easily derived from the buyer's type distribution $f$ (see Equation~\eqref{transformed}) and expresses the marginal change in the seller's revenue under marginal changes in the rent paid to subsets of buyer types. Our formulation is summarized in Theorem~\ref{setupclaim}, while Section~\ref{examplesetup} illustrates our formulation in the basic setting of independent uniform items.

 In Theorem~\ref{strongduality}, we formulate a dual in the form of an optimal transportation problem, and establish strong duality between the two problems. Roughly speaking, our dual formulation is given the signed measure $\mu$ (from Theorem~\ref{setupclaim}) and solves the following minimization problem: (i)~first, it is allowed to choose any measure $\mu'$ that stochastically dominates $\mu$ with respect to convex increasing functions; (ii) second, it is supposed to find a coupling of the positive part $\mu'_+$ of $\mu'$ with its negative part $\mu'_-$ i.e. find a transportation from $\mu'_+$ to $\mu'_-$;
(iii) if a unit of mass of $\mu'_+$ at $x$ is transported to a unit of mass of $\mu'_-$ at $y$, we are charged $\|x-y\|_1$. The goal is to minimize the cost of the coupling with respect to the decisions in (i) and (ii). 

 While our dual formulation takes a simple form, establishing strong duality is quite technical. At a high level, our proof follows the proof of Monge-Kantorovich duality in~\cite{Villani}, making use of the Fenchel-Rockafellar duality theorem, but the technical aspects of the proof are different due to the convexity constraint on feasible utility functions. The proof is presented in the online appendix, but it is not necessary to understand the other results in this paper. We note that our formulation from Theorem~\ref{setupclaim} defines a convex optimization problem. One would hope then that infinite-dimensional linear programming techniques~\cite{luenberger1968optimization, anderson1987linear} can be leveraged to establish the existence of a strong dual. We are not aware of such an approach, and expect that such formulations will fail to establish existence of interior points in the primal feasible set, which is necessary for strong duality.

 As already emphasized earlier, our identification of a strong dual implies that the optimal mechanism admits a certificate of optimality, in the form of a dual witness, for all settings of $n$ and $f$. Hence, our duality framework can play the role of first-order conditions certifying the optimality of single-dimensional mechanisms. Where optimality of single-dimensional mechanisms can be certified by checking virtual welfare maximization, optimality of multi-dimensional mechanisms is always certifiable by providing dual solutions whose value matches the revenue of the mechanism, and such dual solutions take a simple form: they are transportation maps between measures. 

 Using our framework, we can provide shorter proofs of optimality of known mechanisms. As an illustrating example, we show in Section~\ref{sec:uniform 0-1} how to use our framework to establish the optimality of the mechanism for two i.i.d. uniform $[0,1]$ items proposed by~\cite{ManelliV06}. Then in Section~\ref{sec: uniform non 0-1}, we provide a simple illustration of the power of our framework, obtaining the optimal mechanism for two independent uniform $[4,16]$ and uniform $[4,7]$ items, a setting where the results of~\cite{ManelliV06,Pavlov11,DaskalakisDT13,GiannakopoulosK14} fail to apply. The optimal mechanism has the somewhat unusual structure shown in the diagram in Section~\ref{examplesection}, where types in $Z$ are allocated nothing (and pay nothing),  types in $W$ are allocated the grand bundle (at price $12$), while types in $Y$ are allocated item $2$ and get item $1$ with probability $50\%$ (at price $8$).

\paragraph{Characterization of Optimal Mechanisms.} Substantial effort in the literature has been devoted to studying optimality of mechanisms with a simple structure such as pricing mechanisms; see, e.g., \cite{ManelliV06} and~\cite{DaskalakisDT13} for sufficient conditions under which mechanisms that only price the grand bundle of all items are optimal. Our second main result (presented as Theorem~\ref{bundlingtheorem}) obtains {\em necessary and sufficient} conditions characterizing the optimality of arbitrary mechanisms with a finite menu size. We proceed to describe our characterization result in more detail.

 Suppose that we are given a feasible mechanism $\cal M$ whose set of possible allocations is finite. We can then partition the type set into finitely many subsets (called regions) ${\cal R}_1,\ldots,{\cal R}_k$ of types who enjoy the same price and allocation. The question is this: {\em for what type distributions is ${\cal M}$ optimal?} Theorem~\ref{bundlingtheorem} answers this question with a sharp characterization result: $\cal M$ is optimal {\em if and only if} the measure $\mu$ (derived from the type distribution as described above) satisfies $k$ stochastic dominance conditions, one per region in the afore-defined partition. The type of stochastic dominance that $\mu$ restricted to region ${\cal R}_i$ ought to satisfy depends on the allocation to types from ${\cal R}_i$, namely which set of items are allocated with probability $1$, $0$, or non-$0$/$1$.

Theorem~\ref{bundlingtheorem} is important in that it reduces checking the optimality of mechanisms to checking standard stochastic dominance conditions between measures derived from the type distribution $f$, which is a concrete and easier task than arguing optimality against all possible mechanisms. 

 Theorem~\ref{bundlingtheorem} is a corollary of our strong duality framework (Theorem~\ref{strongduality}), but requires a sequence of technical results. One direction of our characterization result requires turning the stochastic dominance conditions  into dual solutions that can be plugged into Theorem~\ref{strongduality} to establish the optimality of a given mechanism. The other direction requires showing that a dual solution certifying the optimality of a given mechanism also implies that the stochastic dominance conditions of Theorem~\ref{bundlingtheorem} must hold. 

 A particularly simple special case of our characterization result pertains to the optimality of the grand-bundling mechanism. See Theorem~\ref{grandbundlingtheorem}. We show that the mechanism offering the grand bundle at price $p$ is optimal {\em if and only if}  measure $\mu$ satisfies a pair of stochastic dominance conditions. 
In particular, if $Z$ are the types who cannot afford the grand bundle and $W$ the types who can, then offering the grand bundle for $p$ is optimal if and only if the following conditions hold:
\begin{itemize}[label={-},leftmargin=10pt]
\item $\mu_-\vline_Z$, the negative part of $\mu$ restricted to $Z$, stochastically dominates $\mu_+\vline_Z$, the positive part of $\mu$ restricted to $Z$, with respect to all convex increasing functions;
\item $\mu_+\vline_W$  stochastically dominates $\mu_-\vline_W$ with respect to all concave increasing functions.
\end{itemize}
Already our characterization of grand-bundling optimality settles a long line of research which only obtained sufficient conditions for the optimality of grand-bundling.

 In turn, we illustrate the power of our characterization of grand-bundling optimality with Theorems~\ref{nuniform} and~\ref{nuniform-notbundling}, two results that are interesting on their own right. Theorem~\ref{nuniform} generalizes the corresponding result of~\cite{Pavlov11} from two to an arbitrary number of items. We show that, for any number of items $n$, there exists a large enough $c$ such that the optimal mechanism for $n$ i.i.d. uniform $[c,c+1]$ items is a grand-bundling mechanism. While maybe an intuitive claim, we do not see a direct way of proving it. Instead, we utilize Theorem~\ref{grandbundlingtheorem} and construct intricate couplings establishing the stochastic dominance conditions required by the theorem. In view of Theorem~\ref{nuniform}, our companion theorem, Theorem~\ref{nuniform-notbundling}, seems even more surprising. We show that in the same setting of $n$ i.i.d. uniform $[c,c+1]$ items, for any fixed $c$ it holds that, for all sufficiently large $n$, the optimal mechanism {\em is not} (!) a grand-bundling mechanism. See Section~\ref{bundlingsection} for the proofs of these results.

\paragraph{Related Work.} 

There is a rich literature on multi-item mechanism design pertaining to the multiple good monopoly problem that we consider here. We refer the reader to the surveys~\cite{rochet2003economics,ManelliV07,figalli2011multidimensional} for a detailed description, focusing on the work closest to ours. 

Much work has focused on obtaining sufficient conditions for optimality of mechanisms. Hart and Nisan~\cite{hart2014good}, Menicucci et al~\cite{menicucci2015optimality} and Haghpanah and Hartline~\cite{haghpanah2015reverse} provide sufficient conditions for the grand-bundling mechanism to be optimal. Manelli and Vincent~\cite{ManelliV06} provide conditions for the optimality of more complex deterministic mechanisms and, similarly, \cite{DaskalakisDT13,GiannakopoulosK14} provide sufficient conditions for the optimality of general (possibly randomized) mechanisms. Finally, Haghpanah and Hartline~\cite{haghpanah2015reverse} provide an approach for reverse engineering sufficient conditions for a simple mechanism to be optimal. These works on sufficient conditions apply to limited settings of $n$ and $f$. They typically proceed by relaxing some of the truthfulness constraints and are therefore only applicable when the relaxed constraints are not binding at the optimum.

In addition to sufficient conditions, a lot of work has focused on characterizing properties of optimal mechanisms. Armstrong~\cite{armstrong1996multiproduct} has shown that optimal mechanisms always exclude a fraction of buyer types of low value from the mechanism. Thanassoulis~\cite{Thanassoulis04}, Briest et al~\cite{briest2010pricing} and Hart and Nisan~\cite{HartN13} show that randomization is necessary for optimal revenue extraction. In turn, Manelli and Vincent~\cite{ManelliV07} have shown that there exist type distributions for which optimal mechanisms are arbitrarily complex. Hart and Reny provide an interesting example where a product type distribution over two items stochastically dominates another, yet the optimal revenue from the weaker distribution is higher~\cite{hart2015maximal}. Finally, some literature~\cite{armstrong1999price, hart2014good,babaioff2014simple,li2013revenue,cai2013simple} has focused on the revenue guarantees of simple mechanisms, e.g. bundling all items together or selling them separately.

Rochet and Chon\'e~\cite{RochetChone} study a closely related setting, providing a characterization of the optimal mechanism for the multiple good monopoly problem where the monopolist has a (strictly) convex  cost for producing copies of the goods. With strictly convex production costs, optimal mechanism design becomes a strictly concave maximization problem, which allows the use of first-order conditions to characterize optimal mechanisms. Our problem can be viewed as having a production cost that is~$0$ for selling at most one unit of each good and infinity otherwise. While still convex, our production function is not strictly convex and is discontinuous, making first-order conditions less useful for characterizing optimal mechanisms. This motivates the use of duality theory in our setting. From a technical standpoint, optimal mechanism design necessitates the development of new tools in optimal transport theory~\cite{Villani}, extending Monge-Kantorovich duality to accommodate convexity constraints in the dual of the transportation problem. In our setting, the dual of the transportation problem corresponds to the mechanism design problem and these constraints correspond to the requirement that the utility function of the buyer be convex, which is intimately related to the truthfulness of the mechanism~\cite{Rochet1987}. In turn, accommodating the convexity constraints in the mechanism design problem requires the introduction of mean-preserving spreads of measures in its transportation dual, resembling the  ``multi-dimensional sweeping'' of Rochet and Chon\'e. 

Ultimately, our work relies on and develops further a fundamental connection of optimal transportation to designing optimal mechanisms. See Ekeland's notes on Optimal Transportation~\cite{ekeland2010notes} for more connections to mechanism design.

\section{Revenue Maximization as Optimization Program}\label{setupsection}

\subsection{Setting up the Optimization Program}
Our goal is to find the revenue-optimal mechanism $\cM$ for selling $n$ goods to a single additive buyer. An additive buyer has a \emph{type} $x$ specifying his value for each good. The type $x$ is an element of a \emph{type space} $X = \prod_{i=1}^n [x^{\textrm{low}}_i,x^{\textrm{high}}_i]$, where $x^{\textrm{low}}_i,x^{\textrm{high}}_i$ are non-negative real numbers. While the buyer knows his type with certainty, the mechanism designer only knows the probability distribution over $X$ from which $x$ is drawn. We assume that the distribution has a density $f: X \rightarrow \mathbb{R}$ that is continuous and differentiable with bounded derivatives.

 Without loss of generality, by the revelation principle, we consider direct mechanisms. A (direct) mechanism consists of two functions: {\tt (i)} an \emph{allocation function} $\cP : X \rightarrow [0,1]^n$ specifying the probabilities, for each possible type declaration of the buyer, that the buyer will be allocated each good, and {\tt (ii)} a \emph{price function} $\cT : X \rightarrow \mathbb{R}$ specifying, for each declared type of the buyer, the price that he is charged. When an additive buyer of type $x$ declares himself to be of type $x' \in X$, he receives net expected utility $x \cdot \cP(x') - \cT(x')$.

We restrict our attention to mechanisms that are \emph{incentive compatible}, meaning that the buyer must have adequate incentives to reveal his values for the items truthfully, and \emph{individually rational}, meaning that the buyer has an incentive to participate in the mechanism. 

\begin{definition}
Mechanism $\cM = (\cP, \cT)$ over type space $X$ is \emph{incentive compatible (IC)} if and only if  $x\cdot \cP(x) - \cT(x) \geq x \cdot \cP(x') - \cT(x')$ for all $x,x' \in X$.
\end{definition}

\begin{definition}
Mechanism $\cM= (\cP, \cT)$ over type space $X$ is \emph{individually rational (IR)} if and only if $x\cdot \cP(x) - \cT(x) \geq 0$ for all $x \in X$.
\end{definition}

When a buyer truthfully reports his type to a mechanism $\cM = (\cP,\cT)$ (over type space $X$), we denote by $u: X \rightarrow \mathbb{R}$ the function that maps the buyer's valuation to the utility he receives by $\cM$. It follows by the definitions of $\cP$ and $\cT$ that $u(x) = x\cdot \cP(x) - \cT(x)$. It is well-known (see \cite{Rochet1987}, \cite{RochetChone}, and \cite{ManelliV06}), that an IC and IR mechanism has a convex, nonnegative, nondecreasing, and 1-Lipschitz utility function with respect to the $\ell_1$ norm and that any utility function satisfying these properties is the utility function of an IC and IR mechanism with $\cP(x) = \nabla u(x)$ and $\cT(x) =  \cP(x) \cdot x - u(x)$.\footnote{On the measure-0 set on which $\nabla u$ is not defined, we can use an analogous expression for $\cP$ by choosing appropriate values of $\nabla u$ from the subgradient of $u$.}

We clarify that a function $u$ is 1-Lipschitz with respect to the $\ell_1$ norm if  $u(x) - u(y) \leq \Lone{x-y}$ for all $x,y \in X$. This is essentially equivalent to all partial derivatives having magnitude at most 1 in each dimension.

We will formulate the mechanism design problem as an optimization problem over feasible utility functions $u$. We first define the notation:

\begin{itemize}[label={-},leftmargin=10pt]
	\item $\mathcal{U}(X)$ is the set of all continuous, non-decreasing, and convex functions $u: X \rightarrow \mathbb{R}$.
	\item $\mathcal{L}_1(X)$ is the set of all 1-Lipschitz with respect to the $\ell_1$ norm functions $u: X \rightarrow \mathbb{R}$.
\end{itemize} 

\noindent In this notation, a mechanism $\cM$ is IC and IR if and only if its utility function $u$ satisfies $u \geq 0$ and $u \in \mathcal{U}(X) \cap \mathcal{L}_1(X)$. It follows that the optimal mechanism design problem can be viewed as an optimization problem: $$\sup_{\substack{u \in \mathcal{U}(X) \cap \mathcal{L}_1(X)\\ u \ge 0}}\int_{X} [ \nabla u(x) \cdot x - u(x) ] f(x) dx.$$

Notice that for any utility $u$ defining an IC and IR mechanism, the function $\tilde{u}(x) =  u(x)-u(x^{\textrm{low}})$ also defines a valid IC and IR mechanism since $\tilde{u} \in \mathcal{U}(X) \cap \mathcal{L}_1(X)$ and $\tilde{u} \ge 0$. Moreover, $\tilde{u}$ achieves at least as much revenue as $u$, and thus it suffices in the above program to look only at feasible $u$ with $u(x^{\textrm{low}})=0$.

We claim that we can therefore remove the constraint $u\ge0$ and equivalently focus on solving
\begin{align}
\sup_{u \in \mathcal{U}(X) \cap \mathcal{L}_1(X)}\int_{X} [ \nabla u(x) \cdot x - ( u(x) - u(x^{\textrm{low}}) ) ] f(x) dx. \label{eq:revenue opt1}
\end{align}
Indeed, this objective function agrees with the prior one whenever $u(x^{\textrm{low}})=0$. Furthermore, for any $u \in \mathcal{U}(X) \cap \mathcal{L}_1(X)$, the function  $\tilde{u}(x) =  u(x)-u(x^{\textrm{low}})$ is nonnegative and achieves the same objective value. Applying the divergence theorem as in \cite{ManelliV06} we may rewrite the expression for expected revenue in~\eqref{eq:revenue opt1} as follows:
\begin{align}
&\int_{X} [ \nabla u(x) \cdot x - ( u(x) - u(x^{\textrm{low}}) ) ] f(x) dx = \notag\\
&~~~~~\int_{\partial X} u(x)f(x)(x \cdot \hat{n})dx -\int_X u(x) ( \nabla f(x) \cdot x + (n+1)f(x))dx + u(x^{\textrm{low}}) \label{eq:revenue expression 1}
\end{align}
where $\hat{n}$ denotes the outer unit normal field to the boundary $\partial X$. To simplify notation we make the following definition.
\begin{definition}[Transformed measure]
  \label{def:transformed measure}
The \emph{transformed measure} of $f$ is the (signed) measure $\mu$ (supported within $X$) given by the property that
 \begin{align}\label{transformed}
\mu(A) \triangleq \hspace{-5pt} \int_{\partial X} \hspace{-5pt} \mathbb{I}_A(x) f(x)(x \cdot \hat{n})dx - \hspace{-5pt} \int_X \mathbb{I}_A(x) ( \nabla f(x) \cdot x + (n+1)f(x))dx + \mathbb{I}_A(x^{\textrm{low}})
\end{align}
for all measurable sets $A$.\footnote{It follows from boundedness of $f$'s partial derivatives that $\mu$ is a Radon measure. Throughout this paper, all ``measures'' we use will be Radon measures.}
\end{definition}

\begin{mdframed}[style=boxed]
{\bf Interpretation of Transformed Measure:} Given~\eqref{eq:revenue expression 1} and~\eqref{transformed}, the revenue of the seller in Formulation~\eqref{eq:revenue opt1} can be written as $\int_{X} u d\mu$, which is a linear functional of $u$ with respect to the measure $\mu$. Hence, we will maintain the following intuition of what measure $\mu$ represents:\\

\begin{minipage}{13cm}
``Measure $\mu$ quantifies the marginal change in revenue with respect to marginal changes in the rent paid to subsets of buyer types.'' 
\end{minipage}

\bigskip \noindent Moreover, our measure satisfies that $\mu(X) = \int_{X} 1 d\mu = 0.$ Indeed, if we substitute $u(x)=1$ to the left hand side of~\eqref{eq:revenue expression 1}, we have that
$$\int_{X} [ \nabla u(x) \cdot x - ( u(x) - u(x^{\textrm{low}}) ) ] f(x) dx = 0.$$
Furthermore, we have $|\mu|(X)<\infty$, since $f$, $\nabla f$ and $X$ are bounded.
\end{mdframed}

Summarizing the above derivation, we obtain the following theorem.

\begin{theorem}[Multi-Item Monopoly Problem]\label{setupclaim}
The problem of determining the optimal IC and IR mechanism for a single additive buyer whose values for $n$ goods are distributed according to the joint distribution $f:X \rightarrow \mathbb{R}_{\geq 0}$ is equivalent to solving the optimization problem
\vspace{-15pt}
\begin{align}
\sup_{u \in \mathcal{U}(X) \cap \mathcal{L}_1(X)} \int_X ud\mu \label{eq:multigood monopoly formulation}
\end{align}
where $\mu$ is the transformed measure of $f$ given in (\ref{transformed}). 
\end{theorem}

\subsection{Example}\label{examplesetup}
Consider $n$ independently distributed items, where the value of each item $i$ is drawn uniformly from the bounded interval $[a_i,b_i]$ with $0\leq a_i < b_i < \infty$. The support of the joint distribution is the set $X = \prod_i [a_i,b_i]$.

For notational convenience, define $v \triangleq \prod_i (b_i - a_i)$, the volume of $X$. The joint distribution of the  items is given by the constant density function $f$ taking value $1/v$ throughout $X$.
The transformed measure $\mu$ of $f$ is given by the relation
$$
 \mu(A) = \mathbb{I}_A(a_1,\ldots,a_n) + \frac{1}{v} \int_{\partial X} \mathbb{I}_A(x) (x \cdot \hat{n})dx - \frac{n+1}{v} \int_X \mathbb{I}_A(x)dx
$$
for all measurable sets $A$. Therfore, by Theorem~\ref{setupclaim}, the optimal revenue is equal to $\sup_{u \in \mathcal{U}(X) \cap \mathcal{L}_1(X)} \int_X u d\mu$, where $\mu$ is the sum of:
\begin{itemize}
	\item A point mass of $+1$ at the point $(a_1,\ldots,a_n)$.
	\item A mass of $-(n+1)$ distributed uniformly throughout the region $X$.
	\item A mass of $+\frac{b_i}{b_i-a_i}$ distributed uniformly on each surface $\left\{ x \in \partial X : x_i = b_i\right\}$.
	\item A mass of $-\frac{a_i}{b_i-a_i}$ distributed uniformly on each surface $\left\{ x \in \partial X : x_i = a_i\right\}$.
\end{itemize}

\section{The Strong Mechanism Design Duality Theorem}

Thus far, we have compactly formulated the problem facing the multi-item monopolist as an optimization problem with respect to the buyer's utility function; see Formulation~\eqref{eq:multigood monopoly formulation}. Unfortunately, the problem is infinite dimensional and cannot be solved directly. Moreover, the problem is not strictly convex so we cannot characterize its optimum using first order conditions. This is an important point of departure in comparison with the work of Rochet and Chon\'e~\cite{RochetChone}, where the strict convexity of the cost function allowed first order conditions to drive the characterization. For more discussion see Section~\ref{sec:introduction}.

In the absence of strict convexity, our approach is to use duality theory. We are seeking to identify a minimization problem, called ``the dual problem,'' and which is linked to Formulation~\eqref{eq:multigood monopoly formulation}, henceforth called ``the primal problem,'' as follows:
\begin{enumerate}
\item We want that the value of any solution to the dual problem is larger than the revenue achieved by any solution to the primal problem. If a minimization problem satisfies this property, it is called a ``weak dual problem.''
\item Additionally, we want that the optimum of the dual problem matches the optimum of the primal problem. A minimization problem satisfying this property is called a ``strong dual problem.'' It is clear that a strong dual problem is also a weak dual problem. This type of strong dual problem is what we will identify in Theorem~\ref{strongduality} of this section.
\end{enumerate}

The importance of identifying a strong dual problem is the following. Given a candidate optimal mechanism, we are guaranteed that a  solution to the dual problem with a matching objective value exists if and only if the candidate mechanism is indeed optimal. Therefore, solutions to the dual problem constitute ``certificates of optimality'' for solutions to the primal, and strong duality guarantees that such dual certificates are always possible to find for optimal solutions to the primal. Accordingly, we will be seeking solutions to our dual problem from Theorem~\ref{strongduality} to obtain ``certificates, or witnesses, of optimality'' for candidate optimal mechanisms. By this we mean that we will be seeking solutions to the dual that prove (via duality theory) that a candidate optimal mechanism is indeed optimal. Moreover, these dual solutions take the form of optimal transportation maps between submeasures induced by measure $\mu$ of Definition~\ref{def:transformed measure}. This tight connection between optimal mechanisms (primal solutions) and optimal transportation maps (dual solutions) drives our characterization of optimal mechanisms in Theorem~\ref{bundlingtheorem}, as well as the concrete examples we work out in Sections~\ref{examplesection}, \ref{bundlingexamplesection} and \ref{sec:further examples}. Moreover, by ``reverse-engineering the duality theorem'' we provide a framework for identifying optimal mechanisms in Section~\ref{weakstructural}.

Recent work has applied duality theory to identify optimal mechanisms in the same setting as ours~\cite{ManelliV06,DaskalakisDT13,GiannakopoulosK14}, albeit this work is restricted in that they only provide weak dual problems. These approaches remove constraints related to truthfulness from the primal formulation, and identify weak dual formulations to such relaxed primal formulations. As such, they provide no guarantee that they can identify dual certificates of optimality for optimal mechanisms. Indeed, while these techniques suffice in certain settings (namely when the constraints removed from the primal happen not to be binding at the optimum), there are simple examples where they fail to apply. Section~\ref{sec: uniform non 0-1} provides such a two-item example with uniformly distributed values. In contrast to prior work, we achieve strong duality for the (unrelaxed) primal formulation and our approach is always guaranteed to work.

In this section, we show how to pin down the right dual formulation for the problem and prove strong duality. The proof of the result requires many analytical tools from measure theory. We give a rough sketch of the proof in this section and postpone the more technical details to the online appendix.

\subsection{Measure-Theoretic Preliminaries}

We start with some useful measure-theoretic notation:

\begin{itemize}[label={-},leftmargin=10pt]
	\item $\Gamma(X)$ and $\Gamma_+(X)$ denote the sets of signed and unsigned (Radon) measures on $X$.
	\item Given an unsigned measure $\gamma \in \Gamma_+(X \times X)$, we denote by $\gamma_1, \gamma_2$ the two marginals of $\gamma$, i.e. $\gamma_1(A) = \gamma(A \times X)$ and  $\gamma_2(A) = \gamma(X \times A)$ for all measurable sets $A \subseteq X$.
		\item For a (signed) measure $\mu$ and a measurable $A \subseteq X$, we define the {\em restriction of $\mu$ to $A$}, denoted $\mu|_A$, by the property $\mu|_A(S) = \mu(A \cap S)$ for all measurable $S$.
	\item For a signed measure $\mu$, we will denote by $\mu_+, \mu_- $ the positive and negative parts of $\mu$, respectively. That is, $\mu = \mu_+ - \mu_-$, where $\mu_+$ and $\mu_-$ provide mass to disjoint subsets of $X$.
\end{itemize}

We will also be needing certain stochastic dominance properties, namely first- and second-order stochastic dominance as well as the notion of convex dominance.
\begin{definition}
We say that $\alpha$ \emph{first-order} (respectively {\em second-order}) dominates $\beta$ for $\alpha, \beta \in \Gamma(X)$, denoted $\alpha \succeq_{1} \beta$ (respectively $\alpha \succeq_{2} \beta$), if for all non-decreasing continuous (respectively  non-decreasing concave) functions $u: X \rightarrow \mathbb{R}$,
$\int u d\alpha \geq \int u d\beta.$

Similarly, for vector random variables $A$ and $B$ with values in $X$, we say that $A \succeq_{1} B$ (respectively $A \succeq_{2} B$) if $\mathbb{E}[u(A)] \geq \mathbb{E}[u(B)]$ for all non-decreasing continuous (respectively  non-decreasing concave) functions $u: X \rightarrow \mathbb{R}$.
\end{definition}

\begin{definition}
We say that $\alpha$ \emph{convexly dominates} $\beta$ for $\alpha, \beta \in \Gamma(X)$, denoted $\alpha \succeq_{cvx} \beta$, if for all (non-decreasing, convex) functions $u \in \mathcal{U}(X)$,
$\int u d\alpha \geq \int u d\beta.$

Similarly, for vector random variables $A$ and $B$ with values in $X$, we say that $A \succeq_{cvx} B$ if $\mathbb{E}[u(A)] \geq \mathbb{E}[u(B)]$ for all $u \in \mathcal{U}(X)$.
\end{definition}

\begin{mdframed}[style=boxed]
 {\bf Interpretation of Convex Dominance:} For intuition, a measure $\alpha \succeq_{cvx} \beta$ if we can transform $\beta$ to $\alpha$ by doing the following two operations:
\begin{enumerate}
\item sending (positive) mass to coordinatewise larger points: this makes the integral $\int u d\beta$ larger since $u$ is non-decreasing.
\item spreading (positive) mass so that the mean is preserved:
this makes the integral $\int u d\beta$ larger since $u$ is convex.
\end{enumerate}
The existence of a valid transformation using the above operations is equivalent to convex dominance. This follows by Strassen's theorem presented in the online appendix.
\end{mdframed}

\subsection{Mechanism Design Duality}

The main result of this paper is that the mechanism design problem, formulated as a maximization problem in Theorem~\ref{setupclaim}, has a strong dual problem, as follows:
\begin{theorem}[Strong Duality Theorem]\label{strongduality}
Let $\mu \in \Gamma(X)$ be the transformed measure of the probability density $f$ according to Definition~\ref{def:transformed measure}. Then
\begin{align}\sup_{u \in \mathcal{U}(X) \cap \mathcal{L}_1(X)} \int_X u d\mu = \inf_{\substack{\gamma \in \Gamma_+(X \times X) \\ \gamma_1 - \gamma_2 \succeq_{cvx} \mu  }}\int_{X \times X} \Lone{x-y} d\gamma(x,y)\label{eq:strong duality}
  \end{align}
and both the supremum and infimum are achieved. Moreover, the infimum is achieved for some $\gamma^*$ such that $\gamma_1^*(X) = \gamma_2^*(X) = \mu_+(X)$, $\gamma_1^* \succeq_{cvx} \mu_+$, and $\gamma_2^* \preceq_{cvx} \mu_-$.
\end{theorem}

\begin{mdframed}[style=boxed]
{\bf Interpretation of the Strong Dual Problem:} The dual problem of minimizing $\int \Lone{x-y}d\gamma$ is an optimization problem that can be intuitively thought as a two step process:

\textbf{Step 1}: Transform $\mu$ into a new measure $\mu'$ with $\mu'(X) = 0$ such that $\mu' \succeq_{cvx} \mu$. This step is similar to sweeping as defined in \cite{RochetChone} where they transform the original measure by mean-preserving spreads. However, here we are also allowed to perform positive mass transfers to coordinatewise larger points.

\textbf{Step 2}: Find a joint measure $\gamma \in \Gamma_+(X \times X)$ with $\gamma_1 = \mu'_+, \gamma_2 = \mu'_-$ such that $\int \Lone{x-y}d\gamma(x,y)$ is minimized. This is an optimal mass transportation problem where
  the cost of transporting a unit of mass from a point $x$ to a point $y$ is the $\ell_1$ distance $\Lone{x-y}$, and we are asked for the cheapest method of transforming the positive part of $\mu'$ into the negative part of $\mu'$.
Transportation problems of this form have been studied in the mathematical literature. See \cite{Villani}. 

Overall, our goal in the dual problem is to match the positive part of $\mu$ to the negative part of $\mu$ at a minimum cost where some operations come for free, namely we can choose any $\mu' \succeq_{cvx} \mu$ that is convenient to us, foreseeing that transporting $\mu'_+$ to $\mu_{-}'$ comes at a cost equal to the total $\ell_1$ distance
that mass travels.
\end{mdframed}

We remark that establishing that the right hand side of~\eqref{eq:strong duality} is a weak dual for the left hand side is easy. Proving strong duality is significantly more challenging, and relies on non-trivial analytical tools such as the Fenchel-Rockafellar duality theorem. We postpone that proof to the online appendix, and proceed to show weak duality.

\begin{lemma}[Weak Duality]\label{weakduality}
Let $\mu \in \Gamma(X)$. Then
$$\sup_{u \in \mathcal{U}(X) \cap \mathcal{L}_1(X)} \int_X u d\mu \leq \inf_{\substack{  \gamma \in \Gamma_+(X \times X) \\ \gamma_1 - \gamma_2 \succeq_{cvx} \mu   }}\int_{X\times X} \Lone{x-y} d\gamma.$$
\end{lemma}
\begin{prevproof}{Lemma}{weakduality}
For any feasible $u$ for the left-hand side and feasible $\gamma$ for the right-hand side, we have
$$\int_X u d \mu \leq \int_X u d (\gamma_1 - \gamma_2) = \int_{X \times X} (u(x) - u(y))d\gamma(x,y) \leq \int_{X \times X} \Lone{x-y} d\gamma(x,y)$$
where the first inequality follows from $\gamma_1 - \gamma_2 \succeq_{cvx} \mu$ and the second inequality follows from the 1-Lipschitz condition on $u$.
\end{prevproof}

From the proof of Lemma~\ref{weakduality}, we note the following ``complementary slackness'' conditions that a pair of optimal primal and dual solutions must satisfy.

\begin{corollary}\label{linearintegral}
Let $u^*$ and $\gamma^*$ be feasible for their respective problems above.
  Then $\int u^* d\mu= \int \Lone{x-y} d\gamma^*$ if and only if both of the following conditions hold:
\begin{enumerate}
	\item $\int u^* d(\gamma_1^* - \gamma_2^*) = \int u^* d  \mu$.
	\item $u^*(x) - u^*(y) = \Lone{x-y}$, $\gamma^*(x,y)$-almost surely.
\end{enumerate}
\end{corollary}
\begin{prevproof}{Corollary}{linearintegral}
The inequalities in the proof of Lemma~\ref{weakduality} are tight precisely when both conditions hold.
\end{prevproof}

\begin{mdframed}[style=boxed]
{\bf Interpretation of the Complementary Slackness Conditions}
\begin{remark}\label{geometricremark}
It is useful to geometrically interpret Corollary~\ref{linearintegral}:

\textbf{Condition 1:} We view $\gamma_1^* - \gamma_2^*$ (denote this by $\mu'$) as a ``shuffled'' $\mu$. Stemming from the $\mu' \succeq_{cvx} \mu$ constraint, the shuffling of $\mu$ into $\mu'$ is obtained via any sequence of the following operations: {\tt (1)} Picking a positive point mass $\delta_{x}$ from  $\mu_+$ and sending it from point $x$ to some other point $y \ge x$ (coordinate-wise). The constraint $\int u^*d\mu' = \int u^* d\mu$ requires that $u^*(x) = u^*(y)$. Recall that $u^*$ is non-decreasing, so $u^*(z) = u^*(x)$ for all $z \in \prod_j [x_j,y_j]$. Thus, if $y$ is strictly larger than $x$ in coordinate $i$, then $(\nabla u^*)_i = 0$ at all points $z$ ``in between'' $x$ and $y$. The other operation we are allowed, called a ``mean-preserving spread,'' is {\tt (2)} picking a positive point mass $\delta_{x}$ from  $\mu_+$, splitting the point mass into several pieces, and sending these pieces to multiple points while preserving the center of mass.  The constraint $\int u^*d\mu' = \int u^* d\mu$ requires that $u^*$ varies \textit{linearly} between $x$ and all points $z$ that received a piece.

\textbf{Condition 2:} The second condition is more straightforward than the first. We view $\gamma^*$ as a ``transport'' map between its component measures $\gamma^*_1$ and $\gamma^*_2$. The condition states that if $\gamma^*$ transports from location $x$ to location $y$, 
then $u^*(x) = u^*(y) + \Lone{x-y}$. If for some coordinate $i$, $x_i < y_i$, then $\Lone{z-y} < \Lone{x-y}$ for $z$ with $z_j = \max(x_j,y_j)$. This leads to a contradiction since $u^*(x) - u^*(y) \le u^*(z) - u^*(y) \le \Lone{z-y} < \Lone{x-y}$.
Therefore, it must be the case that \texttt{(1)} $x$ is component-wise greater than or equal to $y$ and \texttt{(2)} if $x_i > y_i$ in coordinate $i$, then $(\nabla u^*)_i = 1$ at all points ``in between'' $x$ and $y$. That is, the mechanism allocates item $i$ with probability 1 to all those types.
\end{remark}
\end{mdframed}

By Lemma~\ref{weakduality} and Corollary~\ref{linearintegral}, if we can find a ``tight pair'' of $u^*$ and $\gamma^*$, then they are optimal for their respective problems. This is useful since constructing a $\gamma$ that satisfies the conditions of Corollary~\ref{linearintegral} serves as a certificate of optimality for a mechanism. Theorem~\ref{strongduality} shows that this approach always works: for any optimal $u^*$ there always exists a $\gamma^*$ satisfying the conditions of Corollary~\ref{linearintegral}.

\begin{remark} \label{remark:comparison to DDT13}
It is useful to discuss what in our dual formulation in the RHS of~\eqref{eq:strong duality} makes it a strong dual, comparing to the previous work~\cite{DaskalakisDT13,GiannakopoulosK14}. If we were to tighten the $\gamma_1 - \gamma_2 \succeq_{cvx} \mu$ constraint in our dual formulation to a first-order stochastic dominance constraint, we essentially recover the duality framework of~\cite{DaskalakisDT13,GiannakopoulosK14}. Tightening the dual constraint, maintains the weak duality but creates a gap between the optimal primal and dual values. In particular, the dual problem resulting from tightening this constraint becomes a strong dual problem for a relaxed version of the mechanism design problem in which the convexity constraint on $u$ is dropped.
\end{remark}

 \section{Single-Item Applications and Interpretation} \label{sec:intuition}

Before considering multi-item settings, it is instructive to study the application of our strong duality theorem to single-item settings. We seek to relate the task of minimizing the transportation cost in the dual problem from Theorem~\ref{strongduality} to the structure of Myerson's solution~\cite{Myerson81}. 

Consider the task of selling a single item to a buyer whose value $z$ for the item is distributed according to a twice-differentiable regular distribution $F$ supported on $[\underbar{$z$},\bar{z}]$.\footnote{We remind the reader that a differentiable distribution $F$ is {\em regular} when its Myerson virtual value function $\phi(z) = z-{1-F(z) \over f(z)}$ is increasing in its support, where $f$ is the distribution density function.} Since $n=1$, if we were to apply our duality framework to this setting, we would choose $\mu$ according to~\eqref{transformed} as follows:
\begin{align*}
\mu(A) &= \mathbb{I}_A(\underbar{$z$})\cdot (1-f(\underbar{$z$}) \cdot \underbar{$z$}) + \mathbb{I}_A(\bar{z}) \cdot f(\bar{z}) \cdot \bar{z} -\int_{\underbar{$z$}}^{\bar{z}} \mathbb{I}_A(z) ( f'(z) \cdot z + 2 f(z))dz\\
&=\mathbb{I}_A(\underbar{$z$})\cdot (1-f(\underbar{$z$}) \cdot \underbar{$z$}) + \mathbb{I}_A(\bar{z}) \cdot f(\bar{z}) \cdot \bar{z} -\int_{\underbar{$z$}}^{\bar{z}}  \mathbb{I}_A(z) \left(\left( z - {1-F(z) \over f(z)}\right) f(z) \right)' dz
\end{align*}
We can interpret the transportation problem of Theorem~\ref{strongduality}, defined in terms of $\mu$, as:
\begin{itemize}
\item The sub-population of buyers having the right-most type, $\bar{z}$, in the support of the distribution have an excess supply of $f(\bar{z}) \cdot \bar{z}$;
\item The sub-population of buyers with the left-most type, $\underbar{$z$}$, in the support have an excess supply of $1-f(\underbar{$z$}) \cdot \underbar{$z$}$;
\item Finally, the sub-population of buyers at each other type, $z$, have a demand of $$\left(\left( z - {1-F(z) \over f(z)}\right) f(z) \right)' dz$$
\end{itemize}
One way to satisfy the above supply/demand requirements is to have every infinitesimal buyer of type $z$ push mass of $z - {1-F(z) \over f(z)}$ to its left. Since the fraction of buyers at $z$ is $f(z)$, the total amount of mass staying with them is then $\left(\left( z - {1-F(z) \over f(z)}\right) f(z) \right)' dz$ as required. Notice, in particular, that buyers with positive virtual types will push mass to their left, while buyers with negative virtual types will push mass to their right. 

The afore-described transportation map is feasible for our transportation problem as it satisfies all demand/supply constraints. We also claim that this solution is optimal.
To see this consider the mechanism that allocates the item to all buyers with non-negative virtual type at a fixed price $p^*$.  
The resulting utility function is of the form $\max\{z-p^*,0\}$.
We claim that this utility function satisfies the complementary slackness conditions of Remark~\ref{geometricremark} with respect to the transportation map identified above. Indeed,
when $z > p^*$, $u$ is linear with $u'(z)=1$ and mass is sent to the left---which is allowed by Part 2 of the remark, while, when $z < p^*$, $u$ is $0$ with $u'(z)=0$ and mass is sent to the right---allowed by Part 1({\tt 1}) of the remark.

In conclusion, when $F$ is regular, the virtual values dictate exactly how to optimally solve the optimal transportation problem of Theorem~\ref{strongduality}. Each infinitesimal buyer of type $z$ will push mass that equals its virtual value to its left. In particular, the optimal transportation  does not need to use mean-preserving spreads. Moreover, measure $\mu$ can be interpreted as the ``negative marginal normalized virtual value,'' as it assigns measure $-\left(\left( z - {1-F(z) \over f(z)}\right) f(z) \right)'dz$ to the interval $[z,z+dz]$, when $z \neq \underbar{$z$}, \bar{z}$.

When $F$ is not regular, the afore-described transportation map is not optimal due to the non-monotonicity of the virtual values. In this case, we need to pre-process our measure $\mu$ via mean-preserving spreads, prior to the transport, and ironing dictates how to do these mean-preserving spreads. In other words, ironing dictates how to perform the sweeping of the type set prior to transport.

\section{Multi-Item Applications of Duality}\label{examplesection}

We now give two examples of using Theorem~\ref{strongduality} to prove optimality of  mechanisms for selling two uniformly distributed independent items.

\subsection{Two Uniform $[0,1]$ Items}\label{sec:uniform 0-1}

Using Theorem~\ref{strongduality}, we provide a short proof of optimality of the mechanism for two i.i.d. uniform $[0,1]$ items proposed by \cite{ManelliV06} which we refer to as the MV-mechanism:

\begin{example}\label{unifexample01}
The optimal IC and IR mechanism for selling two items whose values are distributed uniformly and independently on the interval $[0,1]$ is the following menu:
\begin{itemize}
	\item buy any single item for a price of $\frac 2 3$; or
	\item buy both items for a price of ${4 - \sqrt{2} \over 3}$.
\end{itemize}
\end{example}

Let $Z$ be the set of types that receive no goods and pay $0$ to the MV-mechanism. Also, let $A$, $B$ be the set of types that receive only goods $1$ and $2$ respectively and $W$ be the set of types that receive both goods.  The sets $A,B, Z,W$ are illustrated in Figure~\ref{uniformfig} and separated by solid lines.

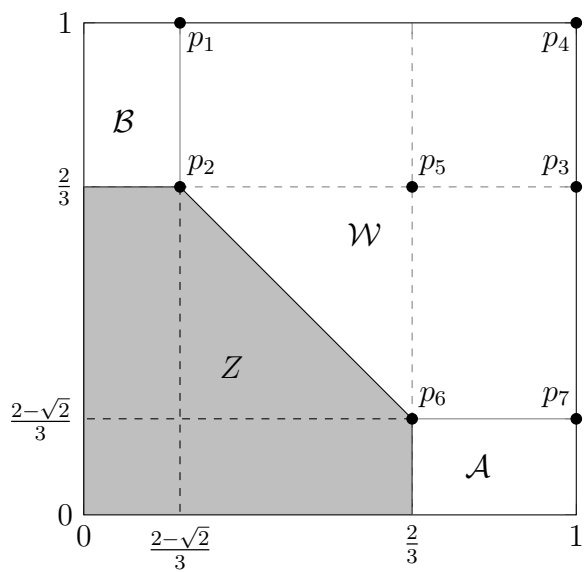
\begin{figure}[!ht]
\centering
\begin{tikzpicture}
\begin{axis}[
    height=3.2in, width=3.2in, ymin=0, ymax=1, xmin=0, xmax=1,
    xtick={0,0.19526214,0.66666,1},
    xticklabels={$0$,$2 - \sqrt{2}\over 3$,$2 \over 3$,$1$},
    ytick={0,0.19526214,0.66666,1},
    yticklabels={$0$,$2 - \sqrt{2}\over 3$,$2 \over 3$,$1$}
]

\addplot[color=black, fill=gray!50, mark=none] coordinates{
(0,0.66666)
(0.19526214,0.66666)
(0.66666,0.19526214)
(0.66666,0)
}\closedcycle;

\addplot[color=gray, mark=none]coordinates{
(0.66666,0.19526214)
(1,0.19526214)};

\addplot[color=black, mark=none, style=dashed] coordinates{
(0.66666,0.19526214)
(0,0.19526214)
};

\addplot[color=gray, mark=none]coordinates{
(0.19526214,0.66666)
(0.19526214,1)
};

\addplot[color=black, mark=none, style=dashed] coordinates{
(0.19526214,0.66666)
(0.19526214,0)
};

\addplot[color=gray, mark=none, style=dashed] coordinates{
(0,0.66666)
(1,0.66666)
};

\addplot[color=gray, mark=none, style=dashed] coordinates{
(0.66666,0)
(0.66666,1)
};

\node at (axis cs:0.08,0.8){$\mathcal{B}$};
\node at (axis cs:0.8,0.1){$\mathcal{A}$};
\node at (axis cs:0.57,0.57){$\mathcal{W}$};
\node at (axis cs:0.3,0.3){$Z$};

\node at (axis cs:0.24, 0.96){${p_1}$};
\node at (axis cs:0.24, 0.71){${p_2}$};
\node at (axis cs:0.96, 0.71){${p_3}$};
\node at (axis cs:0.96, 0.96){${p_4}$};
\node at (axis cs:0.71, 0.71){${p_5}$};
\node at (axis cs:0.71,0.24){${p_6}$};
\node at (axis cs:0.96, 0.24){${p_7}$};

\addplot[color=black, mark=*] coordinates{
(0.19526214, 1)};
\addplot[color=black, mark=*] coordinates{
(0.19526214, 0.666666)};
\addplot[color=black, mark=*] coordinates{
(1, 0.666666)};
\addplot[color=black, mark=*] coordinates{
(1, 1)};
\addplot[color=black, mark=*] coordinates{
(0.666666, 0.666666)};
\addplot[color=black, mark=*] coordinates{
(0.666666, 0.19526214)};
\addplot[color=black, mark=*] coordinates{
(1, 0.19526214)};
\end{axis}
\end{tikzpicture}
\caption{The MV-mechanism for two i.i.d. uniform $[0,1]$ items.}\label{uniformfig}
\end{figure}

Let us now try to prove that the MV mechanism is indeed optimal. As a first step, we need to compute the transformed measure $\mu$ of the uniform distribution on $[0,1]^2$. We have already computed $\mu$ in Section~\ref{examplesetup}. It has a point mass of $+1$ at $(0,0)$, a mass of $-3$ distributed uniformly over $[0,1]^2$, a mass of $+1$ distributed uniformly on the top boundary of $[0,1]^2$, and a mass of $+1$ distributed uniformly on the right boundary. Notice that the total net mass is equal to 0 within each region $Z$, $A$, $B$, or $W$.

To prove optimality of the MV-mechanism, we will construct an optimal $\gamma^*$  for the dual program of Theorem~\ref{strongduality} to match the positive mass $\mu_+$ to the negative $\mu_-$. Our $\gamma^*$ will be decomposed into $\gamma^* = \gamma^Z + \gamma^A + \gamma^B + \gamma^W$ and to ensure that $\gamma^*_1 - \gamma^*_2 \succeq_{cvx} \mu$, we will show that
$$\gamma^Z_1 - \gamma^Z_2 \succeq_{cvx} \mu|_Z; \quad \gamma^A_1 - \gamma^A_2 \succeq_{cvx} \mu|_A; \quad \gamma^B_1 - \gamma^B_2 \succeq_{cvx} \mu|_B; \quad \gamma^W_1 - \gamma^W_2 \succeq_{cvx} \mu|_W.$$
We will also show that the conditions of Corollary~\ref{linearintegral} hold for each of the measures $\gamma^Z$, $\gamma^A$, $\gamma^B$, and $\gamma^W$ separately, namely $\int u^* d(\gamma^S_1 - \gamma^S_2) = \int_S u^* d  \mu$ and $u^*(x) - u^*(y) = \Lone{x-y}$ hold $\gamma^S$-almost surely
for $S$ = $Z$, $A$, $B$,  and $W$.

\begin{itemize}[label={},leftmargin=0pt]
\item {\bf Construction of $\gamma^Z$:} Since $\mu_+|_Z$ is a point-mass at $(0,0)$ and  $\mu_-|_Z$ is distributed throughout a region which is coordinatewise greater than $(0,0)$, we notice that $\mu|_Z \preceq_{cvx} 0$. We set $\gamma^Z$ to be the zero measure, and the relation $\gamma^Z_1 - \gamma^Z_2 = 0 \succeq_{cvx} \mu|_Z$, as well as the two necessary equalities from Corollary~\ref{linearintegral}, are trivially satisfied.

\item {\bf Construction of $\gamma^A$ and $\gamma^B$:} In region $A$, $\mu_+|_A$ is distributed on the right boundary while  $\mu_-|_A$ is distributed uniformly on the interior of $A$. We construct $\gamma^A$ by transporting the positive mass $\mu_+|_A$ to the left to match the negative mass $\mu_-|_A$. Notice that this indeed matches completely the positive mass to the negative since $\mu(A) = 0$ and intuitively minimizes the $\ell_1$ transportation distance. To see that the two necessary equalities from Corollary~\ref{linearintegral} are satisfied, notice that $\gamma^A_1 = \mu_+|_A, \gamma^A_2 = \mu_-|_A$ so the first equality holds. The second inequality holds as we are transporting mass only to the left and thus the measure $\gamma^A$ is concentrated on pairs $(x,y) \in A \times A$ such that $1 = x_1 \ge y_1 \ge \frac 2 3$ and $x_2 = y_2$. Moreover, for all such pairs $(x,y)$, we have that $u(x) - u(y) = (x_1 - \frac 2 3) - (y_1 - \frac 2 3) = x_1 - y_1 = \Lone{x-y}$. The construction of $\gamma^B$ is similar.

\item {\bf Construction of $\gamma^W$} We construct an explicit matching that only matches leftwards and downwards without doing any prior mass shuffling. We match the positive mass on the segment $p_1p_4$ to the negative mass on the rectangle $p_1p_2p_3p_4$ by moving mass downwards. We match the positive mass of the segment $p_3p_7$ to the negative mass on the rectangle $p_3p_5p_6p_7$ by moving mass leftwards. Finally, we match the positive mass on the segment $p_3p_4$ to the negative mass on the triangle $p_2p_5p_6$ by moving mass downwards and leftwards. Notice that all positive/negative mass in region $ W$ has been accounted for, all of $(\mu|_{ W})_+$ has been matched to all of $(\mu|_{ W})_-$ and all moves were down and to the left, establishing  $u(x) - u(y) = (x_1 + x_2 - {4 - \sqrt{2} \over 3}) - (y_1 + y_2 - {4 - \sqrt{2} \over 3}) = x_1+x_2 - y_1 - y_2 = \Lone{x-y}$.
\end{itemize}

\subsection{Two Uniform But Not Identical Items}\label{sec: uniform non 0-1}

We now present an example with two items whose values are distributed uniformly and independently on the intervals $[4,16]$ and $[4,7]$. We note that the distributions are not identical, and thus the characterization of  \cite{Pavlov11} does not apply. In addition, the relaxation-based duality framework of \cite{DaskalakisDT13,GiannakopoulosK14} (see Remark~\ref{remark:comparison to DDT13}) fails in this example: if we were to relax the constraint that the utility function $u$ be convex, the ``mechanism design program'' would have a solution with greater revenue than is actually possible.

\begin{example}\label{unifexample}
The optimal IC and IR mechanism for selling two items whose values are distributed uniformly and independently on the intervals $[4,16]$ and $[4,7]$ is as follows:
\begin{itemize}
	\item If the buyer's declared type is in region $Z$, he receives no goods and pays nothing.
	\item If the buyer's declared type is in region $Y$, he  pays a price of 8 and receives the first good with probability $50\%$ and the second good with probability 1.
	\item If the buyer's declared type is in region $W$, he gets both goods for a price of 12.
\end{itemize}

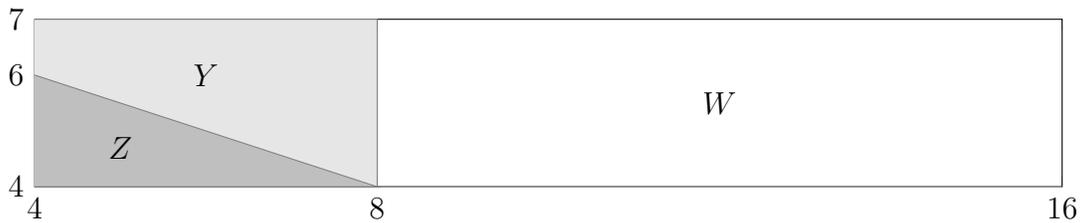
\begin{figure}[!ht]
\begin{center}
\begin{tikzpicture}
\begin{axis}[width=6in, height=1.5in,ymin=4, ymax=7, xmin=4, xmax=16,   ytick pos=left, ytick={4,6,7}, yticklabels={4,6,7}, xtick={4, 8,16},xticklabels={4,$8$,16}]
  \addplot+[color=gray, fill=gray!20, domain=0:8,mark=none]
 {7}
 \closedcycle;
  \addplot+[color=gray, fill=gray!50, domain=0:16,mark=none]
 {8-0.5*x}
 \closedcycle;
\node at (axis cs:5,4.7){$Z$};
\node at (axis cs: 6, 6){$Y$};
\node at (axis cs: 12, 5.5){$W$};
\end{axis}
\end{tikzpicture}
\caption{Partition of $[4,16] \times [4,7]$ into different regions by the optimal mechanism.}\label{fig:unifexample}
\end{center}
\end{figure}
\end{example}

The proof of optimality of our proposed mechanism works by constructing a measure $\gamma = \gamma^Z + \gamma^Y+\gamma^W$ separately in each region. The constructions of $\gamma^W$ and $\gamma^Z$ are similar to the previous example. The construction of $\gamma^Y$, however, is a little more intricate as it requires an initial shuffling of the mass before computing the optimal way to transport the resulting mass. The proof is presented in the online appendix.

\subsection{Discussion}
Our examples in this section serve to illustrate how to use our duality theorem to verify the optimality of our proposed mechanisms, without explaining how we identified these mechanisms. These mechanisms were in fact identified by ``reverse-engineering'' the duality theorem. The next two sections provide tools for performing this reverse-engineering. In particular, Section~\ref{bundlingsection} provides a characterization of mechanism optimality in terms of stochastic dominance conditions satisfied in regions partitioning the type space. Alleviating the need to reverse-engineer the duality theorem, Section~\ref{weakstructural} prescribes a straightforward procedure for identifying optimal mechanisms. We use this procedure to solve several examples in Section~\ref{sec:further examples}.

\section{Characterizing Optimal Finite-Menu Mechanisms}\label{bundlingsection}

To prove the optimality of our mechanisms in the examples of Section~\ref{examplesection}, we explicitly constructed a measure $\gamma$ \emph{separately} for each subset of types enjoying the same allocation in the optimal mechanism, establishing that the conditions of Corollary~\ref{linearintegral} are satisfied for each such subset of types separately. In this section, we show that decomposing the solution $\gamma$ of the optimal transportation dual of Theorem~\ref{strongduality} into ``regions'' of types enjoying the same allocation in the optimal solution $u$ of the primal, and working on these regions separately to establish the complementary slackness conditions of Corollary~\ref{linearintegral} is guaranteed to work. 

Even with this understanding of the structure of  dual witnesses, it may still be non-trivial work to identify a witness certifying the optimality of a given mechanism. We thus develop a more usable framework for certifying the optimality of mechanisms, which does not involve finding dual witnesses at all. In particular, we show in Theorem~\ref{bundlingtheorem} that a given mechanism $\cM$ is optimal for some $f$ {\em if and only if} appropriate stochastic dominance conditions are satisfied by the restriction of the transformed measure $\mu$ of Definition~\ref{def:transformed measure} to each region of types enjoying the same allocation under $\cM$. We thus provide conditions that are both \emph{necessary} and \emph{sufficient} for a given mechanism $\cM$ to be optimal, a characterization result.

To describe our characterization, we define the intuitive notion a ``menu'' that a certain mechanism offers.
\begin{definition}
  The {\em menu of a mechanism} $\cM = (\cP, \cT)$ is the set 
  $$\textrm{Menu}_{\cM} = \{(p,t): \exists x \in X, (p,t) = (\cP(x),\cT(x))\}.$$
\end{definition}
\noindent Clearly, an IC mechanism allocates to every type $x$ the option in the menu that maximizes that type's utility. Figure~\ref{fig:Voronoi} shows an example of a menu and the corresponding partition of the type set into subsets of types that prefer each option in the menu.

\begin{figure}[htbp]
	\centering
		\includegraphics[width=0.6\textwidth]{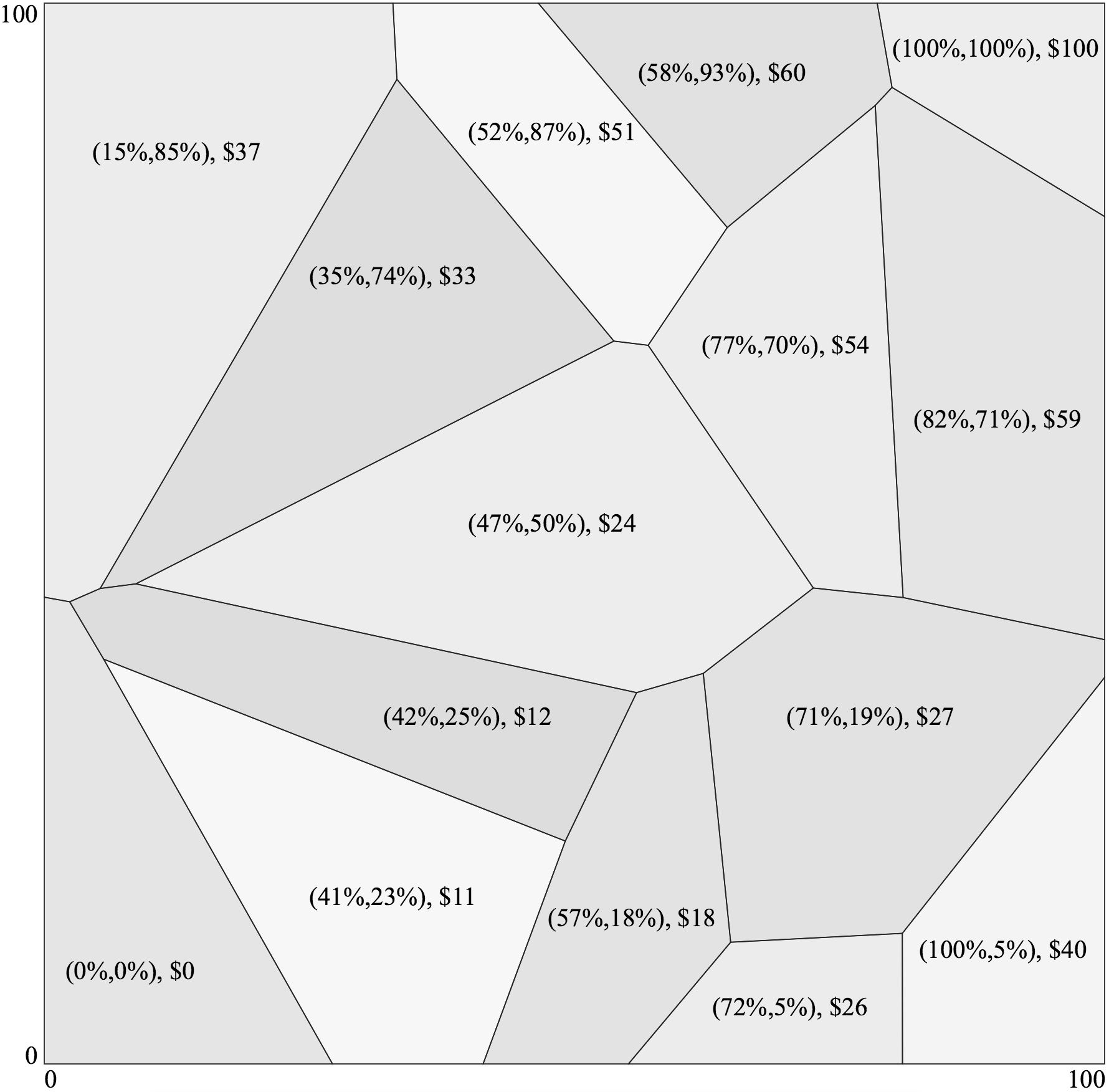}
	\caption{Partition of the type set $X = [0,100]^2$ induced by some menu of lotteries.}
	\label{fig:Voronoi}
\end{figure}

The revenue of a mechanism with a finite menu-size comes from choices in the menu that are bought with strictly positive probability. The menu might contain options that are only bought with probability $0$, but we can get another mechanism that gives identical revenue by removing all those options. We call this the {\em essential form } of a mechanism.
\begin{definition}
  A mechanism ${\cM}$ is in {\em essential form } if for all options $(p,t) \in \textrm{Menu}_{\cM}$, $\Pr_f[ \{ x \in X : (p,t) = (\cP(x),\cT(x))\} ] > 0$.
\end{definition}

We will now show our main result of this section under the assumption that the menu size is finite. We expect that our tools can be used to extend the results to the case of infinite menu size with a more careful analysis. We stress that the point of our result is not to provide sufficient conditions to certify optimality of mechanisms, as in \cite{ManelliV06,DaskalakisDT13,GiannakopoulosK14}, but to provide necessary and sufficient conditions. In particular, we show  that verifying optimality is \emph{equivalent} to checking a collection of measure-theoretic inequalities, and this applies to arbitrary mechanisms with a finite menu-size. The proof of our result is intricate, requiring several technical lemmas, so it is postponed to the online appendix. The most crucial component of the proof establishes that the optimal dual solution $\gamma$ in Theorem~\ref{strongduality} never convexly shuffles mass across regions of types that enjoy different allocations. (I.e. to obtain $\mu'=\gamma_1-\gamma_2$ from $\mu$ we never need to move mass across different regions.)
Similarly, we argue that the optimal $\gamma$ never transports mass across regions.

Before formally stating our result, it is helpful to provide some intuition behind it. Consider a region $R$ corresponding to a menu choice $(\vec p,t)$ of an optimal mechanism $\cal M$. As we have already discussed, we can establish that the dual witness $\gamma$, which witnesses the optimality of $\cal M$, does not transport  mass between regions and, likewise, the associated ``convex shuffling'' transforming $\mu$ to $\mu'=\gamma_1-\gamma_2$ doesn't shuffle across regions. Given this, our complementary slackness conditions of Corollary~\ref{linearintegral} imply then that $\mu_+|_R$ can be transformed to $\mu_-|_R$ using the following (intra-region $R$) operations:
\begin{itemize}
\item spreading positive mass within $R$ so that the mean is preserved
\item sending (positive) mass from a point $x \in R$ to a coordinatewise larger point $y\in R$ if for all coordinates where $y_i > x_i$ we have that the corresponding probability of the menu choice satisfies $p_i = 0$
\item sending (positive) mass from a point $x\in R$ to a coordinatewise smaller point $y \in R$ if for all coordinates where $y_i < x_i$ we have that $p_i = 1$
\end{itemize}

Our characterization result involves stochastic dominance conditions that are slightly more general than the standard notions of first, second and convex dominance. 
We need the following definition, which extends the notion of convex dominance.

\begin{definition}\label{firstorderdef}
We say that a function $u: X \rightarrow \mathbb{R}$ is $\vec v$-monotone for a vector $\vec v \in \{-1,0,+1\}^n$ if it is non-decreasing in all coordinates $i$ for which $v_i=1$ and non-increasing in all coordinates $i$ for which $v_i=-1$.

A measure $\alpha$ \emph{convexly dominates} a measure $\beta$ \emph{with respect to}  a vector $\vec v \in \{-1,0,+1\}^n$, denoted $\alpha \succeq_{\textrm{cvx}(\vec v)} \beta$, if for all convex $\vec v$-monotone functions $u \in {\cal U}(X)$:
$$\int u d\alpha \geq \int u d\beta.$$

Similarly, for vector random variables $A$ and $B$ with values in $X$, we say that $A \succeq_{\textrm{cvx}(\vec v)} B$ if $\mathbb{E}[u(A)] \geq \mathbb{E}[u(B)]$ for all convex $\vec v$-monotone functions $u \in {\cal U}(X)$.
\end{definition}
The definition of convex dominance presented earlier coincides with convex dominance with respect to the vector $\vec 1$. Moreover, convex dominance with respect to the vector $-\vec 1$ is related to second-order stochastic dominance as follows: 
$$\alpha \succeq_{\textrm{cvx}(- \vec 1)} \beta \Leftrightarrow \beta \succeq_{2} \alpha.$$
Measures satisfying the dominance condition of Definition~\ref{firstorderdef} must have equal mass.

\begin{proposition}\label{prop:equal-mass}
Fix two measures $\alpha, \beta \in \Gamma(X)$ and a vector $v \in \{-1,0,1\}^n$. 
If it holds that $\alpha \succeq_{\textrm{cvx}(\vec v)} \beta$, then $\alpha(X) = \beta(X)$.
\end{proposition}

We are now ready to describe our main characterization theorem. Our characterization, stated below as Theorem~\ref{bundlingtheorem} and proven in the online appendix, is given in terms of the conditions of Definition~\ref{optconditions}.

\begin{definition}[Optimal Menu Conditions]
\label{optconditions}
A mechanism $\cM$ {\em satisfies the optimal menu conditions with respect to $\mu$} if for all
menu choices $(p,t) \in \textrm{Menu}_{\cM}$ we have
$$ \mu_+|_{R} \preceq_{cvx(\vec v)} \mu_-|_{R}$$
where $R = \left\{x \in X : (\cP(x),\cT(x))=(p,t) \right\}$ is the subset of types that receive $(p,t)$  and $\vec v$ is the vector whose $i$-th coordinate $v_i$ takes value 1 if $p_i = 0$, value $-1$ if $p_i = 1$ or value $0$ if $p_i \in (0,1)$.

\end{definition}

\begin{theorem}[Optimal Menu Theorem]\label{bundlingtheorem}
Let $\mu$ be the transformed measure of a probability density $f$ as per Definition~\ref{def:transformed measure}. Then a mechanism $\cM$ with finite menu size is an optimal IC and IR mechanism for a single additive buyer whose values for $n$ goods are distributed according to the joint distribution $f$ \textbf{if and only if} its essential form satisfies the optimal menu conditions with respect to $\mu$.
\end{theorem}

\begin{mdframed}[style=boxed]
{\bf Interpretation of the Optimal Menu Conditions:} A simple interpretation of the optimal menu conditions that Theorem~\ref{bundlingtheorem} claims are necessary and sufficient for the optimality of mechanisms is this. Take some region $R$ of the type set $X$ corresponding to the types that are allocated a specific menu choice $(p,t)$ by optimal mechanism~$\cM$. Let us consider the revenue $\int_R u^* d\mu$ extracted by $\cM$ from the types in region $R$. Is it possible to extract more revenue from these types? We claim that the optimal menu condition for region $R$ guarantees that no mechanism can possibly extract more from the types in region $R$. Indeed, consider any utility function $u$ induced by some other mechanism.
The revenue extracted by this other mechanism in region $R$ is $\int_R u d\mu = \int_R u^* d\mu + \int_R (u - u^*) d\mu \le \int_R u^* d\mu$. That $\int_R (u - u^*) d\mu \le 0$ follows directly from the optimal menu condition for region $R$. Indeed, since $u^*(x) = p \cdot x - t$ in region $R$, it follows that, whatever choice of $u$ we made, $u - u^*$ is a convex $\vec v$-monotone function in region $R$, where $\vec v$ is the vector defined by $p$ as per Definition~\ref{optconditions}. Our condition in region $R$ reads $\mu|_R \preceq_{\textrm{cvx}(\vec v)} 0$, hence $\int_R (u - u^*) d\mu \le 0$. Our line of argument implies the sufficiency of the optimal menu conditions, as they imply that for each region separately no mechanism can beat the revenue extracted by $\cM$. The more surprising part (and harder to prove) is that the conditions are also necessary, implying that optimal mechanisms are locally optimal for every region $R$ of types that they allocate the same menu choice to.
\end{mdframed}

A particularly simple special case of our characterization result, pertains to the optimality of the grand-bundling mechanism. Theorem~\ref{bundlingtheorem} implies that the mechanism that offers the grand bundle at price $p$ is optimal {\em if and only if} the  transformed measure $\mu$ satisfies a pair of stochastic dominance conditions.
In particular, we obtain the following theorem:

\begin{theorem}[Grand Bundling Optimality]\label{grandbundlingtheorem}
For a single additive buyer whose values for $n$ goods are distributed according to the joint distribution $f$, the mechanism that only offers the bundle of all items at  price $p$ is optimal if and only if the transformed measure $\mu$ of $f$ satisfies $\mu|_{\mathcal{W}} \succeq_2 0 \succeq_{\textrm{cvx}} \mu|_{\mathcal{Z}}$, where $\mathcal{W}$ is the subset of types that can afford the
grand bundle at price $p$, and $\mathcal{Z}$ the subset of types who cannot.
\end{theorem}

\noindent Next, we explore implications of our characterization of grand bundling optimality.

\subsection{Example Applications of Grand Bundling Optimality}\label{bundlingexamplesection}

We now present an example application of our characterization result to determine the optimality of mechanisms that make a take-it-or-leave-it offer of the grand bundle of all items at some price.  Our result applies to a setting with arbitrarily many items, which is relatively rare in the  literature. More specifically, we consider a setting with $n$ iid goods whose values are uniformly distributed on $[c,c+1]$. It is easy to see that the ratio of the revenue achievable by grand bundling to the social welfare goes to 1 when either $n$ or $c$ goes to infinity.\footnote{This follows by setting a price for the grand-bundle equal to $(c + \frac 1 2) n - \sqrt{n} \log{c n}$ and noting that a straightforward application of Hoeffding's inequality gives that the bundle is accepted with probability close to 1.} This implies that grand-bundling is optimal or close to optimal for large values of $n$ and $c$. Indeed, the following theorem shows that, for every $n$, grand bundling is the optimal mechanism for large values of $c$.

\begin{theorem}\label{nuniform}
For any integer $n > 0$ there exists a $c_0$ such that for all $c \ge c_0$, the optimal mechanism for selling $n$ iid goods whose values are uniform on $[c,c+1]$ is a take-it-or-leave-it offer for the grand bundle.
\end{theorem}

\begin{remark}
\cite{Pavlov11} proved the above result for two items, and explicitly solved for $c_0 \approx 0.077$. In our proof, for simplicity of analysis, we do not attempt to exactly compute $c_0$ as a function of $n$.
\end{remark}

Our proof of Theorem~\ref{nuniform} uses the following lemma, which enables us to appropriately match regions on the surface of a hypercube. The proof of this lemma and of Theorem~\ref{nuniform} appears in the online appendix.
\begin{lemma}\label{matching}
For $n \geq 2$ and $\rho > 1$,  define the $(n-1)$-dimensional subsets of $[0,1]^n$:
\begin{align*}
A &= \left\{x:1 = x_1 \geq x_2 \geq \cdots \geq x_n \textrm{ and } x_n \leq 1-   \left(\frac{\rho-1}{\rho} \right)^{1/(n-1)} \right\}\\
 B &= \left\{ y:y_1 \geq \cdots \geq y_n = 0 \right\}.
\end{align*}
There exists a continuous bijective map $\varphi : A \rightarrow B$ such that
\begin{itemize}
\item For all $x\in A$, $x$ is componentwise greater than or equal to $\varphi(x)$
\item For subsets $S \subseteq A$ which are measurable under the $(n-1)$-dimensional surface Lebesgue measure $v(\cdot)$, it holds that $\rho \cdot v(S) = v(\varphi(S))$.
\item For all $\epsilon > 0$, if $\varphi_1(x) \leq \epsilon$ then $x_n \geq 1- \left( \frac{\epsilon^{n-1} +\rho-1}{\rho}  \right)^{1/(n-1)}$.
\end{itemize}
\end{lemma}
\begin{figure}[!ht]
\begin{center}
\includegraphics[scale=0.2]{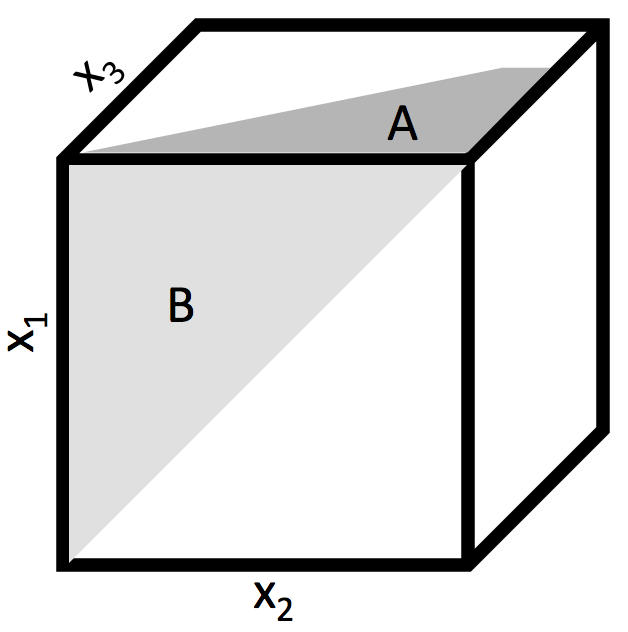}
\end{center}
\caption{The regions of Lemma~\ref{matching} for the case $n = 3$.}
\end{figure}

The main difficulty in proving Theorem~\ref{nuniform} is verifying the necessary stochastic dominance relations above the grand bundling hyperplane. Our proof appropriately partitions this part of the hypercube into $2(n!+1)$ regions and uses Lemma~\ref{matching} to show a desired stochastic dominance relation holds for an appropriate pairing of regions. The proof of Theorem~\ref{nuniform} is in the online appendix.

We now consider what happens when $n$ becomes large while $c$ remains fixed. In this case, in contrast to the previous result, we show using our strong duality theorem that grand bundling is \emph{never} the optimal mechanism for sufficiently large values of $n$.

\begin{theorem}\label{nuniform-notbundling}
For any $c \ge 0$ there exists an integer $n_0$ such that for all $n \ge n_0$, the optimal mechanism for selling $n$ iid goods whose values are uniform on $[c,c+1]$ is {\bf not}  a take-it-or-leave-it offer for the grand bundle.
\end{theorem}

\begin{proof}
Given $c$, let $n$ be large enough so that
$$\frac{n+1}{n!}+\frac{nc}{(n-1)!}<1.$$
To prove the theorem, we will assume that an optimal grand bundling price $p$ exists and reach a contradiction. 

As shown in Section~\ref{examplesetup}, under the transformed measure $\mu$ the hypercube has mass $-(n+1)$ in the interior, $+1$ on the origin, $c+1$ on every positive surface $x_i=c+1$, and $-c$ on every negative surface $x_i = c$.

According to Theorem~\ref{bundlingtheorem}, for grand bundling at price $p$ to be optimal it must hold that $\mu|_{Z_p} \preceq_{cvx} 0$ for the region $Z_p = \{x :\Lone{x} \leq p\}$. If $p>n c + 1$ this could not happen, since for the function $\mathbbm{1}_{x_1 = c + 1}(x)$ (which is increasing and convex in $[c,c+1]^n$) we have that $\int_{Z_p} \mathbbm{1}_{x_1 = c + 1} \, d\mu = \mu(Z_p \cap \{x_1 = c + 1\}) = \mu_+(Z_p \cap \{x_1 = c + 1\}) > 0$ which violates the $\mu|_{Z_p} \preceq_{cvx} 0$ condition.

To complete the proof, we now consider the case that $p \le n c+1$ and will derive a contradiction. For the necessary condition $\mu|_{Z_p} \preceq_{cvx} 0$ to hold, it must be that $\mu(Z_p)=0$. Since $p \le n c + 1$, none of the positive outer surfaces of the cube have nontrivial intersection with $Z_p$, so all the positive mass in $Z_p$ is located at the origin. Therefore, $\mu_{+} (Z_p) = 1$ which means that $\mu_{-} (Z_p) = 1$ as well. Moreover, since $p \le n c+1 \Rightarrow Z_p \subseteq Z_{n c + 1}$, we also have that $\mu_{-} (Z_{n c + 1}) \ge \mu_{-} (Z_p) = 1$. 

To reach a contradiction, we will show that $\mu_-(Z_{nc+1})<1$. We observe that we can compute $\mu_{-}(Z_{n c + 1})$ directly by summing the $n$-dimensional volume of the negative interior with the $(n-1)$-dimensional volumes of each of the $n$ negative surfaces enclosed in $Z_{n c + 1}$.\footnote{The geometric intuition of this step of the argument is that, for large enough $n$, the fraction of the $n$-dimensional hypercube $[0,1]^n$ which lies below the diagonal $||x||=1$ goes to zero, and similarly the fraction of $(n-1)$-dimensional surface area on the boundaries which lies below the diagonal also goes to zero as $n$ gets large.} 
The first is equal to:
\begin{align*}
(n+1) \times \textrm{Vol}\left[ \{x \in (c,c+1)^n :\Lone{x} \leq n c + 1 \} \right] &= \\ (n+1) \times \textrm{Vol}\left[ \{x \in (0,1)^n: \Lone{x} \leq 1 \} \right] &=
\frac { (n+1)} {n!}
\end{align*}
while the latter is equal to:
\begin{align*}
n \times c \times \textrm{Vol}\left[ \{x \in (c,c+1)^{n-1} :\Lone{x} + c \leq n c + 1 \} \right] &= \\ 
n \times c \times \textrm{Vol}\left[ \{x \in (0,1)^{n-1}:\Lone{x} \leq 1 \} \right] &=
\frac { n c} {(n-1)!}
\end{align*}
Therefore, we get that $1 \le \mu_{-} (Z_{n c + 1}) = \frac { (n+1)} {n!} + \frac{ n c } {(n-1)!}$ which is a contradiction since we chose $n$ to be sufficiently large to make this quantity less than 1.
\end{proof}
  
\section{Constructing Optimal Mechanisms}\label{weakstructural}

\subsection{Preliminaries}

The results of the previous section characterize optimal mechanisms and give us the tools to check if a mechanism is optimal. In this section, we show how to use the optimal menu conditions we developed to identify candidate mechanisms. In particular, Theorem~\ref{bundlingtheorem} implies that (in the finite menu case) to find an optimal mechanism we need to identify a set of choices for the menu, such that for every region $R$ that corresponds to a menu outcome it holds that $\mu_+|_{R} \preceq_{cvx(\vec v)} \mu_-|_{R}$ for the appropriate vector $\vec v$. This implies that $\mu_+({R}) = \mu_-({R})$, so at the very least the total positive and the total negative mass in each region need to be equal. This property immediately helps us exclude a large class of mechanisms and guides us to identify potential candidates. We note that in this section we will develop techniques which apply not just to finite-menu mechanisms but to mechanisms with infinite menus as well.

We will restrict ourselves to a particularly useful class of mechanisms defined completely by the set of types that are excluded from the mechanism, i.e. they receive no items and pay nothing. We call this set of types the \emph{exclusion set} of a mechanism.  The exclusion set gives rise to a mechanism where the utility of a buyer is equal to the $\ell_1$ distance between the buyer's type and the closest point in the exclusion set. All known instances of optimal mechanisms for independently distributed items fall under this category. We proceed to define these concepts formally.

\begin{definition}[Exclusion Set]
Let $X = \prod_{i=1}^n [x^{\textrm{low}}_i,x^{\textrm{high}}_i]$. An {\em exclusion set $Z$ of $X$} is a convex, compact, and decreasing\footnote{A decreasing subset $Z \subset X$ satisfies the property that for all $a, b \in X$ such that $a$ is component-wise less than or equal to $b$, if $b \in Z$ then $a \in Z$ as well.} subset of $X$ with nonempty interior. 
\end{definition}

\begin{definition}[Mechanism of an Exclusion Set] \label{def:mechanism from exclusion set}
Every exclusion set $Z$ of $X$ induces a mechanism whose utility function $u_Z: X \rightarrow \mathbb{R}$ is defined by:
$$u_Z(x) = \min_{z \in Z} \Lone{z-x}.$$
\end{definition}
Note that, since the exclusion set $Z$ is closed, for any $x \in X$ there exists a $z \in Z$ such that $u_Z(x) = \Lone{z-x}$. Moreover, we show below that any such utility function $u_Z$ satisfies the constraints of the mechanism design problem. That is, the mechanism corresponding to $u_Z$ is IC and IR. The proof of the following claim is straightforward casework and appears in the online appendix. 

\begin{claim}\label{zerosetICIR}
Let $Z$ be an exclusion set of $X$. Then $u_Z$ is non-negative,  non-decreasing, convex, and has Lipschitz constant (with respect to the $\ell_1$ norm) at most $1$. In particular, $u_Z$ is the utility function of an incentive compatible and individually rational mechanism.
\end{claim}

\subsection{Constructing Optimal Mechanisms for 2 Items}

To provide sufficient conditions for $u_Z$ to be optimal for the case of 2 items, we define the concept of a canonical partition. A canonical partition divides $X$ into regions such that the mechanism's allocation function within each region has a similar form. Roughly, the canonical partition separates $X$ based on which direction (either ``down,'' ``left,'' or ``diagonally'') one must travel to reach the closest point in $Z$. While the definition is involved, the geometric picture of Figure~\ref{canonicalfig} is straightforward.

\begin{definition}[Critical price, Critical point, Outer boundary functions] \label{def:outer boundary}
Let $Z$ be an exclusion set of $X$. Denote by $P$ the maximum value $P = \max \{x+y : (x,y) \in Z\}$, we call $P$ the {\em critical price}.
We now define the {\em critical point} $(x_{\textrm{crit}},y_{\textrm{crit}})$, such that 
$$x_{\textrm{crit}} = \min \{x : (x,P-x) \in Z\} \, \textrm{  and  } \, y_{\textrm{crit}} = \min \{y : (P-y,y) \in Z\}$$
We define the {\em outer boundary functions of $Z$} to be the functions $s_1, s_2$ given by $$s_1(x) = \max \{y : (x,y) \in Z\} \textrm{ and } s_2(y) = \max \{x : (x,y) \in Z\},$$
with domain $[0,x_{\textrm{crit}}]$ and $[0,y_{\textrm{crit}}]$ respectively.
\end{definition}

\begin{definition}[Canonical partition]\label{canonicalpartitiondef}
Let $Z$ be an exclusion set of $X$ with critical point $(x_{\textrm{crit}},y_{\textrm{crit}})$ as in Definition~\ref{def:outer boundary}.
We define {\em the canonical partition of $X$ induced by $Z$} to be the partition of $X$ into $Z \cup \mathcal{A} \cup \mathcal{B} \cup \mathcal{W}$, where
$$\mathcal{A} = \{(x,y) \in X : x < x_{\textrm{crit}}\} \setminus Z; \quad \mathcal{B} = \{(x,y) \in X : y < y_{\textrm{crit}}\} \setminus Z; \quad \mathcal{W} = X \setminus (Z \cup \mathcal{A} \cup \mathcal{B}),$$
as shown in Figure~\ref{canonicalfig}.
\end{definition}

Note that the outer boundary functions $s_1, s_2$ of an exclusion set $Z$ are concave and thus are differentiable almost everywhere on $[0,c_1]$ and have non-increasing derivatives.

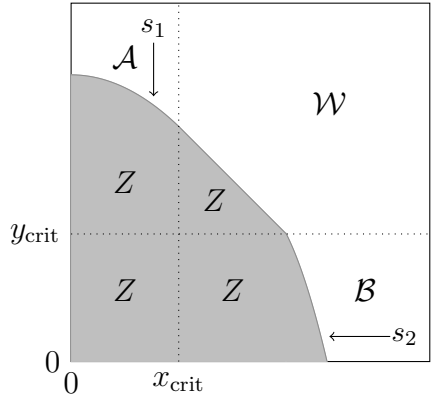
\begin{figure}[!ht]\centering
\begin{tikzpicture}
\begin{axis}[width=2.5in, height=2.5in,ymin=0, ymax=10, xmin=0, xmax=10,   ytick pos=left, ytick={3.56,0} , yticklabels={$y_{\textrm{crit}}$,$0$}, xtick={3,0},xticklabels={$x_{\textrm{crit}}$,$0$}, xtick pos=left]
  
  \addplot+[color=gray!50, fill=gray!50, domain=0:3,mark=none]
 {8-.16*x*x}
 \closedcycle;
   \addplot+[color=gray!50, fill=gray!50, domain=3:6,mark=none]
 {9.56-x}
 \closedcycle;
\addplot+[color=gray!50,fill=gray!50,domain=6:7.135,mark=none]
{4.56-(x-5)*(x-5)} \closedcycle;

  \addplot+[color=gray!90, fill=gray!50, domain=0:3,mark=none]
 {8-.16*x*x};
   \addplot+[color=gray!90, fill=gray!50, domain=3:6,mark=none]
 {9.56-x};
\addplot+[color=gray!90,fill=gray!50,domain=6:7.135,mark=none, solid]
{4.56-(x-5)*(x-5)};
 
 \addplot+[color=black, domain=0:10,samples=2,  mark=none, dotted]
 {3.56} 
 ;
 \addplot[color=black, , dotted, mark=none] coordinates{
(3,0)
(3,11)
};
\node at (axis cs:1.5,5){$Z$};
\node at (axis cs:4.5,2){$Z$};
\node at (axis cs:1.5,2){$Z$};
\node at (axis cs:4,4.5){$Z$};
\node at (axis cs:8.2,2){$\mathcal{B}$};    
\node at (axis cs:7.2,7.2){$\mathcal{W}$};
\node at (axis cs:1.5,8.5){$\mathcal{A}$};
\node (dest4) at (axis cs: 2.3, 7.1){};
\node (lab4) at (axis cs: 2.3,9.2){};
\node (dest5) at (axis cs: 6.9,0.7){};
\node (lab5) at (axis cs: 9.2,0.7){};
\draw [->] (lab4)--(dest4);
\draw [->] (lab5)--(dest5);
\node at (axis cs: 2.3,9.3){\small{$s_1$}};
\node at (axis cs: 9.3,0.7){\small{$s_2$}};
\end{axis}
\end{tikzpicture}
\caption{The canonical partition}\label{canonicalfig}
\end{figure}

We  now restate the utility function $u_Z$ of a mechanism with exclusion set $Z$ in terms of a canonical partition.

\begin{claim}
Let $Z$ be an exclusion set of $X$ with outer boundary functions $s_1, s_2$ and critical price $P$, and let $Z \cup \mathcal{A} \cup \mathcal{B} \cup \mathcal{W}$ be its canonical partition. Then for all $(v_1, v_2) \in X$, the utility function $u_Z$ of the mechanism with exclusion set $Z$ is given by:
$$
u_Z(v_1,v_2) = 
\begin{cases}
0 & \mbox{if } (v_1,v_2) \in Z \\
v_2 - s_1(v_1) & \mbox{if } (v_1,v_2) \in \mathcal{A}\\
v_1 - s_2(v_2) & \mbox{if } (v_1,v_2) \in \mathcal{B} \\
v_1 + v_2 - P & \mbox{if } (v_1,v_2) \in \mathcal{W}.
\end{cases}
$$
\end{claim}
\begin{proof}
The proof is fairly straightforward casework. {We prove one of the cases here, and the remaining cases are similar.}

Pick any $v = (v_1, v_2) \in \mathcal{A}$. We will show that the closest $z \in Z$ is the point $z^* = (v_1, s_1(v_1))$. Pick $z' = (z_1', z_2') \in Z$ such that $u_Z(v) = \Lone{v-z'}$. It must be the case that $z_1' \leq v_1$, since otherwise $(v_1, z_2')$ would be in $Z$ (as $Z$ is decreasing) and strictly closer to $v$. 

We now have that $\Lone{v-z'} \ge \Lone{v}-\Lone{z'} \ge \Lone{v}-\max_{x \in [0,v_1]} (x + s_1(x))$.
Since the less restricted maximization problem, $\max_{x \in [0,x_{\textrm{crit}}]} (x + s_1(x))$ is maximized at $x_{\textrm{crit}}$ and the function $(x + s_1(x))$ is concave, the maximum of the more constrained version is achieved at $x = v_1$. Thus, we have that, $\Lone{v-z'} \ge \Lone{v} - v_1 - s_1(v_1) = v_2 - s_1(v_1) = \Lone{v-z^*}$.
\end{proof}

We now describe sufficient conditions under which  $u_Z$ is optimal.

\begin{definition}[Well-formed canonical partition]\label{wellformeddef}
Let $Z \cup \mathcal{A} \cup \mathcal{B} \cup \mathcal{W}$ be a canonical partition of $X$ induced by exclusion set $Z$ and let $\mu$ be a signed Radon measure on $X$ such that $\mu(X) = 0$.  We say that the canonical partition is {\em well-formed with respect to $\mu$} if the following conditions are satisfied:
\begin{enumerate}
	\item\label{canonicalcondition1} $\mu|_Z \preceq_{cvx} 0$ and $\mu|_{\mathcal{W}} \succeq_{2} 0$, and
	\item\label{canonicalcondition2} for all $v \in X$ and all $\epsilon > 0$: 
    \begin{itemize}
      \item $\mu|_\mathcal{A}\left([v_1, v_1+\epsilon] \times [v_2,\infty) \right) \geq 0$, with equality whenever $v_2 = 0$
      \item $\mu|_\mathcal{B}\left([v_1, \infty) \times [v_2, v_2 + \epsilon]\right) \geq 0$, with equality whenever $v_1 = 0$
    \end{itemize}
\end{enumerate}
\end{definition}

We point out the similarities between a well-formed canonical partition and the sufficient conditions for menu optimality of Theorem~\ref{bundlingtheorem}. Condition~\ref{canonicalcondition1} gives exactly the stochastic dominance conditions that need to hold in regions $Z$ and $\mathcal{W}$.
  We interpret Condition~\ref{canonicalcondition2} as saying that $\mu|_\mathcal{A}$ (resp. $\mu|_\mathcal{B}$) allows for the positive mass in any vertical (resp. horizontal) ``strip'' to be matched to the negative mass in the strip by only transporting ``downwards'' (resp. ``leftwards''). These conditions, guarantee (single-dimensional) first order dominance of the  measures along each strip which is stronger requirement than the convex dominance conditions of Theorem~\ref{bundlingtheorem}. In practice, when $\mu$ is given by a density function, we verify these conditions by analyzing the integral of the density function along appropriate vertical or horizontal lines. Even though Theorem~\ref{bundlingtheorem} applies only for mechanisms with finite menus, we prove in Theorem~\ref{canonicalpartitiontheorem} that a mechanism induced by an exclusion set is optimal for a 2-item instance if the canonical partition of its exclusion set is well-formed. Refer back to Figure~\ref{canonicalfig} to visualize such a mechanism.

\begin{theorem}\label{canonicalpartitiontheorem}
Let $\mu$ be the transformed measure of a probability density function $f$. If there exists an exclusion set $Z$ inducing a canonical partition $Z \cup \mathcal{A} \cup \mathcal{B} \cup \mathcal{W}$ of $X$ that is well-formed with respect to $\mu$, then the optimal IC and IR mechanism for a single additive buyer whose values for two goods are distributed according to the joint distribution $f$ is the mechanism induced by exclusion set $Z$. In particular, the mechanism uses the following allocation and price for a buyer with reported type $(x,y) \in X$:
\begin{itemize}
	\item if $(x,y) \in Z$, the buyer receives no goods and is charged $0$;
	\item if $(x,y) \in \mathcal{A}$, the buyer receives item $1$ with probability $-s'_1(x)$, item $2$ with probability $1$, and is charged $s_1(x) - x s_1'(x)$;
	\item if $(x,y) \in \mathcal{B}$, the buyer receives item $2$ with probability $-s'_2(y)$, item $1$ with probability $1$, and is charged $s_2(y) - y s'_2(y)$;
	\item if $(x,y) \in \mathcal{W}$, the buyer receives both goods with probability $1$ and is charged $P$;
\end{itemize}
where $s_1, s_2$ are the boundary functions and $P$ is the critical price  as in Definition~\ref{def:outer boundary}.
\end{theorem}

\begin{proof}
We will show that $u_Z$ maximizes $\sup_{u \in \mathcal{U}(X) \cap \mathcal{L}_1(X)} \int_X u d\mu$.
By Corollary~\ref{linearintegral}, it suffices to provide a $\gamma \in \Gamma_+(X \times X)$ such that $\gamma_1 - \gamma_2 \succeq_{cvx} \mu$, $\int u_Z d(\gamma_1 - \gamma_2) = \int u_Z d\mu$, and $u_Z(x)  - u_Z(y) = \Lone{x-y}$ holds $\gamma$-almost surely.
The $\gamma$ we construct will never transport mass between regions. That is, $\gamma = \gamma_Z + \gamma_\mathcal{W} + \gamma_{\mathcal{A}} + \gamma_{\mathcal{B}}$ where\footnote{We chose this notation for simplicity, where $\gamma_Z \in \Gamma_+(Z \times Z)$, $\gamma_\mathcal{W} \in \Gamma_+(\mathcal{W} \times \mathcal{W})$, and so on.}
\begin{itemize}
	\item $\gamma_Z = 0$. We notice that $(\gamma_Z)_1 - (\gamma_Z)_2 = 0 \succeq_{cvx} \mu|_Z$.
	\item $\gamma_\mathcal{W}$ is constructed such that $(\gamma_\mathcal{W})_1-(\gamma_\mathcal{W} )_2 \succeq_{cvx} \mu|_\mathcal{W}$ and the component-wise inequality $x \geq y$ holds $\gamma_\mathcal{W}(x,y)$ almost surely.\footnote{As in Example~\ref{unifexample} and as discussed in Remark~\ref{geometricremark}, we aim for $\gamma_\mathcal{W}$ to transport ``downwards and leftwards'' since both items are allocated with probability 1 in $\mathcal{W}$.} As in our proof of Theorem~\ref{bundlingtheorem}, the existence of such a $\gamma_\mathcal{W}$ is guaranteed by Strassen's theorem  for second order dominance (presented in the appendix).

	\item {$\gamma_\mathcal{A} \in \Gamma_+(\mathcal{A} \times \mathcal{A})$ will be constructed to have respective marginals $\mu_+|_\mathcal{A}$ and $\mu_-|_\mathcal{A}$, and so that, $\gamma_\mathcal{A}(x,y)$ almost surely, it holds that $x_1 = y_1$ and $x_2 \geq y_2$. Thus, $(\gamma_\mathcal{A})_1 - (\gamma_\mathcal{A})_2 = \mu|_\mathcal{A}$, and $\gamma_\mathcal{A}$ sends positive mass ``downwards.''\footnote{{Once again, the intuition for this construction follows Remark~\ref{geometricremark}.}} We claim that such a map can indeed be constructed, by noticing that Property~\ref{canonicalcondition2} of Definition~\ref{wellformeddef} guarantees that, restricted to any vertical strip inside ${\cal A}$, $\mu_+$ first-order stochastically dominates $\mu_-$.\footnote{Indeed, as $\epsilon \rightarrow 0$, 
  Property~\ref{canonicalcondition2} states exactly the one-dimensional equivalent condition for first-order stochastic dominance in terms of cumulative density functions.} Hence, Strassen's theorem for first-order dominance guarantees that restricted to that strip $\mu_+$ can be coupled with $\mu_-$ so that, with probability $1$, mass is only moved downwards.}
	
Measure $\gamma_{\cal A}$ satisfies $x_1 = y_1$, $\gamma_\mathcal{A}(x,y)$ almost surely,  and hence also
$$u_Z(x) - u_Z(y) = (x_2 - s(x_1)) - (y_2 - s(y_1)) = x_2 - y_2 = \Lone{x-y}.$$
	
	\item $\gamma_\mathcal{B} \in \Gamma_+(\mathcal{B} \times \mathcal{B})$ is constructed analogously to $\gamma_\mathcal{A}$, except sending mass ``leftwards.'' That is, $\gamma_\mathcal{B}(x,y)$ almost-surely, the relationships $x_1 \geq y_1$ and $x_2 = y_2$ hold.
\end{itemize}
It follows by our construction that $\gamma =  \gamma_Z + \gamma_\mathcal{W} + \gamma_{\mathcal{A}} + \gamma_{\mathcal{B}}$ satisfies all necessary properties to certify optimality of $u_Z$.
\end{proof}
 
\section{Applying Theorem~\ref{canonicalpartitiontheorem} to find optimal mechanisms}
\label{sec:further examples}

In this section, we provide example applications of Theorem~\ref{canonicalpartitiontheorem}. A technical difficulty is verifying the stochastic dominance relation  $\mu|_\mathcal{W} \succeq_2 0$ required to apply the theorem. {In our examples, we will have the stronger condition $\mu|_\mathcal{W} \succeq_1 0$, which is easier to verify, yet still imposes technical difficulties.}  In Section~\ref{verifyingfirstorder}, we present a useful tool, Lemma~\ref{regionthm}, for verifying first-order stochastic dominance. In Section~\ref{weakexamples} we then provide example applications of Theorem~\ref{canonicalpartitiontheorem} and Lemma~\ref{regionthm} to solve for optimal mechanisms.

\subsection{Verifying First-Order Stochastic Dominance }\label{verifyingfirstorder}

A useful tool for verifying first order dominance between measures is the following.\footnote{The lemma also appeared as Theorem 7.4 of \cite{DaskalakisDT13} without a proof. We provide a detailed proof in the online appendix.}

\begin{lemma}\label{regionthm}
Let $\mathcal{C} = [p_1, q_1) \times [p_2, q_2)$ where $q_1$ and $q_2$ are possibly infinite and let $R$ be a decreasing nonempty subset of $\mathcal{C}$. Consider two measures $\kappa, \lambda \in \Gamma_+(\mathcal{C})$ 
 with bounded integrable density functions $g, h: \mathcal{C} \rightarrow \mathbb{R}_{\geq 0}$ respectively that satisfy the conditions:
\begin{itemize}
	\item $g(x,y) = h(x,y) = 0$ for all $(x,y) \in R$.
	\item $\int_{\mathcal{C}}g(x,y)dxdy = \int_{\mathcal{C}}h(x,y)dxdy$.
	\item For any basis vector $e_i \in \{ e_1 \equiv (1,0), e_2 \equiv (0,1) \}$ and any point $z \in R$:
	$$\int_0^{q_i - z_i} g(z + \tau e_i) - h(z + \tau e_i)d\tau \leq 0.$$
	\item There exist non-negative functions $\alpha: [p_1, q_1) \rightarrow \mathbb{R}_{\geq 0}$ and $\beta : [p_2, q_2) \rightarrow \mathbb{R}_{\geq 0}$, and an increasing function $\eta : \mathcal{C} \rightarrow \mathbb{R}$ such that for all $(x,y) \in \mathcal{C} \setminus R$:
	$$g(x,y) - h(x,y) = \alpha(x)\cdot \beta(y) \cdot \eta(x,y)$$
\end{itemize}
Then $\kappa \succeq_1 \lambda$.
\end{lemma}

Lemma~\ref{regionthm} provides a sufficient condition for a measure to stochastically dominate another in the first order. Its proof is given in the online appendix and is an application of a claim which states that an equivalent condition for first-order stochastic dominance is that one measure has more mass than the other on all sets that are  unions of \textit{finitely many} ``increasing boxes.'' When the conditions of Lemma~\ref{regionthm} are satisfied, we can induct on the number of boxes by removing one box at a time. {We note that Lemma~\ref{regionthm} is applicable even to distributions with unbounded support.}

\begin{mdframed}[style=boxed]
{\bf Interpreting the Conditions of Lemma~\ref{regionthm}:} 
Lemma~\ref{regionthm} is applicable whenever two {density functions,} $g$ and $h$, are nonzero  on some set $\mathcal{C} \setminus R$, where $R$ is a decreasing subset of some two-dimensional box $\mathcal{C}$. This setting is motivated by Figure~\ref{canonicalfig} and Theorem~\ref{canonicalpartitiontheorem}. Recall that, in order to apply Theorem~\ref{canonicalpartitiontheorem}, we need to check a second order stochastic dominance condition in region $\cal W$, namely $\mu|_\mathcal{W} \succeq_2 0$. 

While Theorem~\ref{canonicalpartitiontheorem} demands checking a second order stochastic dominance condition, an easier and sufficient goal is to check first order stochastic dominance, namely $\mu|_\mathcal{W} \succeq_1 0$. To do this, we can readily use Lemma~\ref{regionthm}, by taking $\mathcal{C} = [x_{\textrm{crit}}, \infty) \times [y_{\textrm{crit}}, \infty)$, $R = \mathcal{C} \cap Z$, and $g, h$ the densities corresponding to measures $\mu_+|_\mathcal{W}$ and $\mu_-|_\mathcal{W}$. The way region $\mathcal{W}$ is defined in Theorem~\ref{canonicalpartitiontheorem} guarantees that the two measures have equal mass, so the first two conditions of the lemma will be satisfied automatically. For the third condition, we need to verify that, if we integrate $g-h$ along either a vertical or a horizontal line outwards starting from any point in $R$, the result is \emph{non-positive}. 
The last condition of Lemma~\ref{regionthm} requires that the density function of the measure $\mu|_{\cal W}$, i.e. $g-h$, have an appropriate form. If the  values of the buyer for the two items are independently distributed according to distributions with densities $f_1$ and $f_2$, then the density of measure $\mu$ in the interior according to Equation~\ref{transformed} can be written as $- f_1(x) f_2(y) \left( \frac {f'_1(x) x} {f_1(x)} + \frac {f'_2(y) y} {f_2(y)} + 3\right)$. The last condition of the lemma is thus satisfied if the functions $\frac {f'_1(x) x} {f_1(x)}$ and $\frac {f'_2(y) y} {f_2(y)}$ are decreasing, a condition that is easy to verify.
\end{mdframed}

\subsection{Examples }\label{weakexamples}

We apply Theorem~\ref{canonicalpartitiontheorem} to obtain optimal mechanisms in several two-item settings. In Section~\ref{sec:beta}, we consider two independent items distributed according to beta distributions. We find the optimal mechanism, showing that it actually offers an uncountably infinite menu of lotteries. 
We conclude with Section~\ref{sec:examples unbounded support} where we discuss extensions of Theorem~\ref{canonicalpartitiontheorem} to distributions with infinite support, providing the optimal mechanism for two arbitrary independent exponential items, as well as the optimal mechanism for an instance with two independent power-law items.

\subsubsection{An Optimal Mechanism with Infinite Menu Size: Two Beta Items}\label{sec:beta}

In this section, we will use Theorem~\ref{canonicalpartitiontheorem} to calculate the optimal mechanism for two items distributed according to Beta distributions. In doing so we illustrate a general approach for finding closed-form descriptions of optimal mechanisms  via the following steps:
{\tt (i)} definition of the sets $S_{\rm top}$ and $S_{\rm right}$, {\tt (ii)} computation of a critical price $p^*$, {\tt (iii)} definition of a canonical partition in terms of {\tt (i)} and {\tt (ii)}, and {\tt (iv)} application of Theorem~\ref{canonicalpartitiontheorem}. Our approach succeeds in pinning down optimal mechanisms in all examples considered in Sections~\ref{sec:beta}---\ref{sec:examples unbounded support}, and we expect it to be broadly applicable. Finally, it is noteworthy that the optimal mechanism for the setting studied in this section offers the buyer a menu of uncountably infinitely many lotteries to choose from. Using our approach we can nevertheless compute and succinctly describe the optimal mechanism. We also note in Remark~\ref{remark:uniqueness of Beta mechanism} that our identified mechanism is essentially unique, hence the uncountability of the menu is inevitable.

Consider two items whose values are distributed independently according to the distributions $\textrm{Beta}(a_1,b_1)$ and $\textrm{Beta}(a_2,b_2)$, respectively. That is, the distributions are given by to the following two density functions on $[0,1]$:
$$
f_1(x) = \frac{1}{B(a_1,b_1)} x^{a_1-1}(1-x)^{b_1-1}; \qquad
f_2(y) = \frac{1}{B(a_2,b_2)}y^{a_2-1}(1-y)^{b_2-1}.
$$
To find the optimal mechanism for our example setting, we first compute the measure $\mu$ induced by $f$. Notice that
\begin{align*}
-\nabla f(x,y) \cdot (x,y) &- 3  f(x,y) = -xf_2(y)\frac{ \partial f_1(x)}{\partial x} - yf_1(x)\frac{ \partial f_2(y)}{\partial y}  -3f_1(x)f_2(y)\\
&= -(a_1-1)f_1(x)f_2(y) + (b_1-1)\frac{x}{1-x}f_1(x)f_2(y)\\
&\quad -(a_b-1)f_1(x)f_2(y) + (b_2-1)\frac{y}{1-y}f_1(x)f_2(y) - 3f_1(x)f_2(y)\\
&= f_1(x)f_2(y)\left(\frac{b_1-1}{1-x}+\frac{b_2-1}{1-y} +(1-a_1-b_1-a_2-b_2) \right)
\end{align*}
where the last equality used the identity $\frac{x}{1-x} = \frac{1}{1-x}-1$.
We also observe that $f_1(x)x = 0$ whenever $x=0$ or $x=1$ (as long as $b_1>1$), and an analogous property holds for $y$. Thus, the transformed measure $\mu$ is comprised of:
\begin{itemize}
	\item a point mass of +1 at the origin; and
	\item mass distributed on $[0,1]^2$ according to the density function
	$$f_1(x)f_2(y) \left(\frac{b_1-1}{1-x}+\frac{b_2-1}{1-y} +(1-a_1-b_1-a_2-b_2) \right).$$
\end{itemize}
Note that in the case $b_i=1$, our analysis still holds, except there is also positive mass on the boundary $x_i=1$.

\paragraph{Deriving the Optimal Mechanism for a Concrete Setting of Parameters.} 
We now analyze a concrete example of two independent Beta distributed items where $a_1=a_2=1$ and $b_1=b_2=2$. That is, we consider two items whose values are distributed independently according to the following two density functions on $[0,1]$:
$$
f_1(x) = 2 (1-x); \qquad
f_2(y) = 2 (1-y).
$$

As discussed above, the transformed measure $\mu$ comprises:
\begin{itemize}
	\item a point mass of +1 at the origin; and
	\item mass distributed on $[0,1]^2$ according to the density function
	$$f_1(x)f_2(y) \left(\frac{1}{1-x} + \frac{1}{1-y}  - 5 \right).$$
\end{itemize}
{Note that the density of $\mu$ is positive on $\mathcal{P} = \left\{ (x,y) \in (0,1)^2 :\frac{1}{1-x} + \frac{1}{1-y} > 5\right\} \cup \{ \vec{0}\}$ and non-positive on $ \mathcal{N} = \left\{ (x,y) \in [0,1)^2 \setminus \{\vec{0}\}:\frac{1}{1-x} + \frac{1}{1-y} \leq 5  \right\}$, and that $\mathcal{N} \cup \{\vec{0}\}$ is a decreasing set.}

\noindent \textbf{Step (i).} We first attempt to identify candidate functions for $s_1$ and $s_2$ that will lead to a well-formed canonical partition. We do this by defining two sets $S_{\textrm{top}},S_{\textrm{right}} \subset [0,1)^2$. We require that $(x,y) \in S_{\textrm{top}}$ iff $ \int_{y}^1 \mu(x,t) dt = 0.$
That is, starting from any point $z\in S_{\textrm{top}}$ and integrating the density of $\mu$
``upwards'' from $t=y$ to $t=1$ yields zero. Since ${\cal N} \cup \{\vec{0}\}$ is a decreasing set, it follows that $S_{\textrm{top}} \subset \mathcal{N}$ and that integrating $\mu$ upwards starting from any point above $S_{\rm top}$ yields a positive integral.
Similarly, we say that $(x,y) \in S_{\textrm{right}}$ iff
$ \int_{x}^1 \mu(t,y) dt = 0,$ noting that $S_{\textrm{right}} \subset {\cal N}$. $S_{\textrm{top}}$ and $S_{\textrm{right}}$ are shown in Figure~\ref{betafig}.

We analytically compute  that $(x, y) \in  S_{\textrm{top}}$ if and only if
$y = \frac{2 - 3 x}{4 - 5 x}$.
Similarly, $(x,y) \in S_{\textrm{right}}$ if and only if
$x = \frac{2 - 3 y}{4 - 5 y}.$ 

In particular, for any $x \le 2/3$ there exists a $y$ such that $(x,y) \in   S_{\textrm{top}}$, and there does not exist such a $y$ if $x > 2/3$. Furthermore, it is easy to verify by computing the second derivative of $\frac {\partial^2}{\partial x^2} \frac{2 - 3 x}{4 - 5 x} = -\frac{20}{(4 - 5 x)^3} < 0$ that the region below $S_{top}$ and the region below $S_{right}$ are strictly convex.

\medskip
\noindent \textbf{Step (ii).} We now need to calculate the critical point and the critical price. To do this we set the critical price $p^* \approx .5535$  as the intercept of the $45^\circ$ line in Figure~\ref{betafig} which causes $\mu(Z) = 0$ for the set $Z \subset [0,1]^2$ lying below $S_{\rm top}$, $S_{\rm right}$ and the $45^\circ$ line. 
We can also compute the critical point $(x_{\textrm{crit}},y_{\textrm{crit}}) \approx (.0618,.0618)$ by finding the intersection of the critical price line with the sets $S_{\rm top}$ and $S_{\rm bottom}$. Moreover, by the definition of the sets $S_{\rm top}$ and $S_{\rm bottom}$, we know that the candidate boundary functions are $s_1(x) =\frac{2 - 3 x}{4 - 5 x}$ and $s_2(y) = \frac{2 - 3 y}{4 - 5 y}$, with domain $[0,x_{\textrm{crit}})$ and $[0,y_{\textrm{crit}})$ respectively.

\medskip
\noindent \textbf{Step (iii).}
We can now compute the canonical partition and decompose $[0,1]^2$ into the following regions:
$$\mathcal{A} =  \{(x,y): x \in [0,x_{\textrm{crit}}) \textrm{ and } y \in [s_1(x),1]\}; \mathcal{B} =  \{(x,y): y \in [0,y_{\textrm{crit}}) \textrm{ and } x \in [s_2(y),1]\}$$
$$\mathcal{W} =  \{(x,y) \in [x_{\textrm{crit}},1] \times [y_{\textrm{crit}},1]: x+y \ge p^*\}; \,\, Z =  [0,1]^2 \setminus \left( \mathcal{W} \cup \mathcal{A} \cup \mathcal{B}  \right)$$
as illustrated in Figure~\ref{betafig}.
\begin{figure}[!ht]
\centering
\begin{tikzpicture}
\begin{axis}[height=3.2in, width=3.2in, ymin=0, ymax=1, xmin=0, xmax=1,
  xtick pos=left, xtick={0.5,0.6666666}, xticklabels={$1/2$, $2/3$}, ytick pos=left, ytick={0.5,0.6666666666}, yticklabels={$1/2$,$2/3$}]
  
\addplot+[color=black, mark=none, style=densely dashed, domain=0:0.66666666666] {(2-3*x)/(4-5*x)};
\addplot+[color=black, mark=none, style=dotted, domain=0:0.5] {(2-4*x)/(3-5*x)};
\addplot[color=black, fill=gray!50, mark=none, samples=200, domain=0:0.5] 
{min( (2-3*x)/(4-5*x), (2-4*x)/(3-5*x), 0.5534938-x)}\closedcycle;

\legend{$S_{\textrm{top}}$, $S_{\textrm{right}}$};

\addplot[color=gray, mark=none]coordinates{
(.5534938-0.06187679,0.06187679)
(1,0.06187679)
};
\addplot[color=gray, mark=none]coordinates{
(0.06187679,.5534938-0.06187679)
(0.06187679,1)
};
\node at (axis cs:0.03,0.85){$\mathcal{A}$};
\node at (axis cs:0.85,0.03){$\mathcal{B}$};
\node at (axis cs:0.6,0.65){$\mathcal{W}$};
\node at (axis cs:0.2,0.2){$Z$};
\end{axis}
\end{tikzpicture}
\caption{The well-formed canonical partition for $ f_1(x) = 2 (1-x)$ and $f_2(y) = 2 (1-y)$.}\label{betafig}
\end{figure}
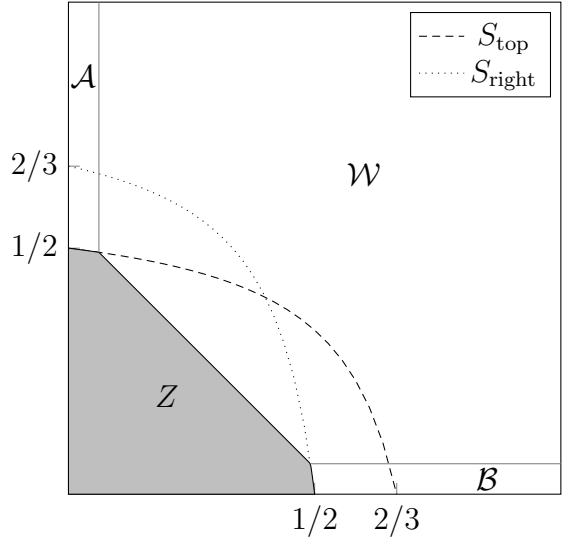

\noindent \textbf{Step (iv).}
We claim that the canonical partition $Z \cup \mathcal{A} \cup \mathcal{B} \cup \mathcal{W}$ is well-formed with respect to $\mu$. Condition~\ref{canonicalcondition2} is satisfied by construction of $S_{\rm top}$ and $S_{\rm right}$ and the corresponding discussion in Step (i). To check for Condition~\ref{canonicalcondition1}, note that given the definition of $p^*$, it holds that for all regions $R = Z, \mathcal{A}, \mathcal{B} \textrm{ and } \mathcal{W}$, we have $\mu(R) = 0$. Recall that $S_{\textrm{top}}, S_{\textrm{right}} \subset \mathcal{N}$ and, since ${\cal N} \cup \{\vec{0}\}$ is a decreasing set, $\mu$ has negative density along these curves and all points below either curve, other than at the origin. Hence, $\mu_-|_Z \succeq_{1} \mu_+|_Z$ which implies that $\mu|_Z \preceq_{cvx} 0$.  Hence, the only non-trivial condition of Definition~\ref{wellformeddef}  that we need to verify is  $\mu|_\mathcal{W} \succeq_2 0$. In fact, we can apply Lemma~\ref{regionthm}  to conclude the stronger dominance relation $\mu|_\mathcal{W} \succeq_1 0$. See the online appendix. Having verified all conditions of Definition~\ref{wellformeddef} we apply Theorem~\ref{canonicalpartitiontheorem} to conclude the following.

\begin{example}\label{betaexample}
The optimal mechanism for selling two independent items whose values are distributed according to $ f_1(x) = 2 (1-x)$ and  $f_2(y) = 2 (1-y)$  has the following outcome for a buyer of type $(x,y)$:
\begin{itemize}
	\item If $(x,y) \in Z$, the buyer receives no goods and is charged 0.
	\item If $(x,y) \in \mathcal{A}$, the buyer receives item 1 with probability $-s'_1(x) = \frac{2}{(4-5x)^2}$, item 2 with probability 1, and is charged $s_1(x) - xs'_1(x) = \frac{2 - 3 x}{4 - 5 x} + \frac{2 x}{(4-5x)^2}$.
	\item If $(x,y) \in \mathcal{B}$, the buyer receives item 2 with probability $-s'_2(y) = \frac{2}{(4-5y)^2}$, item 1 with probability 1, and is charged $s_2(y) - ys'_2(y) = \frac{2 - 3 y}{4 - 5 y} + \frac{2 y}{(4-5y)^2}$.
	\item If $(x,y) \in \mathcal{W}$, the buyer receives both items and is charged $p^* \approx .5535$.
\end{itemize}
\end{example}

\begin{remark} \label{remark:uniqueness of Beta mechanism} Note that the mechanism identified in Example~\ref{betaexample} offers an uncountably large menu of lotteries. One could wonder whether there exists a different optimal mechanism offering a finite menu. Using our duality theorem we can easily argue that the utility function induced by every optimal mechanism equals the utility function $u(x)$ induced by our mechanism in Example~\ref{betaexample}. Hence, up to the choice of subgradients at the measure-zero set of types where $\nabla u(x)$ is discontinuous, the allocations offered by any optimal mechanism must agree with those of our mechanism in Example~\ref{betaexample}. Therefore, every optimal mechanism must offer an uncountably large menu. The proof of uniqueness is given in the online appendix.

\end{remark}

\paragraph{Summary of Beta Distributions.}
Example~\ref{betaexample} shows that the optimal mechanism for two Beta distributed items offers a continuum of lotteries, thereby having infinite menu-size complexity \cite{HartN13}. Still, using our techniques we can obtain a succinct and easily-computable description of the mechanism.

Working similarly to Example~\ref{betaexample}, we can obtain the optimal mechanism for broader settings of parameters. Figure~\ref{fig:Betas} illustrates the optimal mechanism for two items distributed according to Beta distributions with different parameters. The reader can experiment with different settings of parameters at~\cite{BetaLink}.

\begin{figure}[htbp]
	\centering
		\includegraphics[width=0.8\textwidth]{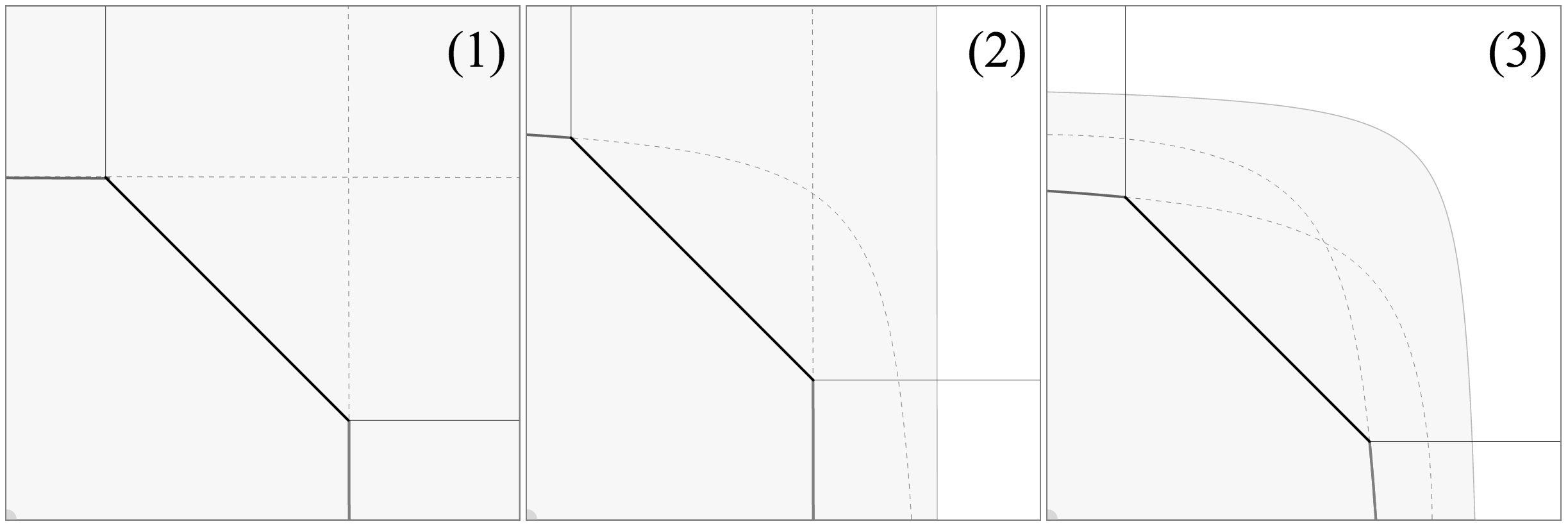}

    \caption{Canonical Partitions for different cases of Beta distributions. The shaded region is where the measure $\mu$ becomes negative. (Note that when the second parameter $b_i$ of the Beta distribution of some item $i$ equals $1$, $\mu$ has positive mass on the outer boundary $x_i=1$.) (1) Beta(1,1) and Beta(1,1), (2) Beta(2,2) and Beta(1,1), (3) Beta(2,2) and Beta(2,2).}
	\label{fig:Betas}
\end{figure}

\subsubsection{Distributions of Unbounded Support: Exponential and Power-Law} \label{sec:examples unbounded support}

So far, this paper has focused on type distributions with bounded support. In this section, we note that Theorem~\ref{setupclaim}, Lemma~\ref{weakduality}, and Theorem~\ref{canonicalpartitiontheorem} can be easily modified to accommodate settings with unbounded type spaces, as long as the type distribution decays sufficiently rapidly towards infinity. On the other hand, we do not know extensions of our strong duality theorem (Theorem~\ref{strongduality}), and the optimal menu conditions (Theorem~\ref{bundlingtheorem}) for unbounded type distributions, due to technical issues.

In the online appendix, we provide a short discussion of the  modifications required to obtain an analog of Theorem~\ref{canonicalpartitiontheorem} for unbounded distributions that are sufficiently fast-decaying, and present below two example settings that can be analyzed using the modified characterization theorem. Both examples are taken from~\cite{DaskalakisDT13}.

In Example~\ref{powerexample}, the optimal mechanism for selling two power-law items is a grand bundling mechanism. The canonical partition induced by the exclusion set of the grand-bundling mechanism is degenerate (regions $\mathcal{A}$ and $\mathcal{B}$ are empty), and establishing the optimality of the mechanism amounts to establishing that the first-order stochastic dominance condition for the induced measure $\mu$ holds in region $\cal W$.

\begin{example}\label{powerexample}
The optimal IC and IR mechanism for selling two items whose values are distributed independently according to the probability densities $f_1(x) = 5/(1+x)^6$ and $f_2(y) = 6/(1+y)^7$ respectively is a take-it-or-leave-it offer of the bundle of the two goods for price $p^* \approx .35725$.
\end{example}

Example~\ref{twoitems} provides a complete solution for the optimal mechanism for two items distributed according to independent exponential distributions. In this case, the canonical partition induced by the exclusion set of the mechanism is missing region $\cal A$, and possibly region $\cal B$ (if $\lambda_1=\lambda_2$).

\begin{example}\label{twoitems}
For all $\lambda_1 \geq \lambda_2 > 0$, the optimal IC and IR mechanism for selling two items  whose values are distributed independently according to exponential distributions $f_1$ and $f_2$ with respective parameters $\lambda_1$ and $\lambda_2$ offers the following menu:
\begin{enumerate}
\item receive nothing, and pay 0;
\item receive the first item with probability 1 and the second item with probability $\lambda_2/\lambda_1$, and pay $2/\lambda_1$; and
\item receive both items, and pay $p^*$;
\end{enumerate}
where $p^*$ is the unique $0<p^* \leq 2/\lambda_2$ such that
$$\mu(\left\{ (x,y) \in \mathbb{R}^2_{\geq 0} :  x+y \leq p^* \textrm{ and }  \lambda_1 x+\lambda_2 y \leq 2 \right\}) = 0,$$
where $\mu$ is the transformed measure of the joint distribution. 
\end{example}

\begin{figure}[!ht]
\centering
\begin{tikzpicture}
\begin{axis}[width=3.1in, height=3.1in, ymin=0, ymax=4, xmin=0, xmax=4, xtick pos=left, xtick={.66666,1,4}, xticklabels={$\frac{2}{\lambda_1}$, $\frac{3}{\lambda_1}$,$x$}, ytick pos=left, ytick={2,3, 1.4,4}, yticklabels={$\frac{2}{\lambda_2}$,$\frac{3}{\lambda_2}$, $p^*$,$y$}]
\addplot+[color=black,  fill=gray!20, mark=none, domain=0:1.5, samples=2]
{3 - 3*x}
\closedcycle;
 \addplot+[color=gray, fill=gray, domain=.3:1.5, mark=none]
 {2-3*x}
 \closedcycle;
  \addplot+[color=gray, fill=gray, domain=0:.3,mark=none]
 {1.4-x}
 \closedcycle;
\addplot+[color=black, domain=0:1,samples=2,  mark=none, dotted]
 {2 - 3*x} 
 \closedcycle;
\addplot[color=black, mark=none] coordinates{
(.3,1.1)
(6,1.1)}; 
\node at (axis cs:.2,.5) {$Z_{p^*}$};
\node (lab1) at (axis cs:1.2,.9){$\mathcal{B} \cap \mathcal{N}$};
\node (dest) at (axis cs:.6,.4){};
\draw[->] (axis cs: 1.2,.75)--(dest);
\node at (axis cs: 2.3,.5){$\mathcal{B} \cap \mathcal{P}$};
\node at (axis cs:2.5,2.8){$\mathcal{W} \cap \mathcal{P}$};
\node at (axis cs:1,2.2){$\mathcal{W} \cap \mathcal{N}$};
\node (lab2) at (axis cs:1.08,2.1){};
\node (dest2) at (axis cs:0,1.35){};
\draw[->] (lab2)--(dest2);
\node (dest3) at (axis cs:.3, 1.3){};
\node (lab3) at (axis cs:1.08,2.13){};
\draw[->] (lab3)--(dest3);
\node (dest4) at (axis cs: .017, 1.68){};
\node (lab4) at (axis cs: 1,3.2){};
\end{axis}
\end{tikzpicture} 
\caption{The canonical partition of $\mathbb{R}^n_{\geq 0}$ for the proof of Example~\ref{twoitems}. In this diagram, $p^* > 2/\lambda_1$. If $p^* \leq 2/\lambda_1$, $\mathcal{B}$ is empty. The positive part $\mu_+$ of $\mu$ is supported inside $\mathcal{P} \cap \{\vec{0}\}$ while the negative part $\mu_-$ is supported within $Z_{p^*} \cup \mathcal{N}$.}\label{fig: exp items first}
\end{figure}
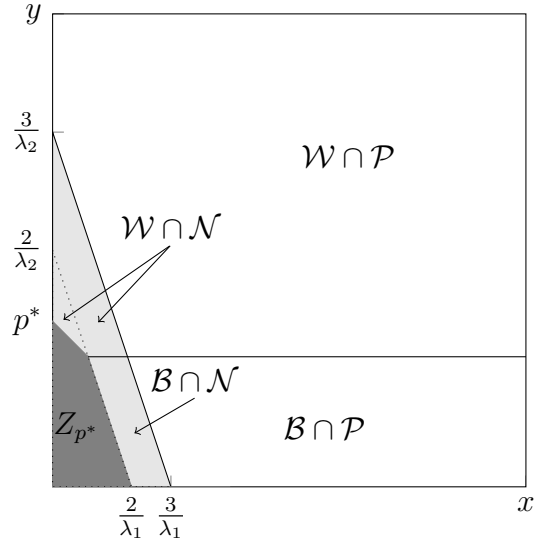

\section{Conclusions}
We provided a duality-based framework for revenue maximization in a multiple-good monopoly. Our framework shows that every optimal mechanism has a certificate of optimality, taking the form of an optimal transportation map between measures. Using this framework, we characterized optimal mechanisms, showing that a mechanism is optimal if and only if certain stochastic dominance conditions are satisfied by a measure induced by the buyer's type distribution. This measure expresses  the marginal change in the seller's revenue under marginal changes in the rent paid to subsets of buyer types.

We also provided several tools for checking the pertinent stochastic dominance conditions in two dimensions. These tools were useful in establishing the optimality of mechanisms in a multitude of two-item examples that we studied. While our characterization holds for an arbitrary number of items, verifying stochastic dominance in higher dimensions becomes significantly harder. An interesting future direction is to develop tools for checking stochastic dominance  in higher dimensions. This will be useful for establishing optimality of mechanisms for three and more items.

Another important research direction is to obtain conditions for the type distribution under which the optimal mechanism has a simple closed-form description. For example, are there broad conditions implying that grand bundling is optimal or that the optimal mechanism takes the form of the mechanisms in Theorem~\ref{canonicalpartitiontheorem}?

Finally a major open problem is to extend our results to multiple bidders. Even for the presumably simple setting of two bidders with independent and identical values for two items that are uniformly distributed in $[0,1]$, the revenue-optimal mechanism is unknown.

\small
\bibliographystyle{alpha}
\bibliography{econometricabib}
\normalsize

\newpage
\appendix
\section*{\centering Strong Duality for a Multiple-Good Monopolist}
\subsection*{\centering Online Appendix}

\section{Strong Mechanism Design Duality - Proof of Theorem~\ref{strongduality}}\label{proofsection}

In this section, we give a formal proof of the strong mechanism duality theorem. To carefully prove the statement, we specify that the proof is for Radon measures. A Radon measure is a locally-finite inner-regular Borel measure. We use $\Gamma(X) = Radon(X)$ (resp. $\Gamma_+(X) = Radon_+(X)$) as the set of signed (resp. unsigned) Radon measures on $X$. The transformed measure of a distribution is always a signed Radon measure as it defines a bounded linear functional on the utility function $u$.\footnote{More formally, this follows from Riesz representation theorem} 

\subsection{A Strong Duality Lemma}

The overall structure of our proof of Theorem~\ref{strongduality} is roughly parallel to the proof of Monge-Kantorovich duality presented in \cite{Villani}, although the technical aspects of our proof are different, mainly due to the added convexity constraint on $u$.
We begin by stating the Legendre-Fenchel transformation and the Fenchel-Rockafellar duality theorem.

\begin{definition}[Legendre-Fenchel Transform]
Let $E$ be a normed vector space and let $\Lambda: E \rightarrow \mathbb{R} \cup \{+\infty\}$ be a convex function. The {\em Legendre-Fenchel transform of $\Lambda$}, denoted $\Lambda^*$, is a map from the topological dual $E^*$ of $E$ to $\mathbb{R} \cup \{\infty\}$ given by
$$\Lambda^*(z^*) = \sup_{z \in E}\left( \langle z^*, z \rangle - \Lambda(z) \right).$$
\end{definition}

\begin{claim}[Fenchel-Rockafellar duality]\label{thm:frdual}
Let $E$ be a normed vector space, $E^*$ its topological dual, and $\Theta, \Xi$ two convex functions on $E$ taking values in $\mathbb{R} \cup \{ + \infty \}$. Let $\Theta^*, \Xi^*$ be the Legendre-Fenchel transforms of $\Theta$ and $\Xi$ respectively. Assume that there exists $z_0 \in E$ such that $\Theta(z_0) < + \infty$, $\Xi(z_0) < + \infty$ and $\Theta$ is continuous at $z_0$. Then
$$\inf_{z \in E} [\Theta(z) + \Xi(z)] = \max_{z^* \in E^*} [-\Theta^*(-z^*)-\Xi^*(z^*)].$$ 
\end{claim}

\begin{lemma}\label{longlemma} Let $X$ be a compact convex subset of $\mathbb{R}^n$, and let $\mu \in \Gamma(X)$ be such that $\mu(X) = 0$. Then
$$\inf_{\substack{\gamma \in \Gamma_+(X \times X)\\ \gamma_1 \succeq_{cvx} \mu_+ \\ \gamma_2 \preceq_{cvx} \mu_-}}\int_{X \times X} \Lone{x-y} d\gamma(x,y) = \sup_{\substack{ \phi, \psi \in \mathcal{U}(X) \\ \phi(x)-\psi(y) \leq \Lone{x-y} } } \left( \int_X \phi d \mu_+ - \int_X \psi d\mu_- \right) $$
and the infimum on the left-hand side is achieved.
\end{lemma}

\begin{prevproof}{Lemma}{longlemma}
We will apply Fenchel-Rockafellar duality with $E = CB(X \times X)$, the space of continuous (and bounded) functions on $X \times X$ equipped with the $\Linf{\cdot}$ norm. Since $X$ is compact, by the Riesz representation theorem $E^* = \Gamma(X \times X)$.

We now define functions $\Theta, \Xi$ mapping $CB(X \times X)$ to $\mathbb{R} \cup \{+\infty\}$ by
\begin{align*}
\hspace{-.3in}\Theta(f) &= \left\{
	\begin{array}{ll}
		0  & \mbox{if } f(x,y) \geq - \Lone{x-y} \textrm{ for all } x,y \in X\\
		+ \infty & \text{otherwise}
	\end{array}
\right.\\
\hspace{-.3in}\Xi(f) &= \left\{
	\begin{array}{ll}
		\int_X \psi d\mu_- - \int_X \phi d\mu_+   & \mbox{if } f(x,y) = \psi(y)-  \phi(x) \text{ for some } \psi, \phi \in \mathcal{U}(X)\\
		+ \infty & \text{otherwise.}
	\end{array}
\right.
\end{align*}

We note that $\Xi$ is well-defined: If $\psi(x) - \phi(y) = {\psi}'(x) - {\phi}'(y)$ for all $x, y \in X$, then $\psi(x)  - {\psi}'(x) =  \phi(y)  - {\phi}'(y)$ for all $x,y \in X$. This means that ${\psi}'$ differs from $\psi$ only by an additive constant, and $\phi$ differs from $\phi'$ by the same additive constant, and therefore (since $\mu_+$ and $\mu_-$ have the same total mass) $\int_X \psi d\mu_- - \int_X \phi d\mu_+ = \int_X {\psi}' d\mu_- - \int_X {\phi}' d\mu_+.$

It is clear that $\Theta(f)$ is convex, since any convex combination two functions for which $f(x,y)\geq - \Lone{x-y}$ will yield another function for which the inequality is satisfied. It is furthermore clear that $\Xi$ is convex, since we can take convex combinations of the $\psi$ and $\phi$ functions as appropriate. (Notice that $\mathcal{U}(X)$ is closed under addition and positive scaling of functions.)

Consider the function $z_0 \in CB(X \times X)$ which takes the constant value of $1$. It is clear that $\Theta(z_0) = 0$ and $\Xi(z_0) = \mu_-(X) < \infty$. Furthermore, $\Theta(z) = 0$ for any $z \in CB(X \times X)$ with $\Linf{z - z_0} < 1$, and therefore $\Theta$ is continous at $z_0$. We can thus apply the Fenchel-Rockafellar duality theorem.

We compute, for any $\gamma \in \Gamma(X \times X)$:
\begin{align*}
\Theta^*(-\gamma) &=  \hspace{-.15in}\sup_{f \in CB(X \times X)} \left[ \int_{X\times X} f(x,y) d(-\gamma(x,y)) \right.  \\
&\left. \hspace{2in} -  {\left\{
	\begin{array}{ll}
		0  & \mbox{if } f(x,y) \geq - \Lone{x-y} \ \forall   x,y \in X \\
		+ \infty & \text{otherwise}
	\end{array}
\right.}  \right]\\
&=  \hspace{-.15in}\sup_{\substack{f \in CB(X \times X)\\ f(x,y) \geq - \Lone{x-y}}} \left(- \int_{X \times X} f(x,y) d\gamma(x,y) \right) =\hspace{-.15in}\sup_{\substack{\tilde{f} \in CB(X \times X)\\ \tilde{f}(x,y) \leq \Lone{x-y}}} \left(  \int_{X \times X} \tilde{f}(x,y) d\gamma(x,y) \right).
\end{align*}

We claim therefore that
\begin{align*}
\Theta^*(-\gamma)  &= \begin{cases}
\int_{X \times X} \Lone{x-y} d\gamma(x,y) & \mbox{if } \gamma \in \Gamma_+(X \times X)\\
\infty & \mbox{otherwise.}
\end{cases}
\end{align*}
Indeed, if $\gamma$ is a positive linear functional, then the result follows from monotonicity, since $\Lone{x-y}$ is the pointwise greatest function $\tilde{f}$ satisfying the constraint $\tilde{f}(x,y) \leq \Lone{x-y}$, and $\Lone{x-y}$ is continuous. Suppose instead that $\gamma$ is a signed Radon measure which is not positive everywhere. Then there exists a continuous nonnegative function $g: X \times X \rightarrow \mathbb{R}$ such that $\int g d\gamma = -\epsilon$ for some $\epsilon > 0$.\footnote{Formally, we have used Lusin's theorem to find such a $g$ which is continuous, as opposed to merely measurable.} Since $g(x,y) \geq 0$, it follows that $-kg(x,y) \leq 0 \leq \Lone{x-y}$ for any $k \geq 0$. Therefore 
$$\sup_{\substack{\tilde{f} \in CB(X \times X)\\ \tilde{f}(x,y) \leq \Lone{x-y}}} \left(  \int_{X \times X} \tilde{f}(x,y) d\gamma(x,y) \right) \geq \int -kg(x,y)d\gamma(x,y) = k \epsilon.$$
The claim follows, since $k > 0$ is arbitrary.

We similarly compute, for any $\gamma \in \Gamma(X \times X)$:
\begin{align*}
\Xi^*&(\gamma) = \sup_{f \in CB(X \times X)} \left[ \int_{X\times X} f(x,y) d\gamma(x,y) - \right. \\
&\hspace{10pt} - \left. {\left\{
	\begin{array}{ll}
		\int_X \psi d\mu_- - \int_X \phi d\mu_+  & \mbox{if } f(x,y) = \psi(y)-\phi(x) \text{ and } \psi,\phi \in \mathcal{U}(X)\\
		+ \infty & \text{otherwise}
	\end{array}
\right. } \right] \\
&= \sup_{\psi, \phi \in \mathcal{U}(X)} \left[ \int_{X\times X} (\psi(y)-\phi(x)) d\gamma(x,y) - 
		\int_X \psi d\mu_- + \int_X \phi d\mu_+ \right]
\end{align*}
We notice that $\Xi^*(\gamma) \geq 0$ for all $\gamma \in \Gamma(X \times X)$ by setting $\psi=\phi = 0$ and thus $\Theta^*(-\gamma) + \Xi^*(\gamma) = \infty$ if $\gamma \not \in \Gamma_+(X\times X)$. Moreover, when $\gamma \in \Gamma_+(X\times X)$:
\begin{align*}
\Xi^*(\gamma)  &= \sup_{\psi, \phi \in \mathcal{U}(X)} \left[ \int_{X\times X} (\psi(y)-\phi(x)) d\gamma(x,y) - 
		\int_X \psi d\mu_- + \int_X \phi d\mu_+ \right]\\
&= \sup_{\psi, \phi \in \mathcal{U}(X)} \left[ \int_X \psi d(\gamma_2 - \mu_-) + \int_X \phi d(\mu_+ - \gamma_1) \right]\\
 &=  \begin{cases}
0 & \mbox{if }  \gamma_1 \succeq_{cvx} \mu_+ \textrm{ and } \gamma_2 \preceq_{cvx} \mu_-\\ 
\infty & \mbox{otherwise.}
\end{cases}
\end{align*}

The last equality is true because if $\gamma_1 \succeq_{cvx} \mu_+$ doesn't hold, we can find a function $\phi \in \mathcal{U}(X)$ such that $\int_X \phi d(\mu_+ - \gamma_1) > 0$. Since we are allowed to scale $\phi$ arbitrarily, we can make the inside quantity as large as we want. The same holds when $\mu_- \not \succeq_{cvx} \gamma_2$.

\noindent We now apply Fenchel-Rockafellar duality:
\begin{align*}
\inf_{f \in CB(X \times X)} [\Theta(f) + \Xi(f)] &= 
\hspace{-.1in}\max_{\gamma \in \Gamma(X \times X)} [-\Theta^*(-\gamma)-\Xi^*(\gamma)]\\
\inf_{\substack{f(x,y) \geq - \Lone{x-y} \\ f(x,y) = \psi(y) - \phi(x) \\ \psi,\phi \in \mathcal{U}(X)}} \left( \int_X \psi d\mu_- - \int_X \phi d\mu_+ \right)  &= 
\hspace{-.1in}\max_{\gamma \in \Gamma_+(X \times X)} \left[ - \int_{X \times X} \Lone{x-y} d\gamma(x,y) -\Xi^*(\gamma)\right]\\
\inf_{\substack{\psi, \phi \in \mathcal{U}(X) \\ \phi(x) - \psi(y) \leq \Lone{x-y}}} \left(\int_X \psi d\mu_- - \int_X \phi d\mu_+\right) &= 
\hspace{-.1in}\max_{\substack{\gamma \in \Gamma_+(X \times X)\\ \gamma_1 \succeq_{cvx} \mu_+ \\ \gamma_2 \preceq_{cvx} \mu_-}} \left( - \int_{X \times X} \Lone{x-y} d\gamma(x,y) \right)\\
\sup_{\substack{\psi, \phi \in \mathcal{U}(X) \\ \phi(x) - \psi(y) \leq \Lone{x-y}}} \left( \int_X \phi d\mu_+ - \int_X \psi d\mu_- \right) &=  
\hspace{-.1in}\min_{\substack{\gamma \in \Gamma_+(X \times X)\\ \gamma_1 \succeq_{cvx} \mu_+ \\ \gamma_2 \preceq_{cvx} \mu_-}}\left(\int_{X \times X} \Lone{x-y} d\gamma(x,y) \right).
\end{align*}
\end{prevproof}

\subsection{From Two Convex Functions to One}

\begin{lemma}\label{twotoone}
Let $X = \prod_{i=1}^n [x^{\textrm{low}}_i,x^{\textrm{high}}_i]$ for some $x^{\textrm{low}}_i,x^{\textrm{high}}_i \ge 0$, and let $\mu \in \Gamma(X)$ such that $\mu(X)  =0$.
Then
$$\sup_{\substack{\phi, \psi \in \mathcal{U}(X) \\ \phi(x) - \psi(y) \leq \Lone{x-y} } } \left( \int_X \phi d \mu_+ - \int_X \psi d\mu_- \right) = \sup_{u \in \mathcal{U}(X) \cap \mathcal{L}_1(X)}\left(\int_X u d\mu_+ - \int_X u d \mu_- \right) .$$
Furthermore, if the supremum of one side is achieved, then so is the supremum of the other side.
\end{lemma}

\begin{prevproof}{Lemma}{twotoone}
Given any feasible $u$ for the right-hand side of Lemma~\ref{twotoone}, we observe that $\phi = \psi = u$ is feasible for the left-hand side, and therefore the left-hand side is at least as large as the right-hand side. It therefore suffices to prove the reverse direction of the inequality. Let $\phi$ and $\psi$ be feasible for the left-hand side. Given $\phi$, it is clear that $\psi$ must satisfy $\psi(y) \geq \sup_x [\phi(x) - \Lone{x-y}]$.

Set $\bar{\psi}(y) =  \sup_x [\phi(x) - \Lone{x-y}]$. Since $\psi$ exists, this supremum indeed has finite value. Since $\bar{\psi} \leq \psi$ pointwise, it follows that $\int_X \bar{\psi} d\mu_- \leq \int_X \psi d \mu_-$. We must now prove that $\bar{\psi} \in \mathcal{U}(X)$, thereby showing that $\phi, \bar{\psi}$ is feasible for the left-hand side and that replacing $\psi$ by $\bar\psi$ does not decrease the objective value.

\begin{claim}\label{1234}
 $\bar{\psi} \in \mathcal{U}(X)$ and $\bar{\psi} \in \mathcal{L}_1(X)$.
\end{claim}

\begin{proof}
 We will first show that $\bar{\psi} \in \mathcal{U}(X)$.
We need to show continuity, monotonicity, and convexity.
\begin{itemize}
	\item \textbf{Continuity.} Continuity of $\bar\psi$ follows from the Maximum Theorem since both $\phi$ and $\Lone{\cdot}$ are uniformly continuous.
	\item \textbf{Monotonicity.} Let $y \leq y'$ coordinate-wise and let $x$ be arbitrary. We must show that there exists an $x'$ such that $\phi(x) - \Lone{x-y} \leq \phi(x') - \Lone{x'-y'}$. Set $x_i' = \max\{x_i, y_i' \}$. Since $x \leq x'$, we have $\phi(x) \leq \phi(x')$. We notice that if $x_i \geq y_i'$ then $x_i' =x_i$ and thus $|x_i'-y_i'| \leq |x_i-y_i|$, while if $x_i \leq y_i'$ then $ |x_i' - y_i'| = 0$. Therefore, we have that $\Lone{x-y} \geq \Lone{x'-y'}$ and thus $\phi(x) - \Lone{x-y} \leq \phi(x') - \Lone{x'-y'}$, as desired.
	\item \textbf{Convexity.} Let $y, y', y''$ be collinear points in $X$ such that $y = \frac{y'+y''}{2}$. Then, given any $x$, we must show that there exist $x'$ and $x''$ such that
	$$\phi(x') - \Lone{x'-y'} + \phi(x'')  - \Lone{x''-y''} \geq 2\phi(x) - 2\Lone{x-y}.$$
	We define $x_i'$ and $x_i''$ as follows:
	\begin{itemize}
		\item If $y_i' \geq y_i''$, set $x_i' = \max\{ x_i, y_i'\}$ and $x_i'' = \max\{ 2x_i - x_i', y_i'' \}$.
		\item If $y_i' < y_i''$, set $x_i'' = \max\{x_i,y_i'\}$ and $x_i' = \max\{2x_i-x_i'', y_i'\}$. 
	\end{itemize}
	Notice that $x' + x'' \geq 2x$, and thus (since $\phi$ is convex and monotone) we have $\phi(x') + \phi(x'') \geq 2\phi(x)$.
	
	Suppose without loss of generality that $y_i' \geq y_i''$. We now consider two cases:
\begin{itemize}
		\item $y_i' \geq x_i$. We then have $x_i' = y_i'$ and $x_i'' = \max\{2x_i - y_i', y_i'' \}$. 
Therefore, $|y_i' - x_i'|$ = 0 and $|y_i''-x_i''| \le |y_i'' - 2x_i + y_i'| = 2 |y_i - x_i|$ since $y_i'+y_i'' = 2y_i$.
		\item $y_i' < x_i$. We now have $x_i' = x_i$ and $x_i'' = \max\{x_i, y_i''\}=x_i.$ Therefore $|y_i'' - x_i''| + |y_i' - x_i'|$  is equal to $|y_i' + y_i'' - 2x_i|$, which equals $|2y_i - 2x_i|$.
	\end{itemize}
	Therefore, we have that $|y_i' - x_i'| + |y_i'' - x_i''| \leq |2y_i - 2x_i|$ for all $i$, which implies that $\Lone{x'-y'} + \Lone{x''-y''} \leq 2\Lone{x-y}$. 
\end{itemize}

We have thus shown that $\bar{\psi} \in \mathcal{U}(X)$. 
We will now show that $\bar{\psi} \in \mathcal{L}_1(X)$. We have
\begin{align*}
\bar{\psi}(x) &- \bar{\psi}(y) = \sup_{z}\inf_{w}( \phi(z) - \Lone{z-x} - \phi(w)  + \Lone{w-y}  )\\
&\leq \sup_z(\phi(z) - \Lone{z-x} - \phi(z) + \Lone{z-y})\\
&= \sup_z (\Lone{z-y} - \Lone{z-x}) \leq \Lone{x-y}.
\end{align*}
\end{proof}

Since $\phi, \bar{\psi}$ are a feasible pair of functions for the left-hand side of Lemma~\ref{twotoone}, we know that $\phi$ satisfies the inequality $\phi(x) \leq \inf_y [\bar{\psi}(y) + \Lone{x-y}]$. We now set $\bar{\phi}(x) =  \inf_y [\bar{\psi}(y) + \Lone{x-y}]$. It is clear that the value of the left-hand objective function under $\bar{\phi}, \bar{\psi}$ is at least as large as its value under $\phi, \bar{\psi}$.

We claim that not only is $\bar{\phi}$ continuous, monotonic, and convex, but in fact that $\bar{\phi} = \bar{\psi}$. We notice that $\bar{\phi}(x) \leq \bar{\psi}(x) + \Lone{x-x} = \bar\psi(x).$ To prove the other direction of the inequality, we compute
$$
\bar{\phi}(x) = \inf_y \left[ \bar{\psi}(y) + \Lone{x-y} \right] = \bar\psi(x) + \inf_y \left[ \bar{\psi}(y) - \bar{\psi}(x) + \Lone{x-y} \right] \geq \bar\psi(x)
$$
where the last inequality holds since $\bar{\psi}(x) - \bar{\psi}(y) \leq \Lone{x-y}$. Therefore $\bar\phi = \bar\psi$, and thus $\bar\phi \in \mathcal{U}(X)$. Since $\bar\phi$ satisfies the inequality $\bar\phi(x) - \bar\phi(y) \leq \Lone{x-y}$ it is feasible for the right-hand side of Lemma~\ref{twotoone}, and the value of the right-hand objective under $\bar\phi$ is at least as large the value of the left-hand objective under $\phi, \psi$. We notice finally that if $\phi, \psi$ are optimal for the left-hand side, then $\bar\phi$ is optimal for the right-hand side.
\end{prevproof}

\subsection{Proof of Theorem~\ref{strongduality}}

 By combining  Lemma~\ref{weakduality}, Lemma~\ref{longlemma}, and Lemma~\ref{twotoone}, we have

\begin{align*}
\inf_{\substack{  \gamma \in \Gamma_+(X \times X) \\ \gamma_1 - \gamma_2 \succeq_{1} \mu   }}&\int_{X\times X} \Lone{x-y} d\gamma \geq \sup_{u \in \mathcal{U}(X) \cap \mathcal{L}_1(X)} \int_X u d\mu\\
& \hspace{-.8in}= \hspace{-.1in} \sup_{\substack{\phi, \psi \in \mathcal{U}(X) \\ \phi(x) - \psi(y) \leq \Lone{x-y} } } \left( \int_X \phi d \mu_+ - \int_X \psi d\mu_- \right)  = \hspace{-.1in} \inf_{\substack{\gamma \in \Gamma_+(X \times X)\\ \gamma_1 \succeq_{cvx} \mu_+ \\ \gamma_2 \preceq_{cvx} \mu_-}}\int_{X \times X} \Lone{x-y} d\gamma(x,y).
\end{align*}
By Lemma~\ref{longlemma}, the last minimization problem above achieves its infimum for some $\gamma^*$. We notice that $\gamma^*$ is also feasible for the first minimization problem above, and therefore the inequality is actually an equality and $\gamma^*$ is optimal for the first minimization problem. In addition, since $\gamma^*$ is feasible for the last minimization problem, it satisfies $\gamma_1^*(X) = \gamma_2^*(X) = \mu_+(X)$.  All that remains is to prove that the supremum to the maximization problem is achieved for some $u^*$. A proof of this fact is in Appendix~\ref{existenceappendix}.

\subsection{Existence of Optimal Mechanism}\label{existenceappendix}

We now prove that the supremum of the maximization problem of Theorem~\ref{strongduality} is achieved for some $u^*$. 
Consider a sequence of feasible functions $u_1, u_2, \ldots \in \mathcal{U}(X) \cap \mathcal{L}_1(X)$  such that $\int_X u_id\mu$ converges monotonically to the supremum value $V$, which we have proven is finite.\footnote{Finiteness is also obvious because $X$ is bounded and the infimum problem is feasible.} Since $\mu(X) = 0$, we may without loss of generality assume that $u_i(0^n) = 0$ for all $u_i$.
Since all of the functions are bounded by $\| x^{\textrm{high}} \|_1$ and are $1$-Lipschitz (which implies equicontinuity), the Arzel\`a-Ascoli theorem implies that there exists a uniformly converging subsequence. Let $u^*$ be the limit of that subsequence. Since the convergence is uniform, the function $u^*$ is $1$-Lipschitz, non-decreasing and convex and thus feasible for the mechanism design problem. Moreover, since the objective is linear, the revenue of the mechanism with that utility is equal to $V$ and thus the supremum is achieved.

 \subsection{Omitted Proofs from Section~\ref{examplesection} - Example \ref{unifexample}} \label{app:proof of uniform non 0-1}

It is straightforward to verify that the mechanism is IC and IR. All that remains is to prove that the utility function $u^*$ induced by the mechanism is optimal.

The transformed measure $\mu$ of the type distribution is composed of:
\begin{itemize}
\item A point mass of $+1$ at $(4,4)$.
\item Mass $-3$ distributed  throughout the  rectangle (Density $-\frac{1}{12}$)
\item Mass  $+\frac{7}{3}$ distributed  on  upper edge of  rectangle (Linear density $+\frac{7}{36}$)
\item Mass  $-\frac{4}{3}$ distributed  on  lower edge of  rectangle (Linear density $-\frac{1}{9}$)
\item Mass  $+\frac{4}{3}$ distributed  on  right edge of  rectangle (Linear density $+\frac{4}{9}$)
\item Mass  $-\frac{1}{3}$ distributed  on  left edge of  rectangle (Linear density $-\frac{1}{9}$)
\end{itemize}
We claim that $\mu(Z) = \mu(Y) = \mu(W) = 0$, which is straightforward to verify. 

We will construct an optimal $\gamma^*$  for the dual program of Theorem~\ref{strongduality}, using the intuition of Remark~\ref{geometricremark}. Our $\gamma^*$ will be decomposed into $\gamma^* = \gamma^Z + \gamma^Y + \gamma^W$ with $\gamma^Z \in \Gamma_+(Z \times Z)$, $\gamma^Y \in \Gamma_+(Y \times Y)$, and $\gamma^W \in \Gamma_+(W \times W)$. To ensure that $\gamma^*_1 - \gamma^*_2 \succeq_{cvx} \mu$, we will show that
$$\gamma^Z_1 - \gamma^Z_2 \succeq_{cvx} \mu|_Z; \quad \gamma^Y_1 - \gamma^Y_2 \succeq_{cvx} \mu|_Y; \quad \gamma^W_1 - \gamma^W_2 \succeq_{cvx} \mu|_W.$$
We will also show that the conditions of Corollary~\ref{linearintegral} hold for each of the measures $\gamma^Z$, $\gamma^Y$, and $\gamma^W$ separately, namely $\int u^* d(\gamma^A_1 - \gamma^A_2) = \int_A u^* d  \mu$ and $u^*(x) - u^*(y) = \Lone{x-y}$ hold $\gamma^A$-almost surely
for $A$ = $Z$, $Y$, and $W$.
\begin{itemize}[label={},leftmargin=0pt]
\item {\bf Construction of $\gamma^Z$.} Since $\mu_+|_Z$ is a point-mass at $(4,4)$ and  $\mu_-|_Z$ is distributed throughout a region which is coordinatewise greater than $(4,4)$, we notice that $\mu|_Z \preceq_{cvx} 0$. We therefore set $\gamma^Z$ to be the zero measure, and the relation $\gamma^Z_1 - \gamma^Z_2 = 0 \succeq_{cvx} \mu|_Z$, as well as the two necessary equalities from Corollary~\ref{linearintegral}, are trivially satisfied.

\item {\bf Construction of $\gamma^W$.} We will construct $\gamma^W \in \Gamma(\mu_+|_W, \mu_-|_W)$ such that  $x \geq y$ component-wise holds $\gamma^W(x,y)$ almost surely. Geometrically, we view this as ``transporting'' $\mu_+|_W$ into $\mu_-|_W$ by moving mass downwards and leftwards. Indeed, since both items are allocated with probability 1 in $W$, being able to transport both downwards and leftwards is in line with  our interpretation of the second condition of Corollary~\ref{linearintegral}, as explained in Remark~\ref{geometricremark}.\footnote{To prove the existence of such a map, it is equivalent by Strassen's theorem to prove that $\mu_+|_W$ stochastically dominates $\mu_-|_W$ in the first order, but in this example we will directly define such a map.}  

We notice that $\mu_+|_W$ consists of mass distributed on the top and right edges of $W$, while $\mu_-|_W$ consists of mass on the interior and bottom of $W$. We first  match the $\mu_+$ mass on $[8, 16] \times \{7\}$ with the $\mu_-$ mass on $[8,16] \times [\frac{14}{3},7]$ by moving mass downwards, then we match the $\mu_+$ mass on $\{16\} \times [4,{\frac{14}{3}}]$ with the $\mu_-$ mass on $[\frac{32}{3},16] \times (4, \frac{14}{3}]$ by moving mass to the left, and we finally match the $\mu_+$ mass on $\{16\} \times [\frac{14}{3},7]$ with the remaining negative mass arbitrarily. Noticing that $u^*(x) = \Lone{x} - 12$ for all $x \in W$, it is straightforward to verify the desired properties from Corollary~\ref{linearintegral}.
 
\item {\bf Construction of $\gamma^Y$.} This is the most involved step of the proof. Since item 2 is allocated with 100\% probability in region $Y$, by Remark~\ref{geometricremark} we would like to transport the positive mass $\mu_+|_Y$ into $\mu_-|_Y$ by moving mass straight downwards. However, this is impossible without first ``shuffling'' $\mu|_Y$, due to the negative mass on the left boundary of $Y$. Therefore, we first ``shuffle'' the positive part of $\mu|_Y$ (on the top boundary) to push positive mass onto the point $(4,7)$ (the top-left corner of $Y$), and only then do we transport the positive part of the shuffled measure into the negative part by sending mass downwards. Since the positive and negative parts of $\mu|_Y$ must be matchable by only sending mass downwards, we know how the post-shuffling measure should look. In particular, on every vertical line in region $Y$ the net post-shuffling mass should be zero.

So rather than constructing $\gamma^Y$ with $\gamma^Y_1 - \gamma^Y_2$ equal  to $\mu|_Y$, we will have $\gamma^Y_1 - \gamma^Y_2 = \mu|_Y + \alpha$, where the ``shuffling'' measure
 $\alpha =  \alpha_+ - \alpha_- \succeq_{cvx} 0$. As discussed above, we set $\alpha$ to have density function
 $$f_\alpha(z_1,z_2) = \mathbb{I}_{z_2 = 7} \cdot \left( \frac{1}{9} \mathbb{I}_{z_1 = 4} + \frac{1}{24}\left(z_1 - \frac{20}{3}\right)     \right) \cdot \mathbb{I}_{z \in Y}.$$
The measure $\alpha$ is supported on the line $[4,8] \times \{7\}$ and consists of a point mass of $\frac{1}{9}$ at $(4,7)$ followed by allocating mass along the 1-dimensional upper boundary of $Y$ according to a density function which begins negative and increases linearly. It is straightforward to verify that $\alpha \succeq_{cvx} 0$,\footnote{Since $\alpha$ is supported on a 1-dimensional line, this verification uses a property analogous to the standard characterization of one-dimensional second-order stochastic dominance via the cumulative density function. Informally, we can argue that $\alpha \succeq_{cvx} 0$ by considering integrals of one-dimensional test functions (by restricting our attention to the line $z_2 = 7$) and noticing that, since $\alpha(Y) = 0$, we need only consider test functions $h$ which have value 0 at $z_1 = 4$. We then use the fact that all linear functions integrate to 0 under $\alpha$ and that (ignoring the point mass at $z_1 = 4$, since $h$ is $0$ at this point) the density of $\alpha$ is monotonically increasing.} which we need for feasibility, and that $\int_Y u^* d\alpha = 0$, which we need to satisfy complementary slackness.

We are now ready to define $\gamma^Y \in \Gamma(\mu_+|_Y + \alpha_+, \mu_-|_Y + \alpha_-)$. We construct $\gamma^Y$ so that $x_1 = y_1$ and $x_2 \geq y_2$ hold $\gamma^Y(x,y)$ almost surely. Since $\mu_+|_Y + \alpha_+$ only assigns mass to the upper boundary of $Y$, to show that $\gamma^Y$ can be constructed so that all mass is transported ``vertically downwards'' we need only verify that $\mu_+|_Y + \alpha_+$ and $\mu_-|_Y + \alpha_-$ assign the same density to any vertical ``strip'' in $Y$. Indeed,
\begin{align*}
\hspace{-0.5in}(\mu_-|_Y + \alpha_-)(\{4\} \times [6,7]) &=\mu_-|_Y (\{4\} \times [6,7]) =  \frac{1}{9} =  \alpha_+(\{4\} \times [6,7])  \\
&= (\mu_+|_Y + \alpha_+)(\{4\} \times [6,7])
\end{align*}
and, for all $z_1 \pm \epsilon \in (4,8]$, we compute the following, using the fact that the surface area of $Y \cap \left([z_1 - \epsilon, z_1 + \epsilon] \times [4,7] \right)$ is $2\epsilon \cdot \left(\frac{z_1}{2} - 1 \right)$:
\begin{align*}
(\mu_-|_Y - &\alpha|_Y)([z_1- \epsilon, z_1+\epsilon] \times [4,7])  \\
&= \frac{1}{12}\cdot \left(2\epsilon \cdot \left(\frac{z_1}{2} - 1 \right) \right)   - \frac{1}{24}\int_{z_1-\epsilon}^{z_1+\epsilon}(z - \frac{20}{3})dz\\
&=  \frac{\epsilon z_1}{12} - \frac{\epsilon}{6} - \frac{1}{24}(2\epsilon z_1 - \frac{40\epsilon}{3}) = \frac{7\epsilon}{18}
= \mu_+|_Y ([z_1- \epsilon, z_1+\epsilon] \times [4,7]) .
\end{align*}

Since $u^*$ has the property that $u^*(z_1, a) - u^*(z_1,b) = a-b$ for all $(z_1, a), (z_1,b) \in Y$ (as the second good is received with probability 1), it follows that $\gamma^Y$ satisfies the necessary conditions of Corollary~\ref{linearintegral}.
\end{itemize} \section{Proof of Stochastic Conditions of Section~\ref{bundlingsection}}\label{fullbundlingproofappendix}

Our goal in this section is to prove Theorem~\ref{bundlingtheorem}. We begin by presenting some useful probabilistic tools that will be essential for the proof.

\subsection{Probabilistic Lemmas}\label{probabilisticlemmas}
We first present a useful result about convex dominance of random variables. For more information about this result, see Theorem 7.A.2 of \cite{Shaked}.
 
\begin{lemma}[Strassen's Theorem]\label{strassen}
Let $A$ and $B$ be random vectors. Then $A \preceq_{cvx} B$ if and only if there exist random vectors $\hat{A}$ and $\hat{B}$, defined on the same probability space, such that $\hat{A} =_{st} A$, $\hat{B} =_{st} B$, and $\mathbb{E}[\hat{B}|\hat{A}] \geq \hat{A}$ almost surely, where the final inequality is componentwise and where $=_{st}$ denotes equality in distribution.
\end{lemma}

It is easy to extend the above result to convex dominance with respect to a vector $\vec v$ as defined in Definition~\ref{firstorderdef}.

\begin{lemma}[Extended Strassen's Theorem]\label{strassenv}
Let $A$ and $B$ be random vectors. Then $A \preceq_{cvx(\vec v)} B$ if and only if there exist random vectors $\hat{A}$ and $\hat{B}$, defined on the same probability space, with $\hat{A} =_{st} A$, $\hat{B} =_{st} B$, such that (almost surely):
\begin{itemize}
\item if $v_i = +1$, then $E[\hat B_i | \hat A] \ge \hat A_i$
\item if $v_i = 0$, then $E[\hat B_i | \hat A] = \hat A_i$
\item if $v_i = -1$, then $E[\hat B_i | \hat A] \le \hat A_i$
\end{itemize}
\end{lemma}

We now state a  multivariate variant of Jensen's inequality along with the necessary condition for equality to hold. The proof of this result is standard and straightforward, and thus is omitted.
\begin{lemma}[Jensen's inequality]\label{jensen}
Let $V$ be a vector-valued random variable with values in $[0,M]^n$ and let $u$ be a convex Lipschitz-continuous function mapping $[0,M]^n \rightarrow \mathbb{R}$. Then $\mathbb{E}[u(V)] \geq u(\mathbb{E}[V]).$ Furthermore, equality holds if and only if, for every $a$ in the subdifferential of $u$ at $\mathbb{E}[V]$, the equality $u(V) = a\cdot (V - \mathbb{E}[V]) + u(\mathbb{E}[V])$ holds almost surely.
\end{lemma}

The following lemma is a conditional variant of Lemma~\ref{jensen}, based on the multivariate conditional Jensen's inequality, as in Theorem 10.2.7 of \cite{Dudley}.  This lemma is used as a tool for Lemma~\ref{conditionallem}, the main result of this subsection.
\begin{lemma}\label{myjensen}
Let $(\Omega, \mathcal{A}, P)$ be a probability space, $V$ be a random variable on $\Omega$ with values in $X$ where $X = \prod_{i=1}^n [x^{\textrm{low}}_i,x^{\textrm{high}}_i]$, and $u : X \rightarrow \mathbb{R}$ be convex and Lipschitz continuous. Let $\mathcal{C}$ be any sub-$\sigma$-algebra of $\mathcal{A}$ and suppose that $\mathbb{E}[u(V) | \mathcal{C}] = u(\mathbb{E}[V | \mathcal{C}])$ almost-surely. Then for almost all $x \in \Omega$ the equality $u(y) = a_{y_x} \cdot (y - y_x) + u(y_x)$ holds almost surely with respect to the law \footnote{The law $P_{V| \mathcal{C}}(\cdot, x)$ allows us to express the conditional distribution of $V$ given $\mathcal{C}$} $P_{V| \mathcal{C}}(\cdot, x)$ , where $y_x$ is the expectation of the random variable with law $P_{V | \mathcal{C}}(\cdot, x)$ and $a_{y_x}$ is any subgradient of $u$ at $y_x$.
\end{lemma}

\begin{prevproof}{Lemma}{myjensen}
The proof is based on the proof of the multivariate conditional Jensen's inequality, as in Theorem 10.2.7 of \cite{Dudley}. This theorem requires $|V|$ and $u \circ V$ to be integrable, which is true in our setting. We note that the theorem applies when $u$ is defined in an open convex set, but because $u$ is Lipschitz continuous we can extend it to a function with domain an open set containing $X$.
The multivariate conditional Jensen's inequality states that, almost surely, $\mathbb{E}[V | \mathcal{C}] \in \mathcal{C}$ and
$\mathbb{E}[u(V) | \mathcal{C}] \geq u(\mathbb{E}[V | \mathcal{C}]).$ The proof of Theorem 10.2.7 in \cite{Dudley} furthermore shows that the following two equalities hold:
$$
\mathbb{E}[V | \mathcal{C}](x) = \int_{X} y P_{V | \mathcal{C}}(dy,x); \qquad
\mathbb{E}[u(V) | \mathcal{C}](x) = \int_{X} u(y) P_{V| \mathcal{C}}(dy,x).
$$
Since $\mathbb{E}[u(V) | \mathcal{C}](x) = u(\mathbb{E}[V | \mathcal{C}])(x)$ for almost all $x$, we apply the unconditional Jensen inequality (Lemma~\ref{jensen}) to the laws $P_{V | \mathcal{C}}(\cdot, x)$ to prove the lemma.
\end{prevproof}

We now present Lemma~\ref{conditionallem}. This lemma states that for random variables $A$ and $B$ with $A \preceq_{cvx} B$ if it holds that $u(A) = u(B)$ for some convex function $u$, then there exists a coupling between $A$ and $B$ with several desirable properties, including that points are only matched if $u$ shares a subgradient at these points.

\begin{lemma}\label{conditionallem}
Let $A$ and $B$ be vector random variables with values in $X$, where $X = \prod_{i=1}^n [x^{\textrm{low}}_i,x^{\textrm{high}}_i]$, such that $A \preceq_{cvx} B$. Let $u: X \rightarrow \mathbb{R}$ be 1-Lipschitz with respect to the $\ell_1$ norm, convex, and monotonically non-decreasing. Suppose that $\mathbb{E}[u(A)] = \mathbb{E}[u(B)]$ and that $g : X \rightarrow [0,1]^n$ is a measurable function such that for all $z \in X$, $g(z)$ is a subgradient of $u$ at $z$.

Then there exist random variables $\hat{A} =_{st} A$ and $\hat{B} =_{st} B$ such that, almost surely:
\begin{itemize}
	\item $u(\hat{B}) = u(\hat{A}) + g(\hat{A}) \cdot (\hat{B} - \hat{A})$
	\item $g(\hat{A})$ is a subgradient of $u$ at $\hat{B}$.
	\item $\mathbb{E}[\hat{B}|\hat{A}]$ is componentwise greater or equal to $\hat{A}$
	\item $u(\mathbb{E}[\hat{B} | \hat{A}]) = u(\hat{A})$.
\end{itemize}
\end{lemma}

\begin{prevproof}{Lemma}{conditionallem}
By Lemma~\ref{strassen}, there exist random variables $\hat{A} =_{st} A$ and $\hat{B} =_{st} B$ such that $\mathbb{E}[\hat{B}|\hat{A}]$ is componentwise greater than or equal to $\hat{A}$ almost surely. We have
$$
0 = \mathbb{E}[u(\hat{B}) - u(\hat{A}) ]\geq \mathbb{E}[u(\hat{B}) - u(\mathbb{E}[\hat{B}|\hat{A}])]
= \mathbb{E}[\mathbb{E}[u(\hat{B})|\hat{A}] - u(\mathbb{E}[\hat{B}|\hat{A}])] \geq 0
$$
and therefore
$\mathbb{E}[\mathbb{E}[u(\hat{B})|\hat{A}]] = \mathbb{E}[u(\mathbb{E}[\hat{B}|\hat{A}])] = \mathbb{E}[u(\hat{B})] = \mathbb{E}[u(\hat{A})].$

Since $u$ is monotonic, $u(\hat{A}) \leq u(\mathbb{E}[\hat{B}|\hat{A}])$ almost surely. Since $\mathbb{E}[u(\hat{A})] = \mathbb{E}[u(\mathbb{E}[\hat{B}|\hat{A}])]]$, it follows that $u(\hat{A}) = u(\mathbb{E}[\hat{B}|\hat{A}])$ almost surely.

Select any collection of random variables $\{\hat{B}|_{\hat{A} =x}\}$ corresponding to the laws $P_{\hat{B}|\hat{A}}(\cdot, x)$.
For almost all values $x$ of $\hat{A}$, $\mathbb{E}[\hat{B}|_{\hat{A} = x}]$ is componentwise greater than $x$ and $u(x) = u(\mathbb{E}[\hat{B}|_{\hat{A} = x}])$. We claim now that any subgradient $a_x$ of $u$ at $x$ is also a subgradient of $u$ at $\mathbb{E}[\hat{B}|_{\hat{A} = x}]$. Indeed, choose such a subgradient $a_x$. We compute
$$
	u(\mathbb{E}[\hat{B}|_{\hat{A} = x}]) \geq u(x) + a_x \cdot (\mathbb{E}[\hat{B}|_{\hat{A} = x}] - x) = u(\mathbb{E}[\hat{B}|_{\hat{A} = x}]) + a_x \cdot (\mathbb{E}[\hat{B}|_{\hat{A} = x}]-x)
$$
and therefore $a_x \cdot \mathbb{E}[\hat{B}|_{\hat{A} = x}] = a_x \cdot x$, by non-negativity of the subgradient. Furthermore, for any point $z \in X$,
\begin{align*}
u(z) &\geq u(x) + a_x \cdot (z - x) = u(\mathbb{E}[\hat{B}|_{\hat{A} = x}]) + a_x \cdot (z - x)\\
&= u(\mathbb{E}[\hat{B}|_{\hat{A} = x}]) + a_x \cdot (z - \mathbb{E}[\hat{B}|_{\hat{A} = x}])
\end{align*} 
and thus $a_x$ is a subgradient of $u$ at $\mathbb{E}[\hat{B}|_{\hat{A} = x}]$.

Since $\mathbb{E}[\mathbb{E}[u(\hat{B})|\hat{A}]] = \mathbb{E}[u(\mathbb{E}[\hat{B}|\hat{A}])]$, by Jensen's inequality it follows that $\mathbb{E}[u(\hat{B})|\hat{A}] =  u(\mathbb{E}[\hat{B}|\hat{A}])$ almost surely. By Lemma~\ref{myjensen}, it therefore holds for almost all values $x$ of $\hat{A}$ that the equality 
\begin{align*}
u(y) &= a_{x} \cdot (y - \mathbb{E}[\hat{B}|_{\hat{A} =x}]) + u(\mathbb{E}[\hat{B}|_{\hat{A} =x}]) = a_x \cdot (y - x) + u(\mathbb{E}[\hat{B}|_{\hat{A} =x}])\\
&= a_x \cdot (y-x) + u(x)
\end{align*}
 holds $\hat{B}|_{\hat{A} =x}$ almost surely.
 
Lastly, we will show that, almost surely, $a_x$ is a subgradient of $u$ at $\hat{B}|_{\hat{A} = x}$. Indeed, for any $p \in X$, and almost all values of $x$ we have
\begin{align*}
u(p) &\geq u(x) + a_x \cdot (p - x) = u(x) + a_x \cdot (\hat{B}|_{\hat{A} = x}  - x) + a_x \cdot(p - \hat{B}|_{\hat{A} = x})\\
&= u(\hat{B}|_{\hat{A} = x}) + a_x \cdot(p - \hat{B}|_{\hat{A} = x}).
\end{align*}
\end{prevproof}

\subsection{Proof of the Optimal Menu Theorem (Theorem~\ref{bundlingtheorem})}

To prove the equivalence we prove both implications of the theorem separately.

\subsubsection{Sufficiency Conditions}

We will show that the Optimal Menu Conditions of Definition~\ref{optconditions} imply that a mechanism $\cM$ is optimal. To show the theorem, we construct a measure $\gamma$ such that the conditions of Corollary~\ref{linearintegral} are satisfied. We will construct this measure separately for every region that corresponds to a menu choice of mechanism $\cM$. 

Consider a menu choice $(p,t) \in \textrm{Menu}_{\cM}$, the corresponding region $R$ and the corresponding vector $\vec v$ as in Definition~\ref{optconditions}. Let $A$ and $B$ be random vectors distributed according to the (normalized) measures $\mu_+|R$ and $\mu_-|R$.
From the Optimal Menu Conditions, we have that $A|_{R} \preceq_{cvx(\vec v)} B|_{R}$ (almost surely). By the extended version of Strassen's theorem (Lemma~\ref{strassenv}), it holds that there exist random vectors $\hat A, \hat B$ with $\hat A =_{st} A|_R$ and $\hat B =_{st} B|_R$, such that (almost surely):
\begin{itemize}
\item if $v_i = +1$, then $E[\hat B_i | \hat A] \ge \hat A_i$
\item if $v_i = 0$, then $E[\hat B_i | \hat A] = \hat A_i$
\item if $v_i = -1$, then $E[\hat B_i | \hat A] \le \hat A_i$
\end{itemize}

Now define the random variable $\hat C = \min ( E[\hat B | \hat A] , \hat A)$ where we take the coordinate-wise minimum. We now have that (almost surely):
\begin{itemize}
\item if $v_i = +1$, then $E[\hat B_i | \hat A] \ge \hat A_i = \hat C_i$
\item if $v_i = 0$, then $E[\hat B_i | \hat A] = \hat A_i  = \hat C_i$
\item if $v_i = -1$, then $\hat C_i = E[\hat B_i | \hat A] \le \hat A_i$
\end{itemize}

Let $\gamma_R$ be the measure according to which the vector $(\hat A, \hat C)$ is distributed. 
By construction, 
$\gamma_{R1} = \mu_+|_R$ and 
$\gamma_{R2} \preceq_{cvx} \mu_-|_R$, and thus $\gamma_{R1} - \gamma_{R2} \succeq_{cvx} \mu|_R$. 
Moreover, the conditions of Corollary~\ref{linearintegral} are satisfied:
\begin{itemize}
	\item $u(x) - u(y) = \Lone{x-y}$, is satisfied $\gamma_R(x,y)$-almost surely since $\hat A$ is larger than $\hat C$ only in coordinates for which $v_i = -1$ and thus $p_i = 1$.
	\item $\int u d(\gamma_{R1} - \gamma_{R2}) = \int u d(\mu_+|_R - \mu_-|_R)$ is satisfied: By definition we have that $\int u d\gamma_{R1} = \int u d\mu_+|_R$. Moreover, we can also show that $\int u d\gamma_{R2} = \int u d\mu_-|_R$ by noting that
$\int u d\mu_-|_R = \mu_-(R) E[ u(\hat B) ] = \mu_-(R) E[ p \cdot \hat B - t ]
= \mu_-(R)  E[ p \cdot E[\hat B | \hat A] - t ]$ and that 
$ \mu_-(R)  E[ p \cdot E[\hat B | \hat A] - t ]$ is equal to $\mu_-(R)  E[ p \cdot \hat C - t ] = \int u d\gamma_{R2}$ since $\hat C_i \neq E[\hat B_i | \hat A]$ only when $E[\hat B_i | \hat A] $ is strictly larger than $\hat A_i$ which only happens only in coordinates $i$ where $v_i = +1$ and thus $p_i = 0$.
\end{itemize}
This completes the proof that the Optimal Menu Conditions imply optimality of the mechanism since we can construct a feasible measure $\gamma$ satisfying the conditions of Corollary~\ref{linearintegral} by considering the sum of the constructed measures for each region.

\subsubsection{Optimality implies Stochastic Conditions}

We will now prove the other direction of the result. Consider an optimal mechanism $\cM = (\cP,\cT)$ with a finite menu size over type space $X = \prod_{i=1}^n [x^{\textrm{low}}_i,x^{\textrm{high}}_i]$. Since $\cM$ is given in essential form, in the menu of $\cM$ there is no dominated option. So for all options on the menu there is a set of buyer types that strictly prefer it from any other option, and that set of types occurs with positive probability. 

Now, define the set $Z = \{x \in X : p \cdot x - t = \cP(x) \cdot x - \cT(x) \textrm{ for } (p,t) \in \textrm{Menu}_{\cM} \textrm{ with }  (p,t) \neq (\cP(x),\cT(x)) \}$. This is the set of types where there is no single option that is the best and it is where the utility function of the mechanism is not differentiable. We show the following lemma.

\begin{lemma}
  \label{lem:zero-mass}
$\mu_-(Z) = 0$
\end{lemma}

\begin{proof}
 Note that, by its construction, $\mu_-$ assigns zero mass to any $k$-dimensional surface for $k\le n-2$. Moreover, it only assigns mass to $(n-1)$-dimensional surfaces which lie along the boundary of $X$.
  
  Every pair of distinct choices $(p,t), (p',t') \in \textrm{Menu}_{\cM}$ defines a hyperplane $p \cdot x - t = p' \cdot x - t'$ containing the types who derive the same utility from these two choices. As the menu is finite, there exist a finite number of such pairs, hence a finite number of hyperplanes. The set $Z$ contains a subset of types in the finite union of these hyperplanes, so $\mu_-$ assigns no mass to the subset of $Z$ which lies on the interior of $X$.

Regarding the $\mu_-$-measure of $Z$ on the boundaries, notice that the intersection of each of the aforementioned hyperplanes $p \cdot x - t = p' \cdot x - t'$ with each boundary $x_i = x_i^{\textrm{low}}$ is $(n-2)$-dimensional, unless the hyperplane coincides with $x_i = x_i^{\textrm{low}}$. If it is $(n-2)$-dimensional then its measure under $\mu_-$ is $0$. Otherwise, it must be that $p_j=p_j'$, for all $j \neq i$, and $p_i \neq p'_i$; say $p_i > p_i'$ without loss of generality . This implies that $(p,t)$ must dominate $(p',t')$, for all types $x \in X$. This contradicts our assumption that no menu choices are dominated.
\end{proof}

Let $u$ be the utility function of the optimal mechanism $\cM = (\cP,\cT)$ and $\gamma$ be the optimal measure of Theorem~\ref{strongduality}. Then, $\gamma$ satisfies the properties of Corollary~\ref{linearintegral}. In particular, it holds that:
\begin{enumerate}[leftmargin=*]
\item 
\begin{equation} \label{item:expectation condition}
  \int u d(\gamma_1 + \mu_-) = \int u d(\mu_+ + \gamma_2)\end{equation}
\item $u(x) - u(y) = \Lone{x-y}$, $\gamma(x,y)$ almost surely. Since this can happen only if $x$ is coordinate-wise greater than $y$, it holds (almost surely with respect to $\gamma$) that $\Lone{x-y} = \sum_i x_i - \sum_i y_i$ which implies that (almost surely) $u(x) - \sum_i x_i = u(y) - \sum_i y_i$ and thus
\begin{align}\label{eq:u-minus-l1}
  \int (u(x) - \sum_i x_i) d\gamma_1 = \int (u(y) - \sum_i y_i) d\gamma_2
\end{align}
Moreover, again since  $x$ is coordinate-wise greater than $y$ almost surely with respect to $\gamma$, it follows that $\gamma_2 \succeq_{cvx(-\vec 1)} \gamma_1$.
\end{enumerate}

We are now ready to use Lemma~\ref{conditionallem} which follows from Jensen's inequality. We will apply it in two different steps, which we will then combine to show that $\mu_+ |_{R} \preceq_{cvx(\vec v)} \mu_- |_{R}.$

\begin{itemize}[label={},leftmargin=0pt]
\item {\bf Step (ia):} We will first apply Lemma~\ref{conditionallem} to random variables $A,B$ distributed according to the measures $ \gamma_2 + \mu_+$ and $\gamma_1 + \mu_-$ respectively. Since $\mu_+ - \mu_- \preceq_{cvx} \gamma_1 - \gamma_2$, by the feasibility of $\gamma$, we have that $A \preceq_{cvx} B$. Moreover, $\mathbb{E}[u(A)] = \mathbb{E}[u(B)]$, from Equation~(\ref{item:expectation condition}) above, and $u$ is convex and non-decreasing, from the feasibility of $u$.

To apply Lemma~\ref{conditionallem}, we choose the function $g(x)$, which is a subgradient functions of u, as follows:
\begin{itemize}
\item For all $x \in X \setminus Z$ the best choice from the menu of $\cM$ is unique, hence the subgradient of  $u$ is uniquely defined. For all such $x$, we set $g(x) = \cP(x)$.
\item For all other $x$, $u$ has a continuum of different subgradients at $x$. In particular, any vector in the convex hull of $\{ p: p \cdot x - t = u(x), (p,t) \in \textrm{Menu}_{\cM}\}$ is a valid subgradient. Thus, we can always choose $g(x)$ to equal a vector of probabilities that doesn't appear as an allocation of any choice in menu $\cM$.
\end{itemize}

{\bf Step (ib):} it follows from Lemma~\ref{conditionallem} that there exist random variables $\hat{A} =_{st} A$ and $\hat{B} =_{st} B$ such that, almost surely, $g(\hat{A})$ is a subgradient of $u$ at $\hat{B}$. Fixing some $(p,t) \in \textrm{Menu}_{\cM}$ and its corresponding region $R = \{x: p = \cP(x)\}$, we denote by ${\rm cl}(R) = R \cup \partial R$ the closure of $R$ and by ${\rm int}(R) = {\rm cl}(R) \setminus Z$ the set of types which strictly prefer $(p,t)$ to any other option in the menu. Note in particular that ${\rm int}(R)$ may contain points on the boundary of $X$. With this notation, we have that almost surely:
\begin{align}
  \hat{B} \in {\rm int}(R) \implies \hat{A} \in {\rm int}(R); \label{eq:implication from B to A}\\
  \hat{A} \in {\rm int}(R) \implies \hat{B} \in {\rm cl}(R). \label{eq:implication from A to B}
\end{align}

This is because, from Lemma~\ref{conditionallem}, we know that $g(\hat{A})$ is a subgradient of $u$ at $\hat{B}$ almost surely, and we know by definition of ${\rm int}(R)$ that the subgradient is unique whenever $\hat{B} \in {\rm int}(R)$. Thus, it holds almost surely that whenever $\hat{B} \in {\rm int}(R)$ we have $g(\hat{A}) = g(\hat{B})$. Since $g$ is chosen to have differing values on ${\rm int}(R)$ and on $Z$, it follows that whenever $\hat{B} \in {\rm int}(R)$, $\hat{A} \in int(R)$ almost surely. The implication $\hat{A} \in {\rm int}(R) \implies \hat{B} \in {\rm cl}(R)$ follows from the fact that the subgradient at any point $x \in {\rm int}(R)$ can only serve as a subgradient for points $y \in {\rm cl}(R)$.

From Lemma~\ref{conditionallem}, we also have that $u(\mathbb{E}[\hat{B}|\hat{A}]) = u(\hat{A})$ almost surely. It follows that, almost surely,
$$u(\mathbb{E}[\hat{B}|\hat{A}]) \cdot \mathbb{I}_{\hat{A} \in {\rm int}(R)} = u(\hat{A}) \cdot \mathbb{I}_{\hat{A} \in {\rm int}(R)}$$
Given~\eqref{eq:implication from A to B} and since $u$ is linear restricted to ${\rm cl}(R)$, it follows that the left hand side equals:
$$\mathbb{E}[u(\hat{B})|\hat{A}] \cdot \mathbb{I}_{\hat{A} \in {\rm int}(R)}$$
We also have from Lemma~\ref{conditionallem} that, almost surely, it holds componentwise
\begin{align}
  \mathbb{E}[\hat{B}|\hat{A}] \ge \hat{A}. \label{eq: Strassen condition 1}
\end{align}
The above imply that, almost surely:
\begin{align}
  p_i >0 \implies \mathbb{E}[\hat{B}_i|\hat{A}] \cdot \mathbb{I}_{\hat{A} \in {\rm int}(R)} = \hat{A}_i \cdot \mathbb{I}_{\hat{A} \in {\rm int}(R)} \label{eq: Strassen condition}
  \end{align}
as otherwise we cannot have $\mathbb{E}[u(\hat{B})|\hat{A}] \cdot \mathbb{I}_{\hat{A} \in {\rm int}(R)} = u(\hat{A}) \cdot \mathbb{I}_{\hat{A} \in {\rm int}(R)}$, given that $u$ is linear and non-decreasing in ${\rm cl}(R)$.

Equations~\eqref{eq: Strassen condition 1}, \eqref{eq: Strassen condition} and Lemma~\ref{strassenv} imply that
\begin{align}
  \hat{A} \cdot \mathbb{I}_{\hat{A} \in {\rm int}(R)} \preceq_{cvx(\vec v)} \hat{B} \cdot \mathbb{I}_{\hat{A} \in  {\rm int}(R)} \label{eq: region stochastic dom}
  \end{align}
for the $\vec v$ defined in Definition~\ref{optconditions} for the menu choice $(p,t)$. Note that:
\begin{align*}
\hat{B} \cdot \mathbb{I}_{\hat{A} \in  {\rm int}(R)} &= 
\hat{B} \cdot \mathbb{I}_{\hat{A}, \hat{B} \in  {\rm int}(R)} + 
\hat{B} \cdot  \mathbb{I}_{ \hat{A} \in {\rm int}(R) \wedge \hat{B} \notin  {\rm int}(R) }\\
&= \hat{B} \cdot \mathbb{I}_{\hat{B} \in  {\rm int}(R)} + 
\hat{B} \cdot  \mathbb{I}_{ \hat{A} \in {\rm int}(R) \wedge \hat{B} \notin  {\rm int}(R) }
\end{align*}
where for the second equality we used~\eqref{eq:implication from B to A}. Hence, \eqref{eq: region stochastic dom} implies:
\begin{align}\gamma_2 |_{{\rm int}(R)} + \mu_+ |_{{\rm int}(R)} \preceq_{cvx(\vec v)} \mu_- |_{{\rm int}(R)} + \gamma_1|_{{\rm int}(R)} + \xi_R \label{eq: stochastic dom wow}
\end{align}
where $\xi_R$ is the non-negative measure corresponding to $\hat{B} \cdot  \mathbb{I}_{ \hat{A} \in {\rm int}(R) \wedge \hat{B} \notin  {\rm int}(R) }$ (scaled back appropriately by $\mu_+(X) = \mu_-(X)$).

\item {\bf Step (iia):} We will now apply a flipped version of Lemma~\ref{conditionallem}, for convex non-increasing functions,\footnote{It is easy to verify that the guarantees of the lemma remain the same except the third guarantee changes to ``componentwise smaller than.''} to the convex function $u(x) - \sum_i x_i$.\footnote{Notice that the partial derivatives are non-positive.}
We set random variables $A',B'$ distributed according to the measures $\gamma_1$ and $\gamma_2$.
Since $\gamma_2 \succeq_{cvx(-\vec 1)} \gamma_1$, we have that $B' \succeq_{cvx(-\vec 1)} A'$. Moreover, $\mathbb{E}[u(A') - \sum_i A_i'] = \mathbb{E}[u(B') - \sum_i B_i']$ from Equation~(\ref{eq:u-minus-l1}) shown above.

We choose the function $g(x) - \vec 1$ as the subgradient of $u(x) - \sum_i x_i$.

\smallskip {\bf Step (iib):} Fixing any region $R$ and the corresponding ${\rm int}(R)$, ${\rm cl}(R)$ and $\vec{v}$ as above, we mirror the arguments of Step (i). Now, the version of Lemma~\ref{conditionallem} for non-increasing functions implies that there exist random variables $\hat{A}' =_{st} A'$ and $\hat{B}' =_{st} B'$ such that, almost surely:
\begin{align}
  &\mathbb{E}[\hat{B}'|\hat{A}'] \le \hat{A}'; \label{eq: Strassen condition 1'}\\
  &p_i < 1 \implies \mathbb{E}[\hat{B}_i'|\hat{A}'] \cdot \mathbb{I}_{\hat{A}' \in {\rm int}(R)} = \hat{A}_i' \cdot \mathbb{I}_{\hat{A}' \in {\rm int}(R)}. \label{eq: Strassen condition'}
\end{align}
Equations~\eqref{eq: Strassen condition 1'}, \eqref{eq: Strassen condition'} and Lemma~\ref{strassenv} imply that
\begin{align}
  \hat{A}' \cdot \mathbb{I}_{\hat{A}' \in {\rm int}(R)} \preceq_{cvx(\vec v)} \hat{B}' \cdot \mathbb{I}_{\hat{A}' \in  {\rm int}(R)} \label{eq: region stochastic dom'}
  \end{align}
and, hence,
\begin{align}
  \gamma_1 |_{{\rm int}(R)}  \preceq_{cvx(\vec v)}  \gamma_2|_{{\rm int}(R)} + \xi_R', \label{eq: stochastic dom wow2}
\end{align}
where similarly to our derivation above $\xi_R'$ is the non-negative measure corresponding to $\hat{B}' \cdot  \mathbb{I}_{ \hat{A}' \in {\rm int}(R) \wedge \hat{B}' \notin  {\rm int}(R) }$.
\end{itemize}

\noindent We now combine the results of Steps (i) and (ii) to finish the proof. Combining~\eqref{eq: stochastic dom wow} and~\eqref{eq: stochastic dom wow2}, we get that:
\begin{align} \mu_+ |_{{\rm int}(R)} \preceq_{cvx(\vec v)} \mu_- |_{{\rm int}(R)}  + \xi_R + \xi_R'. \label{eq:sochastic dom wow final}
\end{align}
From Proposition~\ref{prop:equal-mass}, it must hold that
\begin{align*} \mu_+ |_{{\rm int}(R)} (X) = \mu_- |_{{\rm int}(R)}(X)  + \xi_R(X) + \xi_R'(X). 
\end{align*}
Summing over all regions and noticing that $\sum_R \mu_- |_{{\rm int}(R)}(X) = \mu_-(X)$, from Lemma~\ref{lem:zero-mass}, we get that
\begin{align*} \mu_+(X)-\mu_+(Z) = \mu_-(X) + \sum_R(\xi_R (X) + \xi_R'(X)). 
\end{align*}
But $\mu_+(X) = \mu_-(X)$, hence $\mu_+(Z) = \sum_R(\xi_R (X) + \xi_R'(X)) = 0$, as all of $\mu_+$, $\xi_R$ and $\xi_R'$ are non-negative. Therefore, we can rewrite the property~\eqref{eq:sochastic dom wow final} as:
$$\mu_+ |_{R} \preceq_{cvx(\vec v)} \mu_- |_{R}.$$

 \section{Missing Proofs of Section~\ref{bundlingsection} - Theorem~\ref{nuniform}}
\label{hypercubeappendix}

In this appendix we complete the proof of Theorem~\ref{nuniform}.

\begin{prevproof}{Lemma}{matching}
We define the mapping $\varphi: A \rightarrow B$ by $\varphi(x) = y$, where
$$y_1 = \left[1- \rho \left(1 - (1-x_n)^{n-1}\right) \right]^{1/(n-1)}; \qquad y_i = \frac{x_i-x_n}{1-x_n} \cdot y_1   \,\,\, \text{ for $i>1$ }.$$
We first claim that $\varphi$ is a bijection. As $x_n$ ranges from $0$ to $1-   \left(\frac{\rho-1}{\rho} \right)^{1/(n-1)}$, we see that $y_1$ ranges from 1 to 0, and thus there is a bijection between valid $y_1$ values and valid $x_n$ values. Furthermore, for any fixed $y_1$ and $x_n$, there is a bijection between $x_i$ and $y_i$ for $i = 2,\ldots, n-1$. (By varying $x_i$ between $x_n$ and $1$ we can achieve all values of $y_i$ between 0 and $y_1$.) Furthermore, for any fixed $y_1$ and $x_n$ the mapping from $x_i$ to $y_i$ is an increasing function of $x_i$, and therefore for all $x \in A$ we have $y_1 \in [0,1]$ and $y_1 \geq y_2 \geq \cdots \geq y_n = 0$. Thus, $\varphi$ is a bijection between $A$ and $B$.
Next, we claim that for any $x \in A$, it holds that $x$ is componentwise at least as large as $\varphi(x)$. Since $x_1 = 1$, it trivially holds that $x_1 \geq \varphi_1(x)$. Fix a value of $x_n$ (and hence of $y_1$), and consider the bijection $g: [x_n, 1] \rightarrow [0, y_1]$ given by $g(z) = y_1(z-x_n)/(1-x_n)$. We must show that $z - g(z) \geq 0$ for all $z \in [x_n, 1]$. This follows from noticing that $z-g(z)$ is a linear function of $z$ and both $x_n - g(x_n) = x_n$ and $1 - g(1) = 1 - y_1$ are nonnegative. 

We now show that $\varphi$ scales surface measure of every measurable $S \subset A$ by a factor of $1/\rho$. Instead of directly analyzing surface measures, it suffices to prove that the function $\varphi' : W \rightarrow W$ scales volumes by $\rho$, where $W \subset \mathbb{R}^{n-1}$ is the set $\{w : 1 \geq w_1 \geq \cdots \geq w_{n-1} \geq 0 \}$ and $\varphi'(w)$ drops the last (constant) coordinate of $\varphi(1,w_1,\ldots,w_{n-1})$ and then (for notational convenience) permutes the first coordinate to the end.  That is,
$$\varphi'(w_1,\ldots, w_{n-1}) = \left( \frac{w_1 - w_{n-1}}{1-w_{n-1}}  z(w_{n-1}), \ldots,    \frac{w_{n-2} - w_{n-1}}{1-w_{n-1}}  z(w_{n-1}) ,   z(w_{n-1})  \right)$$
where $z(w_{n-1}) = \left[1- \rho \left(1 - (1-w_{n-1})^{n-1}\right) \right]^{1/(n-1)}$.

We now analyze the determinant of the Jacobian matrix $J$ of $\varphi'$. We notice that the only non-zero entries of $J$ are the diagonals and the rightmost column. In particular, $J$ is upper triangular, and therefore its determinant is the product of its diagonal entries. We therefore compute
\begin{align*}
\hspace{-0.7in} det(J) &= \left(\frac{z(w_{n-1})}{1-w_{n-1}}\right)^{n-2} \cdot \frac{\partial}{\partial w_{n-1}}\left[1- \rho \left(1 - (1-w_{n-1})^{n-1}\right) \right]^{1/(n-1)}\\
&= \left(\frac{z(w_{n-1})}{1-w_{n-1}}\right)^{n-2} \cdot \frac{-1}{n-1}\left(z(w_{n-1})^{-(n-2)}\cdot \rho\cdot (n-1)(1-w_{n-1})^{n-2} \right)= -\rho
\end{align*}
as desired.

Lastly, suppose $y_1 \leq \epsilon$. Then $ \left[1- \rho \left(1 - (1-x_n)^{n-1}\right) \right]^{1/(n-1)} \leq \epsilon$
 and thus
 $x_n \geq 1- \left( \frac{\epsilon^{n-1} +\rho-1}{\rho}  \right)^{1/(n-1)}.$
\end{prevproof}

\begin{prevproof}{Theorem}{nuniform}
We now complete the proof of Theorem~\ref{nuniform}. Fix the dimension $n$. For any value of $c$, the transformed measure on the hypercube $(c,c+1)^n$ we obtain  is as follows:
\begin{itemize}
	\item A point mass of $+1$ at $(c,c,\ldots, c)$.
	\item Mass of $-(n+1)$ uniformly distributed throughout the interior.
	\item Mass of $-c$ distributed on each surface $x_i = c$ of the hypercube.
	\item Mass of $c+1$ distributed on each surface $x_i = c+1$ of the hypercube.
\end{itemize}
For notational convenience when checking the stochastic dominance properties of Theorem~\ref{bundlingtheorem}, we will shift the hypercube to the origin. That is, we will consider instead the measure $\mu^c$ on $[0,1]^n$ which has mass $+1$ at the origin, mass of $-c$ on each each surface $x_i = 0$, et cetera. It is important to notice that the mass that $\mu$ assigns to the interior of $[0,1]^n$ and to the origin do not depend on $c$, while the mass on each surface is a function of $c$.

For any $h \in (0,1)$, define the region  $Z(h) = \{x \in [0,1]^n : \Lone{x} \leq h \}$. For any fixed $c_0$, it holds that $\mu^{c_0}_+(Z(h)) = 1$ for all $h \in (0,1)$ and there exists a small enough $h'>0$ such that  $\mu^{c_0}_-(Z(h')) < 1$.    Since for this fixed $h'$ it holds that $\mu_-^c(Z(h'))$ increases with $c$ (and becomes arbitrarily large as $c$ becomes large), there must exist a $c' > c_0$ such that $\mu_-^{c'}(Z(h')) = 1$, and thus $\mu^{c'}(Z(h')) = 0$. We can therefore pick a decreasing function $p^* : \mathbb{R}_{\geq 0} \rightarrow (0,1)$ such that, for all sufficiently large $c$, $\mu^c(Z(p^*(c))) = 0$.\footnote{Our intention is to argue that for $c$ large enough, the optimal mechanism will be grand bundling for a price of $p^*(c) + c$, where the additive $+c$ term comes from our shift of the hypercube to the origin.} As argued above, for any small enough $h'>0$ there exists a $c'$ such that $\mu_-^{c'}(Z(h')) = 1$ and thus $p^*(c') = h'$. It follows that $p^*(c) \rightarrow 0$ as $c \rightarrow \infty$.

For all $c$, define the following subsets of $[0,1]^n$:
$$Z_{c} = \left\{x : \Lone{x} \leq p^*(c) \right\}; \qquad W_{c} = \left\{x : \Lone{x} \geq p^*(c) \right\}.$$
We notice that $\mu^c_+(Z_c \cap W_c) = \mu^c_-(Z_c \cap W_c) = 0$. By construction, for large enough $c$ we have $\mu^c(Z_c) = 0$. In addition, the only positive mass in $Z_c$ is at the origin, and thus $\mu^c_-|_{Z_c} \succeq_{cvx} \mu^c_+|_{Z_c}$.

To apply Theorem~\ref{bundlingtheorem}, it remains to show that, for sufficiently large $c$, $\mu^c_+|_{W_c} \preceq_{cvx(- \vec 1)}  \mu^c_-|_{W_c}$. To prove this, we partition $W_c$ into $2(n! + 1)$ disjoint\footnote{For notational simplicity, our regions overlap slightly, although the overlap always has zero mass under both $\mu^c_+$ and $\mu^c_-$.} regions, $P_0, P_{\sigma_1}, \ldots, P_{\sigma_{n!}}$ and $N_0, N_{\sigma_1}, \ldots, N_{\sigma_{n!}}$, where $\sigma_j$ is a permutation of $1, \ldots, n$. This partition will be such that $\cup_j P_j$ contains the entire support of $\mu^c_+|_{W_c}$ and $\cup_j N_j$ contains the entire support of $\mu^c_-|_{W_c}$. We will show that $\mu^c_+|_{P_j} \preceq_{cvx(- \vec 1)} \mu^c_-|_{N_j}$ for all $j$, thereby proving $\mu^c_+|_{W_c} \preceq_{cvx(- \vec 1)} \mu^c_-|_{W_c}$.

For every permutation $\sigma$ of $1, \ldots, n$, define:
\begin{align*}
P'_\sigma &= \left\{x : 1 = x_{\sigma(1)} \geq x_{\sigma(2)} \geq \cdots \geq x_{\sigma(n)} \geq 0 \textrm{ and } x_{\sigma(n)} \leq 1 - \left(\frac{1}{c+1} \right)^{1/(n-1)}  \right\}\\
N'_\sigma &= \left\{y : 1 \geq y_{\sigma(1)} \geq \cdots \geq y_{\sigma(n-1)} \geq y_{\sigma(n)} = 0 \right\}
\end{align*}
Denote by $\rho  \triangleq (c+1)/c$  the ratio between the surface densities of $\mu^c_+$ and $\mu^c_-$ on $P'_\sigma$ and $N'_\sigma$, respectively, and let $\varphi_\sigma : P'_\sigma \rightarrow N'_\sigma$ be the bijection given by Lemma~\ref{matching}. By construction, $\mu^c_+(S) = \mu^c_-(\varphi_\sigma(S))$ for all measurable $S \subseteq P'_\sigma$.

Denote $N_\sigma \triangleq N'_\sigma \setminus Z_c$ and $P_\sigma \triangleq \varphi^{-1}(N_\sigma)$. By construction, $\varphi$ is a bijection between $P_\sigma$ and $N_\sigma$, preserving the respective the measures $\mu^c_+$ and $\mu^c_-$, such that for all $x \in P_\sigma$, $x$ is componentwise at least as large as $\varphi(x)$. Therefore, by Strassen's theorem, $\mu^c_+|_{P_\sigma} \preceq_{cvx(- \vec 1)} \mu^c_-|_{N_\sigma}$.
Lastly, we define
$$
P_0 = \left\{x \in [0,1]^n : x_i = 1 \textrm{ for some } i \right\} \setminus \left(\bigcup_\sigma P_\sigma \right); \qquad N_0 = (0,1)^n \setminus Z_c.
$$
$P_0$ consists of all points on the outer surface of the hypercube which have not yet been matched to any $N_\sigma$, and $N_0$ consists of all points on which $\mu^c_-$ is nontrivial which have not yet been matched.\footnote{All other points on which $\mu^c_-$ is nontrivial have been matched either to the origin (if the point lies in $Z_c$), or to some point in $P_\sigma$ (if the point lies in $N'_\sigma \setminus Z_c$).} It therefore remains only to show that $\mu^c_+|_{P_0} \preceq_{cvx(- \vec 1)} \mu^c_-|_{N_0}$.

 We claim that, for large enough $c$, $P_0$ only contains points with all coordinates greater than $3/4$. Indeed:
\begin{itemize}
\item  Every $x$ with $x_i = 1$ but some  $x_j < 1 - \left(\frac{1}{c+1} \right)^{1/(n-1)}$ is in some $P'_{\sigma}$.
\item For large $c$, every $x$ with $x_i = 1$ but some $x_j \leq 3/4$ is in some $P'_\sigma$.
\item We claim that for large $c$, every $x \in P'_\sigma \setminus P_\sigma$ has all coordinates at least $3/4$. Indeed, for every $x \in P'_\sigma \setminus P_\sigma$, it must be that $\varphi(x) \in Z_c$, and thus $\Lone{\varphi(x)} \leq p^*(c)$. By Lemma~\ref{matching}, we have $x_{\sigma(n)} \geq 1- \left( \frac{p^*(c)^{n-1} +\rho-1}{\rho}  \right)^{1/(n-1)}$. As $c$ gets large, $\rho \rightarrow 1$ and $p^*(c) \rightarrow 0$. Thus, for sufficiently large $c$, we have $x \in P'_\sigma \setminus P_\sigma$ implies $x_{\sigma(n)} \geq 3/4$. Since $x_{\sigma(n)}$ is the smallest coordinate of $x$, it follows that all coordinates of any $x \in P'_\sigma \setminus P_\sigma$ are greater than $3/4$.
\item Thus, for sufficiently large $c$, every $x$ with $x_i = 1$ but some $x_j < 3/4$ lies in some $P_\sigma$, and hence does not lie in $P_0$.
\end{itemize}
By construction,  $\mu^c_-|_{N_0}$ and $ \mu^c_+|_{P_0}$ have the same total mass. Consider independent random variables $X$ and $Y$ corresponding to $\mu^c_-|_{N_0}$ and $ \mu^c_+|_{P_0}$, respectively, where we scale both measures so that they are probability distributions. By Lemma~\ref{strassen}, it suffices to show that for sufficiently large $c$, $Y \geq \mathbb{E}[X]$ almost surely.\footnote{In general, to prove second order dominance we might need to nontrivially couple $X$ and $Y$. In this case, however, choosing independent random variables suffices.} Since $\mu^c_+|_{P_0}$ is supported on $P_0$, we need only show that all coordinates of $\mathbb{E}[X]$ are less than 3/4. We recall that $\mu^c_-$ assigns a total mass of $n+1$, distributed uniformly, to the interior of the hypercube. As $c$ gets large, $p^*(c)$ approaches 0, and thus 
$$\frac {\mu^c_-(Z_c \cap (0,1)^n )}{\mu^c_-((0,1)^n)} \rightarrow 0$$ 
For large $c$, therefore, $\mathbb{E}[X]$ becomes arbitrarily close to the center of the hypercube, which is the point with all coordinates equal to 1/2. Therefore we have $$\mu^c_+|_{P_0} \preceq_{cvx(- \vec 1)} \mu^c_-|_{N_0}$$
\end{prevproof}

\section{Supplementary Material for Section~\ref{weakstructural}}\label{weakstructuralappendix}

\begin{prevproof}{Claim}{zerosetICIR}
It is obvious that $u_Z$ is non-negative. 
To show that $u_Z$ is non-decreasing, it  suffices to prove that $u_Z(x) \geq u_Z(y)$ for $x, y \in X \setminus Z$ with $x$ component-wise greater than or equal to $y$.  Let $z_x \in Z$ be the closest point to $x$. Denote by $z_{y}$ the point with each coordinate being the component-wise minimum of $z_x$ and $y$. Since $Z$ is decreasing, $z_{y} \in Z$. We now compute
$$u_{Z}(x) = \Lone{z_x - x} = \sum_i |(z_x)_i - x_i| \geq \sum_i | \min\{(z_x)_i, y_i\} - y_i | = \Lone{z_{y}  - y } \geq u_Z(y)$$
and thus $u_Z$ is non-decreasing.

We will now show that $u_Z$ is convex. Pick arbitrary $x, y \in X$. Denote by $z_x$ and $z_y$ points in $Z$ such that $u_Z(x) = \Lone{x-z_x}$ and $u_Z(y) = \Lone{y-z_y}$. Since $Z$ is convex, the point $(z_x+z_y)/2$ is in $Z$. Thus
$$
u_Z\left( \frac{x+y}{2} \right) \leq \left\| \frac{x+y}{2} - \frac{z_x+z_y}{2} \right\|_1 \leq \frac{\Lone{x-z_x} + \Lone{y-z_y}}{2}  = \frac{u_Z(x)+u_Z(y)}{2}
$$
and therefore $u_Z$ is convex.

Lastly, we verify that $u_Z$ has Lipschitz constant at most 1. Indeed,
$$
u_Z(x) - u_Z(y) \leq \Lone{x - z_y} - u_Z(y) = \Lone{x-z_y} - \Lone{y-z_y} \leq \Lone{x-y}.
$$
\end{prevproof}

\section{Supplementary Material for Sections~\ref{canonicalpartitiontheorem} and~\ref{sec:further examples}}

\subsection{Verifying Stochastic Dominance - Proof of Lemma~\ref{regionthm}}\label{twoitemappendix}

We begin with the standard result that a sufficient condition for first-order stochastic dominance is that one measure assigns more mass than the other to all increasing sets.

\begin{claim}\label{increasingsetssuffice}
Let $\alpha, \beta$ be positive finite Radon measures on $\mathbb{R}^n_{\ge 0}$ with $\alpha(\mathbb{R}^n_{\geq 0}) = \beta(\mathbb{R}^n_{\geq 0})$. A necessary and sufficient condition for $\alpha \succeq_1 \beta$ is that for all increasing\footnote{An increasing set $A \subset \mathbb{R}^n_{\geq 0}$ satisfies the property that for all $a, b \in \mathbb{R}^n_{\geq 0}$ such that $a$ is component-wise greater than or equal to $b$, if $b \in A$ then $a \in A$ as well.} measurable sets $A$, $\alpha(A) \geq \beta(A)$.
\end{claim}

\begin{prevproof}{Claim}{increasingsetssuffice}
{Without loss of generality assume that $\alpha({\mathbb{R}^n_{\ge 0}}) = \beta(\mathbb{R}^n_{\geq 0}) =1$.} 

It is obvious that the condition is necessary by considering the indicator function of any increasing set $A$. To prove sufficiency, suppose that the condition holds and that on the contrary, $\alpha$ does not stochastically dominate $\beta$. Then there exists an increasing, bounded, measurable function $f$ such that
$$\int f d \beta - \int f d \alpha > 2^{-k+1}$$
for some positive integer $k$. Without loss of generality, we may assume that $f$ is nonnegative, by adding the constant of $-f(0)$ to all values. We now define the function $\tilde{f}$ by point-wise rounding $f$ upwards to the nearest multiple of $2^{-k}$. Clearly $\tilde{f}$ is increasing, measurable, and bounded. Furthermore, we have
$$\int \tilde{f} d\beta- \int \tilde{f}d\alpha \geq \int f d \beta- \int f d \alpha - 2^{-k} > 2^{-k+1} - 2^{-k} > 0.$$

We notice, however, that $\tilde{f}$ can be decomposed into the weighted sum of indicator functions of increasing sets. Indeed, let $\{r_1,\ldots, r_m\}$ be the set of all values taken by $\tilde{f}$, where $r_1 > r_2 > \cdots > r_m$. We notice that, for any $s\in \{1,\ldots,m\}$, the set $A_s = \{z : \tilde{f}(z) \geq r_s\}$ is increasing and measurable. Therefore, we may write
$$\tilde{f} = \sum_{s=1}^m (r_s-r_{s-1}) I_s$$
where $I_s$ is the indicator function for $A_s$  and where we set $r_0 = 0$. We now compute
$$\int \tilde{f} d\beta= \sum_{s=1}^m (r_s-r_{s-1})\beta(A_s) \leq \sum_{s=1}^m(r_s - r_{s-1}) \alpha(A_s) = \int \tilde{f} d\alpha,$$
contradicting the fact that $\int \tilde{f} d \beta > \int \tilde{f} d \alpha$.
\end{prevproof}

Due to Claim~\ref{increasingsetssuffice}, to verify that  a measure $\alpha$ stochastically dominates  $\beta$ in the first order, we must ensure that $\alpha(A) \geq \beta(A)$ for all increasing measurable sets $A$. This verification might still be difficult, since an increasing set can have fairly unconstrained structure. In Lemma~\ref{finiteunions} we simplify this task by showing that we need not verify the inequality for all increasing $A$, but rather only for a special class of increasing subsets.

\begin{definition}
For any $z \in \mathbb{R}^n_{\geq 0}$, we define the \emph{base rooted at $z$} to be
$$B_z \triangleq \{ z': z \preceq z'\},$$
the minimal increasing set containing $z$, where the notation $z \preceq z'$ denotes that every component of $z$ is at most the corresponding component of $z'$. 
\end{definition}
We denote by $Q_k$ to be the set of points in $\mathbb{R}^n_{\geq 0}$ with all coordinates multiples of $2^{-k}$.

\begin{definition}
An increasing set $S$ is \emph{$k$-discretized} if $S = \bigcup_{z \in S \cap Q_k} B_z$. A \emph{corner} $c$ of a $k$-discretized set $S$ is a point $c \in S \cap Q_k$ such that there does not exist  $z \in S\setminus \{c\}$ with $z \preceq c$.
\end{definition}

\begin{lemma}\label{finitelymanycorners}
Every $k$-discretized set $S$ has only finitely many corners. Furthermore, $S = \cup_{c \in \mathcal{C}}B_c$, where $\mathcal{C}$ is the collection of corners of $S$.
\end{lemma}

\begin{prevproof}{Lemma}{finitelymanycorners}
We prove that there are finitely many corners by induction on the dimension, $n$. In the case $n=1$ the result is obvious, since if $S$ is nonempty it has exactly one corner. Now suppose $S$ has dimension $n$. Pick some corner $\hat{c} = (c_1, \ldots, c_n) \in S$. We know that any other corner must be strictly less than $\hat{c}$ in some coordinate. Therefore,
$$|\mathcal{C}| \leq {1 +} \sum_{i=1}^n \left| \left\{c \in \mathcal{C} \textrm{ s.t. } c_i < \hat{c}_i \right\}  \right| = {1 +} \sum_{i=1}^n \sum_{j = 1}^{2^k\hat{c}_i} \left| c \in \mathcal{C} \textrm{ s.t. } c_i = \hat{c}_i-2^{-k}j  \right|.$$
By the inductive hypothesis, we know that each set $\left\{ c \in \mathcal{C} \textrm{ s.t. } c_i = \hat{c}_i-2^{-k}j  \right\}$ is finite, since it is contained in the set of corners of the $(n-1)$-dimensional {subset of $S$ whose points have $i^{th}$ coordinate $\hat{c}_i - 2^{-k}j$.} Therefore, $|\mathcal{C}| $ is finite.

To show that $S =  \bigcup_{c \in \mathcal{C}}B_c$, pick any $z \in S$. Since $S$ is $k$-discretized, there exists a $b \in S \cap Q_k$ such that $z \in B_b$. If $b$ is a corner, then $z$ is clearly contained in $\bigcup_{c \in \mathcal{C}}B_c$. If $b$ is not a corner, then there is some other point $b' \in S \cap Q_k$ with $b' \preceq b$. If $b'$ is a corner, we're done. Otherwise, we repeat this process at most $2^k \sum_j b_j$ times, after which time we will have reached a corner $c$ of $S$. By construction, we have $z \in B_c$, as desired.\end{prevproof}

We now show that, to verify that one measure dominates another on all increasing sets, it suffices to verify that this holds for all sets that are the union of finitely many bases. 

\begin{lemma}\label{finiteunions}
Let $g, h : \mathbb{R}^n_{\geq 0} \rightarrow \mathbb{R}_{\geq 0}$ be bounded integrable functions such that $\int_{\mathbb{R}^n_{\geq 0}} g(x) d x$ and $\int_{\mathbb{R}^n_{\geq 0}} h(x)d x$ are finite. Suppose that, for all finite collections $Z$ of points in $\mathbb{R}^n_{\geq 0}$, we have
$$\int_{\bigcup_{z \in Z}B_z}g({x})d{x} \geq \int_{\bigcup_{z \in Z}B_z}h({x})d{x}.$$
Then for all increasing sets $A \subseteq \mathbb{R}^n_{\geq 0}$,
$$\int_A g({x})d{x} \geq \int_A h({x})d{x}.$$
\end{lemma}

\begin{prevproof}{Lemma}{finiteunions}
Let $A$ be an increasing set. We clearly have $A = \bigcup_{z \in A} B_z$. For any point $z \in \mathbb{R}^n_{\geq 0}$, denote by $z^{n,k}$ the point in $\mathbb{R}^n_{\geq 0}$ such that for each component $i$, the $i^{th}$ component of $z^{n,k}$ is the maximum of 0 and $z_i - 2^{-k}$.

We define the following two sets, which we think of as approximations of $A$:
$$A_k^l \triangleq \bigcup_{z \in A \cap Q_k} B_z; \qquad A_k^u \triangleq \bigcup_{z \in A \cap Q_k} B_{z^{n,k}}.$$
It is clear that both $A_k^l$ and $A_k^u$ are $k$-discretized. Furthermore, for any $z \in A$ there exists a $z' \in A \cap Q_k$ such that each component of $z'$ is at most $2^{-k}$ more than the corresponding component of $z$. Therefore $A_k^l \subseteq A \subseteq A_k^u.$

We now will bound
$$\int_{A_k^u} g(x) d{x} - \int_{A_k^l} g(x) d{x}.$$
Let
$$W_k = \left\{z \in \mathbb{R}^n_{\geq 0} : z_i > k \textrm{ for some } i   \right\}; \qquad W^c_k = \left\{z \in \mathbb{R}^n_{\geq 0} : z_i \leq k \textrm{ for all } i \right\}.$$
The set $W^c_{k}$ contains all points which are lie inside in a box of side length $k$ rooted at the origin, and $W_k$ contains all points outside of this box. We have the immediate (loose) bound that
$$\int_{A_k^u \cap W_k} g d{x} - \int_{A_k^l \cap W_k}g d{x} \leq \int_{W_k} g d{x}.$$
Furthermore, since $\lim_{k \rightarrow \infty}\int_{W^c_k}g d{x} = \int_{\mathbb{R}^n_{\geq 0}}g d{x}$,  we know that $\lim_{k \rightarrow \infty} \int_{W_k}g d{x} = 0$. Therefore, 
$$\lim_{k \rightarrow \infty}\left(\int_{A_k^u \cap W_k} g d{x} - \int_{A_k^l \cap W_k}g d{x}\right) = 0.$$
Next, we bound
$$\int_{A_k^u \cap W^c_k} g d{x} - \int_{A_k^l \cap W^c_k}g d{x} \leq |g|_{\sup} \left(V(A_k^u \cap W^c_k) - V(A_k^l \cap W^c_k)  \right)$$
where $|g|_{\sup} < \infty$ is the supremum of $g$, and $V(\cdot)$ denotes the Lebesgue measure.

For each $m \in \{1,\ldots, n+1\}$ and  $z \in \mathbb{R}^n_{\geq 0}$, we define the point $z^{m,k}$ by:
$$z^{m,k}_i =
\begin{cases}
\max\{0, z_i - 2^{-k} \} & \mbox{ if } i < m\\
z_i & \mbox{ otherwise}
\end{cases}$$
and set
$$A_k^m \triangleq \bigcup_{z \in A \cap Q_k}B_{z^{m,k}}.$$
We have, by construction, $A_k^l = A_k^1$ and $A_k^u = A_k^{n+1}$. Therefore,
$$V(A_k^u \cap W^c_k) - V(A_k^l \cap W^c_k) = \sum_{m=1}^n \left(V(A_k^{m+1} \cap W_k^c)-V(A_k^{m} \cap W_k^c) \right).$$
We notice that, for any point $ (z_1,z_2, \ldots, z_{m-1},z_{m+1},\ldots, z_n) \in [0,k]^{n-1}$, there is an interval $I$ of length at most $2^{-k}$ such that $$(z_1,z_2,\ldots,z_{m-1},w,z_{m-2},\ldots,z_n) \in (A_k^{m+1} \setminus A_k^{m}) \cap W_k^c$$
if and only if $w \in I$. Therefore,
\begin{align*}V(A_k^{m+1} \cap W_k^c)&-V(A_k^{m} \cap W_k^c)\\
 &\leq \int_0^k \cdots \int_0^k \int_0^k \cdots \int_0^k 2^{-k}dz_1\cdots dz_{m-1}dz_{m+1}\cdots dz_n = 2^{-k}k^{n-1}.
\end{align*}
We thus have the bound
$$
|g|_{\sup} \left(V(A_k^u \cap W^c_k) - V(A_k^l \cap W^c_k)  \right) \leq
|g|_{\sup}\sum_{m=1}^n 2^{-k}k^{n-1} = n|g|_{\sup}2^{-k}k^{n-1}
$$
and therefore
\begin{align*}
\int_{A_k^u} g  d{x} - \int_{A_k^l} gd{x}  &= \int_{A_k^u \cap W_k}g d{x} - \int_{A_k^l \cap W_k} gd{x}  + \int_{A_k^u \cap W^c_k} g d{x} - \int_{A_k^l \cap W^c_k} g d{x} \\
&\leq  \left( \int_{A_k^u \cap W_k}g d{x} - \int_{A_k^l \cap W_k} gd{x}  \right) + n|g|_{\sup}2^{-k}k^{n-1}.
\end{align*}
In particular, we have
$$\lim_{k \rightarrow \infty} \left(\int_{A_k^u} g d{x} - \int_{A_k^l} g d{x}\right) = 0.$$
Since $\int_{A_k^u} g d{x} \geq \int_A g d{x} \geq  \int_{A_k^l}g d{x}$, we have
$$\lim_{k \rightarrow \infty} \int_{A_k^u}g d{x} = \int_A g d{x} = \lim_{k \rightarrow \infty} \int_{A_k^l}g d{x}.$$
Similarly, we have
$$\int_A h d{x} = \lim_{k \rightarrow \infty}\int_{A_k^l} h d{x}$$
and thus
$$\int_A (g-h) d{x} = \lim_{k \rightarrow \infty}\left(\int_{A_k^l} g d{x} - \int_{A_k^l} h d{x} \right).$$
Since $A_k^l$ is $k$-discretized, it has finitely many corners. Letting $Z_k$ denote the corners of $A_k^l$, we have $A_k^l = \bigcup_{z \in Z_k}B_z$, and thus by our assumption $\int_{A_k^l}gd{x} - \int_{A_k^l}hd{x} \geq 0$
for all $k$. Therefore $\int_A(g-h)d{x} \geq 0$, as desired.
\end{prevproof}

We are now ready to prove Lemma~\ref{regionthm}.

\begin{prevproof}{Lemma}{regionthm}

We begin by defining, for any $a$ and $b$ with  $p_1 \leq a \leq b \leq q_1$, the function $\zeta_a^b : [p_2,q_2] \rightarrow \mathbb{R}$ by
$$\zeta_a^b(w_2) \triangleq \int_a^b(g(z_1,w_2)-h(z_1,w_2))dz_1.$$
This function $\zeta_a^b(w_2)$ represents the integral of $g - h$ along the vertical line from $(a,w_2)$ to $(b,w_2)$.
\begin{claim}\label{claim11}
If $(a,w_2) \in R$, then $\zeta_a^b(w_2) \leq 0$.
\end{claim}
\begin{prevproof}{Claim}{claim11}
The inequality trivially holds unless there exists a $z_1 \in [a, b]$ such that $g(z_1,w_2) > h(z_1,w_2)$, so suppose such a $z_1$ exists. It must be that $(z_1,w_2) \notin R$, since  both $g$ and $h$ are $0$ in $R$. Indeed, because $R$ is a decreasing set it is also true that $(\tilde{z}_1,w_2) \notin R$ for all $\tilde{z}_1 \geq z_1$. This implies by our assumption that
$$g(\tilde{z}_1,w_2) - h(\tilde{z}_1,w_2) = \alpha(\tilde{z}_1)\cdot \beta(w_2) \cdot \eta(\tilde{z}_1,w_2),$$
for all $\tilde{z}_1 \geq z_1$. Given that $g(z_1,w_2) > h(z_1,w_2)$ and that $\eta(\cdot,w_2)$ is an increasing function, we know that $g(\tilde{z}_1,w_2) \geq h(\tilde{z}_1,w_2)$ for all $\tilde{z}_1 \geq z_1$. Therefore, we have
$$\zeta_a^{z_1}(w_2)\leq \zeta_a^b(w_2) \leq \zeta_a^{q_1}(w_2).$$
We notice, however, that $\zeta_a^{q_1}(w_2) \leq 0$ by assumption, and thus the claim is proven.\end{prevproof}

We now claim the following:
\begin{claim}\label{signflip}
Suppose that $\zeta_a^b(w^*_2) > 0$ for some $w^*_2 \in [c_2, q_2)$. Then $\zeta_a^b(w_2) \geq 0$ for all $w_2 \in [w^*_2, q_2)$. 
\end{claim}
\begin{prevproof}{Claim}{signflip}
Given that $\zeta_a^b(w^*_2) > 0$, our previous claim implies that $(a,w^*_2)\not\in R$. Furthermore, since $R$ is a decreasing set and $w_2 \geq w^*_2$, follows that $(a,w_2) \not\in R$, and furthermore that $(c,w_2) \not\in R$ for any $c \geq a$ in $[c_1, q_1)$. Therefore, we may write
$$\zeta_a^b(w_2) = \int_a^b  (g(z_1,w_2) - h(z_1,w_2))dz_1 = \int_a^b (\alpha(z_1)\cdot \beta(w_2)\cdot \eta(z_1,w_2))dz_1.$$
Similarly, $(c,w^*_2) \not\in R$ for any $c \geq a$, so
$$\zeta_a^b(w^*_2) = \int_a^b (\alpha(z_1)\cdot \beta(w^*_2)\cdot \eta(z_1,w^*_2))dz_1.$$
Note that, since $\zeta_a^b(w^*_2) > 0$, we have $\beta(w^*_2) > 0$.
Thus, since $\eta$ is increasing,
\begin{align*}
\zeta_a^b(w_2) &\geq \int_a^b (\alpha(z_1)\cdot \beta(w_2) \cdot \eta(z_1, w^*_2))dz_1= \frac{\beta(w_2)}{\beta(w^*_2)} \zeta_a^b(w^*_2) \geq 0,
\end{align*}
as desired.\end{prevproof}

We extend $g$ and $h$ to all of $\mathbb{R}^2_{\geq 0}$ by setting them to be 0 outside of $\mathcal{C}$. By Claim~\ref{finiteunions}, to prove that $g \succeq_1 h$ it suffices to prove that $\int_A g dxdy \geq \int_A h dxdy$ for all sets $A$ which are the union of finitely many bases. Since $g$ and $h$ are 0 outside of $\mathcal{C}$, it  suffices to consider only bases $B_{z'}$ where $z' \in \mathcal{C}$, since otherwise we can either remove the base (if it is disjoint from $\mathcal{C}$) or can increase the coordinates of $z'$ moving it to $\cal C$ without affecting the value of either integral.

We now complete the proof of Lemma~\ref{regionthm} by induction on the number of bases in the union.
\begin{itemize}[label={},leftmargin=0pt]
	\item \textbf{Base Case.} We aim to show $\int_{B_r} (g-h)dxdy \geq 0$ for any $r = (r_1,r_2) \in \mathcal{C}$.  We have
	\begin{align*}
	\int_{B_r} (g-h)dxdy = \int_{r_2}^{q_2} \int_{r_1}^{q_1} (g-h)dz_1dz_2 = \int_{r_2}^{q_2} \zeta_{r_1}^{q_1}(z_2)dz_2.
	\end{align*}
	By Claim~\ref{signflip}, we know that either $\zeta_{r_1}^{q_1}(z_2) \geq 0$ for all $z_2 \geq r_2$, or $\zeta_{r_1}^{q_1}(z_2) \leq 0$ for all $z_2$ between $p_2$ and $r_2$. In the first case, the integral is clearly nonnegative, so we may assume that we are in the second case. We then have
	\begin{align*}
	\int_{r_2}^{q_2} \zeta_{r_1}^{q_1}(z_2)dz_2 \geq \int_{p_2}^{q_2} \zeta_{r_1}^{q_1}(z_2)dz_2 &=  \int_{p_2}^{q_2} \int_{r_1}^{q_1} (g-h)dz_1dz_2
	\\&= \int_{r_1}^{q_1} \int_{p_2}^{q_2} (g-h)dz_2dz_1.
	\end{align*}
By an analogous argument to that above, we know that either $\int_{p_2}^{q_2} (g-h)(z_1,z_2)dz_2$ is nonnegative for all $z_1 \geq r_1$ (in which case the desired inequality holds trivially) or is nonpositive for all $z_1$ between $p_1$ and $r_1$. We assume therefore that we are in the second case, and thus
\begin{align*}
\int_{r_1}^{q_1} \int_{p_2}^{q_2} (g-h)dz_2dz_1 \geq \int_{p_1}^{q_1} \int_{p_2}^{q_2} (g-h)dz_2dz_1 = \int_{\mathcal{C}}(g-h)dxdy,
\end{align*}
which is nonnegative by assumption.	
	\item \textbf{Inductive Step.}
	Suppose that we have proven the result for all sets which are finite unions of at most $k$ bases. Consider now a set
$$A =  \bigcup_{i=1}^{k+1}  B_{z^{(i)}}.$$
We may assume that all $z^{(i)}$ are distinct and that there do not exist distinct $z^{(i)}$, $z^{(j)}$ with $z^{(i)}$ component-wise less than $z^{(j)}$, since otherwise we could remove one such $B_{z^{(i)}}$ from the union without affecting the set $A$ and the desired inequality would follow from the inductive hypothesis.

We may therefore order the $z^{(i)}$ such that 
$$p_1\leq z^{(k+1)}_1 < z^{(k)}_1< z^{(k-1)}_1 < \cdots < z^{(1)}_1$$ 
$$p_2 \leq z^{(1)}_2 < z^{(2)}_2 < z^{(3)}_2 < \cdots < z^{(k+1)}_2.$$ 

\begin{figure}[!ht]
\begin{center}
\begin{tikzpicture}

\begin{axis}[ymin=1.1, ymax=3.3, xmin=0, xmax=1, xlabel=$z_1$, ylabel=$z_2$,  ytick pos=left, ytick={1.1}, yticklabels={$p_2$}, xtick={0},xticklabels={$p_1$}]

  \addplot+[color=gray, fill=gray!50, domain=0:2,mark=none]
 {1.9-2*x*x}
 \closedcycle;

\addplot[color=black, mark=*] coordinates{
(.1,2.3)
(.1,7)
};

\addplot[color=black, mark=none] coordinates{
(.1,2.3)
(.3,2.3)
};

\addplot[color=black, mark=none] coordinates{
(.3,2.3)
(.3,1.6)
};

\addplot[color=black, mark=*] coordinates{
(.3,1.6)
};

\addplot[color=black, mark=none] coordinates{
(.3,1.6)
(.5,1.6)
};

\addplot[color=black, mark=none] coordinates{
(.5,1.6)
(.5,1.4)
};

\addplot[color=black, mark=*] coordinates{
(.5,1.4)};

\addplot[color=black, mark=none] coordinates{
(.5,1.4)
(7,1.4)
};

\node at (axis cs:.1,2.2){$z^{(k+1)}$};
\node at (axis cs:.3,1.5){$z^{(k)}$};
\node at (axis cs:.5,1.3){$z^{(k-1)}$};
\node at (axis cs:0.15,1.4){$R$};

\end{axis}
\end{tikzpicture}
\end{center}
\caption{We show that either decreasing $z^{(k+1)}_2$ to $z^{(k)}_2$ or removing $z^{(k+1)}$ entirely decreases the value of $\int_A (f-g)$. In either case, we can apply our inductive hypothesis.}
\end{figure}
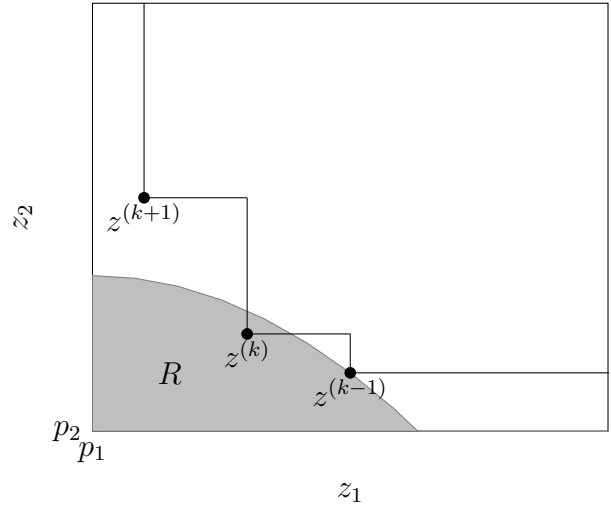

By Claim~\ref{signflip}, we know that one of the two following cases must hold:
\begin{itemize}[label={},leftmargin=0pt]
	\item \textbf{Case 1:} $\zeta_{z_1^{(k+1)}}^{z_1^{(k)}}(w_2) \leq 0$ for all $p_2 \leq w_2 \leq z_2^{(k+1)}$.
	
	In this case, we see that 
	$$\int_{z_2^{(k)}}^{z_2^{(k+1)}} \int_{z_1^{(k+1)}}^{z_1^{(k)}} (f-g)dz_1dz_2 = \int_{z_2^{(k)}}^{z_2^{(k+1)}}  \zeta_{z_1^{(k+1)}}^{z_1^{(k)}}(w)dw \leq 0.$$
	For notational purposes, we denote here by $(f-g)(S)$ the integral $\int_S (f-g)dz_1dz_2$ for any set $S$. We compute
	\begin{align*}
	\hspace{-1in}
	 (f-g)(A) & \geq (f-g)(A)\\
	 &\hspace{.2in}+ (f-g)\left(\left\{z: z_1^{(k+1)}\leq z_1 \leq z_1^{(k)} \textrm{ and } z_2^{(k)}\leq z_2 \leq z_2^{(k+1)} \right\} \right)\\
	&= (f-g)\left({ \bigcup_{i=1}^{k}  B_{z^{(i)}}  \cup B_{(z^{(k+1)}_1,z_2^{(k)})} }\right)\\
	&= (f-g)\left({ \bigcup_{i=1}^{k-1}  B_{z^{(i)}}  \cup B_{(z^{(k+1)}_1,z_2^{(k)})} }\right)
	\end{align*}
	where the last equality follows from $(z_1^{(k)},z_2^{(k)})$ being component-wise greater than or equal to $(z^{(k+1)}_1,z_2^{(k)})$. The inductive hypothesis implies that the quantity in the last line of the above derivation is $\ge 0$.

	\item \textbf{Case 2:} $\zeta_{z_1^{(k+1)}}^{z_1^{(k)}}(w_2) \geq 0$ for all $w_2 \geq z_2^{(k+1)}$.
	
	In this case, we have
	$$\int_{z_2^{(k+1)}}^{q_2} \int_{z_1^{(k+1)}}^{z_1^{(k)}}(f-g)dz_1dz_2 = \int_{z_2^{(k+1)}}^{q_2}  \zeta_{z_1^{(k+1)}}^{z_1^{(k)}}(w)dw \geq 0.$$
	Therefore, it follows that
	\begin{align*}
	\hspace{-1in}
	(f-g)(A) &= (f-g)\left( \bigcup_{i=1}^{k}  B_{z^{(i)}}  \right)\\
	 &\hspace{.2in}+ (f-g)\left( \left\{z: z_1^{(k+1)}\leq z_1 \leq z_1^{(k)} \textrm{ and }   z_2^{(k+1)} \leq z_2 \right\}  \right)\\
	&\geq (f-g)\left( \bigcup_{i=1}^{k}  B_{z^{(i)}}  \right) \geq 0,
	\end{align*}
where the final inequality follows from the inductive hypothesis.
\end{itemize}
\end{itemize}

\end{prevproof}

\subsection{Verifying Stochastic Dominance in Example~\ref{betaexample}}\label{betaappendix}

We sketch the application of Lemma~\ref{regionthm} for verifying that $\mu_+|_\mathcal{W} \succeq_1 \mu_-|_\mathcal{W}$ in Example~\ref{betaexample}.
We set $\mathcal{C} = [x_{\textrm{crit}},1]\times [y_{\textrm{crit}},1]$ and $\mathcal{R} = Z \cap \mathcal{C}$, so that $\mathcal{W} = \mathcal{C} \setminus R$.  We let $g$ and $h$ being the positive and negative parts of the density function of $\mu|_\mathcal{W}$, respectively, so that the density of $\mu|_{\mathcal{W}}$ is given by $g-h$. Since $Z$ lies below \textit{both} curves $S_{\textrm{top}}$ and $S_{\textrm{right}}$, we know that integrating the density of $\mu$ along any horizontal or vertical line outwards starting anywhere on the boundary of $Z$ yields a non-positive quantity, verifying the second condition of Lemma~\ref{regionthm}. In addition, on $\mathcal{W} = \mathcal{C} \setminus R$, we have
$$g(z_1,z_2) -h(z_1,z_2) = f_1(z_1)f_2(z_2)\left(\frac{1}{1-z_1}+\frac{1}{1-z_2} - 5 \right)$$
which satisfies the third condition of Lemma~\ref{regionthm}, as $1/(1-z_1)+1/(1-z_2) - 5 $ is increasing. Finally, we verify the first condition of Lemma~\ref{regionthm}  by integrating~$g-h$~over~$\cal C$. This integral is equal to $\mu(\mathcal{W}) = 0$ and thus all conditions of Lemma~\ref{regionthm} are satisfied.

\subsection{Uniqueness of Mechanism in Example~\ref{betaexample}}\label{app:uniqueness W}

To argue that the utility $u(x)$ is shared by all optimal mechanisms, we start by constructing an optimal solution $\gamma^*$ to the RHS of \eqref{eq:strong duality}. $\gamma^*$ needs to satisfy the complementary slackness conditions of Corollary~\ref{linearintegral} against any optimal solution $u^*$ to the LHS of~\eqref{eq:strong duality}. We will choose our solution $\gamma^*$ so that the complementary slackness conditions will imply $u^*=u$. Let us proceed with the choice of $\gamma^*$. Recall the canonical partition $Z \cup \mathcal{A} \cup \mathcal{B} \cup \mathcal{W}$ of the type space, identified above, and illustrated in Figure~\ref{betafig}. We define a solution $\gamma^*$ to the RHS of~\eqref{eq:strong duality} that separates into the four regions as follows (the optimality of this $\gamma^*$ follows easily by checking that it satisfies the complementary slackness conditions of Corollary~\ref{linearintegral} against $u$):

\paragraph{Region $Z$} Recall that, in region $Z$, we have $\mu|_Z \preceq_{cvx} 0$. Our solution $\gamma^*$ matches the $+1$ unit of mass sitting at the origin to the negative mass spread throughout region $Z$, by moving positive mass to coordinate-wise larger points and performing mean preserving spreads. By the complementary slackness conditions of Corollary~\ref{linearintegral} (see Remark~\ref{geometricremark} for intuition), it follows that $u^*(x)=0$, for any optimal solution $u^*$ to the LHS of~\eqref{eq:strong duality}. 

\paragraph{Regions $\cal A$ and $\cal B$} In regions $\cal A$ and $\cal B$ our solution $\gamma^*$ transports mass vertically and, respectively, horizontally. The complementary slackness conditions imply then that any optimal solution $u^*$ to the LHS of~\eqref{eq:strong duality} $u^*$ must change linearly in the second coordinate in region $\cal A$ and linearly in the first coordinate in region $\cal B$. 

\paragraph{Region $\cal W$} Finally, in region $\cal W$ we want to show that any optimal $u$ satisfies $|u(\vec x) - u(\vec y)| = \Lone{\vec x- \vec y}$ if $\vec x \ge \vec y$ coordinate-wise. This is not as straightforward as the previous 2 cases as we don't have an explicit description of the optimal dual solution. However, we can use Lemma~\ref{regionthm} to show that there exists a measure $\gamma^*$ which is optimal for the dual and matches types on the top right corner (with values $\approx (1,1)$) to types close to the bundling line (with values $x_1+x_2 \approx p^*$) which implies that any optimal function $u$ must be linear in $\cal W$.

By continuity, any optimal $u$ must be equal to $z_1+z_2-p^* = 0$ when $z_1+z_2 = p^*$. Moreover, it holds that $u(z) \le z_1+z_2-p^*$, because $u$ is $1$-Lipschitz. We will now show the reverse inequality by showing that $u(1,1) = 2-p^*$. Recall that the density of measure $\mu$ in region $\mathcal{W}$ is equal to:
$$\mu(z_1,z_2) = f_1(z_1)f_2(z_2)\left(\frac{1}{1-z_1}+\frac{1}{1-z_2} - 5 \right)$$
where $f_1(x) = f_2(x) = (1-x)$.
Lemma~\ref{regionthm} implied that $\mu_+|_\mathcal{W} \succeq_1 \mu_-|_\mathcal{W}$ but didn't give a transport map $\gamma$ constructively. To partially specify a transport map $\gamma$ that is optimal for the dual, we define for sufficiently small $\epsilon >0$ the measure $\mu'$ which has density
$$\mu'(z_1,z_2) = f_1(z_1)f_2(z_2)\left( \frac{1}{\epsilon} + \max \left(\frac{1}{1-z_2}, \frac{1}{1-z_1}\right) - 5 \right)$$
when $(z_1,z_2) \in [1-\epsilon,1]^2$ and $\mu'(z_1,z_2) =\mu(z_1,z_2)$ otherwise. In particular, $\mu'$ is obtained by removing some positive mass from $\mu$ in $[1-\epsilon,1]^2$ and thus $\mu'(\mathcal{W}) < \mu(\mathcal{W}) = 0$. Moreover, notice that we defined $\mu'$ so that $\frac {\mu'(z_1,z_2)}{f_1(z_1)f_2(z_2)}$ is still an increasing function. Now, let $R'$ be the region enclosed within the curves $s_1(x)$, $s_2(y)$, $x+y=p^*$ and $x+y=p'$ for $p'>p^*$ so that $\mu'(\mathcal{W} \setminus R') = 0$. This defines a decomposition of measure $\mu|_{\mathcal{W}}$ into two measures $\mu'|_{\mathcal{W}\setminus R'}$ and $\mu|_{\mathcal{W}} - \mu'|_{\mathcal{W}\setminus R'}$ of zero total mass (Figure~\ref{betafig2}).

\begin{figure}[!ht]
\centering
\begin{tikzpicture}
\begin{axis}[height=3in, width=3in, ymin=0, ymax=1, xmin=0, xmax=1,
  xtick pos=left, xtick={0.5,0.6666666}, xticklabels={$1/2$, $2/3$}, ytick pos=left, ytick={0.5,0.6666666666}, yticklabels={$1/2$,$2/3$}]
  
\addplot+[color=black, mark=none, style=densely dashed, domain=0:0.66666666666] {(2-3*x)/(4-5*x)};
\addplot+[color=black, mark=none, style=dotted, domain=0:0.5] {(2-4*x)/(3-5*x)};
\addplot[color=black, fill=gray!90,mark=none, samples=200, domain=0:0.5] 
{min( (2-3*x)/(4-5*x), (2-4*x)/(3-5*x), 0.7-x)}\closedcycle;
\addplot[color=black, fill=gray!50, mark=none, samples=200, domain=0:0.5] 
{min( (2-3*x)/(4-5*x), (2-4*x)/(3-5*x), 0.5534938-x)}\closedcycle;

\addplot[color=gray, mark=none]coordinates{
(.5534938-0.06187679,0.06187679)
(1,0.06187679)
};
\addplot[color=gray, mark=none]coordinates{
(0.06187679,.5534938-0.06187679)
(0.06187679,1)
};
\draw [color=black, fill=gray!90] (axis cs:.9,.9) rectangle (axis cs:1,1);

\node at (axis cs:0.03,0.85){$\mathcal{A}$};
\node at (axis cs:0.85,0.03){$\mathcal{B}$};
\node at (axis cs:0.6,0.65){$\mathcal{W}$};
\node at (axis cs:0.2,0.2){$Z$};
\node at (axis cs:0.31,0.31){$R'$};
\node at (axis cs:0.95,0.95){$\mathcal{H}$};
\end{axis}
\end{tikzpicture}
\caption{Decomposition of measure $\mu|_{\mathcal{W}}$ into measures $\mu'|_{\mathcal{W}\setminus R'}$ and $\mu|_{\mathcal{W}} - \mu'|_{\mathcal{W}\setminus R'}$. The dark shaded regions $R'$ and $\mathcal{H} = [1-\epsilon,1]^2$ show the support of $\mu|_{\mathcal{W}} - \mu'|_{\mathcal{W}\setminus R'}$.}\label{betafig2}
\end{figure}
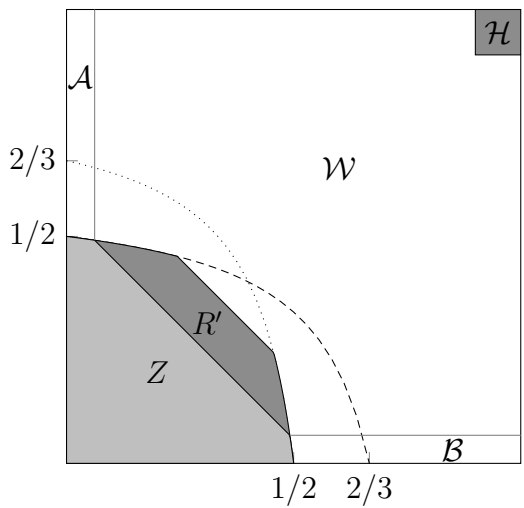

We apply Lemma~\ref{regionthm} for $\mu'$ in region $\mathcal{W} \setminus R'$ to get that $\mu'|_{\mathcal{W} \setminus R'} \succeq_{1} 0$. 
We also have that $(\mu - \mu')|_{\mathcal{W}} \succeq_{1} \mu|_{R'}$ since $(\mu - \mu')|_{\mathcal{W}}$ contains only positive mass supported on $[1-\epsilon,1]^2$ and every point in the support pointwise dominates every point in the support of $\mu|_{R'}$. Thus, there exists an optimal transport map $\gamma^*$ in region $\mathcal{W}$ such that
$\gamma^* = \gamma^{(i)} + \gamma^{(ii)}$ and $\gamma^{(i)}$ transports the mass $\mu'|_{\mathcal{W} \setminus R'}$ while $\gamma^{(ii)}$ transports mass arbitrarily from $(\mu - \mu')|_{\mathcal{W}}$ to $\mu|_{R'}$. Given such an optimal $\gamma^*$, the complementary slackness conditions of Corollary~\ref{linearintegral} imply that any feasible $u$ must satisfy $|u(\vec z) - u(\vec z')| = \Lone{\vec z- \vec z'}$ whenever mass is transfered from $\vec z$ to $\vec z'$. This can only happen if $u(1,1) = 2-p^*$ and implies that $u(\vec z) = z_1+z_2-p^*$ everywhere on $\mathcal{W}$.

 \section{Extending to Unbounded Distributions}\label{infiniteappendix}

Several results of this paper extend to unbounded type spaces, although such extensions impose additional technical difficulties. Here we briefly discuss how some of our results generalize.

We can often obtain a ``transformed measure'' (analogous to Theorem~\ref{setupclaim} even when type spaces are unbounded) using integration by parts. We wish to ensure, however, that the density function $f$ decays sufficiently quickly so that there is no ``surface term at infinity.'' For example, we may require that $\lim_{z_i \rightarrow \infty} f_i(z_i)z_i^2 \rightarrow 0$, as in \cite{DaskalakisDT13}. We note that without some conditions on the decay rate of $f$, it is possible that the supremum revenue achievable is infinite and thus no optimal mechanism exists.

Similar issues arise when integrating with respect to an unbounded measure $\mu$. It is helpful therefore to consider only measures $\mu$ such that $\int \Lone{x} d|\mu| < \infty$, to ensure that $\int u d\mu$ is finite for any utility function $u$. The measures in our examples satisfy this property. We can (informally speaking) attempt to extend this definition to unbounded measures (with regularity conditions such as $\int \Lone{x} d|\mu| < \infty$) by ensuring that whenever the ``smaller'' side has infinite value, so does the larger side.

Importantly, the calculations of  Lemma~\ref{weakduality} (weak duality) hold for unbounded $\mu$, provided $\int \Lone{x} d|\mu| < \infty$. Thus, tight certificates still certify optimality, even in the unbounded case. However, our strong duality proof relies on technical tools which require compact spaces, and thus these proofs do not immediately apply when $\mu$ is unbounded.

To summarize our discussion so far, we can often transform measures and obtain an analogue of Theorem~\ref{setupclaim} for unbounded distributions (provided the distributions decay sufficiently quickly), and can easily obtain a weak duality result for such unbounded measures, but additional work is required to prove whether strong duality holds. 
\end{document}